\documentclass{article}
\usepackage{fullpage}
\usepackage{amssymb}
\usepackage{amsmath}
\usepackage{amsthm}
\usepackage[colorlinks]{hyperref}
\usepackage{cleveref}
\usepackage{enumitem}
\usepackage{algorithm}
\usepackage{xcolor}
\usepackage{tikz}
\usepackage{authblk}
\usepackage{thmtools}

\counterwithin{algorithm}{section}
\counterwithin{figure}{section}
\counterwithin{table}{section}

\theoremstyle{plain}
\newtheorem{theorem}[algorithm]{Theorem}
\newtheorem{lemma}[algorithm]{Lemma}
\newtheorem{corollary}[algorithm]{Corollary}
\theoremstyle{definition}
\newtheorem{definition}[algorithm]{Definition}
\newtheorem{result}[algorithm]{Result}
\newtheorem{open-question}[algorithm]{Open question}

\crefname{algorithm}{Algorithm}{Algorithms}
\crefname{figure}{Figure}{Figures}
\crefname{table}{Table}{Tables}
\crefname{theorem}{Theorem}{Theorems}
\crefname{lemma}{Lemma}{Lemmas}
\crefname{definition}{Definition}{Definitions}
\crefname{section}{Section}{Sections}
\crefname{subsection}{Subsection}{Subsections}
\crefname{equation}{Equation}{Equations}
\crefname{appendix}{Appendix}{Appendices}

\newcommand{\bra}[1]{\ensuremath{\left\langle#1\right|}}
\newcommand{\ket}[1]{\ensuremath{\left|#1\right\rangle}}
\newcommand{\braket}[2]{\ensuremath{\left\langle#1\middle|#2\right\rangle}}
\newcommand{\norm}[1]{\ensuremath{\left\|#1\right\|}}

\newcommand{\C}{\ensuremath{\mathbb{C}}}
\newcommand{\D}{\ensuremath{\mathcal{D}}}
\renewcommand{\H}{\ensuremath{\mathcal{H}}}
\newcommand{\K}{\ensuremath{\mathcal{K}}}
\newcommand{\N}{\ensuremath{\mathbb{N}}}
\newcommand{\R}{\ensuremath{\mathbb{R}}}
\renewcommand{\S}{\ensuremath{\mathcal{S}}}
\newcommand{\T}{\ensuremath{\mathcal{T}}}
\newcommand{\W}{\ensuremath{\mathcal{W}}}

\newcommand{\ALG}{\ensuremath{\mathsf{ALG}}}
\newcommand{\FS}{\ensuremath{\mathsf{FS}}}

\newcommand{\GT}{\ensuremath{\mathsf{GT}}}
\newcommand{\LG}{\ensuremath{\mathsf{LG}}}
\newcommand{\st}{\ensuremath{\mathsf{st}}}
\newcommand{\WDT}{\ensuremath{\mathsf{WDT}}}

\DeclareMathOperator{\argmin}{argmin}
\DeclareMathOperator{\aux}{aux}
\DeclareMathOperator{\EW}{EW}
\DeclareMathOperator{\img}{img}
\DeclareMathOperator{\OR}{OR}
\DeclareMathOperator{\pSEARCH}{pSEARCH}
\DeclareMathOperator{\polylog}{polylog}
\DeclareMathOperator{\QRAG}{QRAG}
\DeclareMathOperator{\QROG}{QROG}
\DeclareMathOperator{\Span}{Span}
\DeclareMathOperator{\Th}{Th}

\title{Quantum algorithms through graph composition}
\author[1]{Arjan Cornelissen}
\affil[1]{Simons Institute, UC Berkeley, California, USA}
\begin{document}
    \maketitle

    \begin{abstract}
        In this work, we unify several quantum algorithmic frameworks for boolean functions that are based on the quantum adversary bound. First, we show that the $st$-connectivity framework subsumes the (adaptive/extended) learning graph framework, and the weighted-decision-tree framework. Additionally, we show that every randomized algorithm can be turned into an $st$-connectivity problem as well, with the same complexity up to constants. This situates the $st$-connectivity framework in between randomized and quantum algorithms, indicating that it's an intermediate computational model between classical and quantum. We also introduce a generalization of the $st$-connectivity framework, the \textit{graph composition framework}, and show that it subsumes part of the quantum divide and conquer framework, and it is itself subsumed by the multidimensional quantum walk framework.

        Second, we investigate these frameworks' power by investigating the most efficient algorithms they can produce, in terms of the number of queries to the input. We show that the weighted-decision-tree framework's power is polynomially related to deterministic query complexity, showing that the quantum speed-ups that can be obtained with this framework are at most quadratic.

        Finally, we turn our attention to time-efficient implementations of the algorithms constructed through the $st$-connectivity and graph composition frameworks. To that end, we convert instances to the two-subspace phase estimation framework, and we show how we can implement these as transducers. This has the added benefit of removing the effective spectral gap lemma from the construction, significantly simplifying the analysis. We showcase the techniques developed in this work to give improved algorithms for various string search problems.
    \end{abstract}

    \section{Introduction}

    Over the last two decades, we have seen rapid developments in the design of quantum algorithms for solving computational problems. When measuring these algorithms' efficiency, the number of times it accesses the input has become a well-studied metric. If an algorithm is allowed to query a computational problem's input coherently in superposition, we say that the algorithm is a \textit{quantum query algorithm}, and we refer to the minimal number of queries required to solve said problem with high probability as its \textit{quantum query complexity}.

    In a landmark result, Reichardt showed that the quantum query complexity of computing any boolean function $f : \{0,1\}^n \supseteq \D \to \{0,1\}$ is characterized up to constants by a semidefinite program (SDP), called the adversary bound~\cite{reichardt2009span,reichardt2011reflections}. Moreover, he exhibited a constructive way to turn any feasible solution of the minimization version of this semidefinite program for $f$ into a quantum query algorithm that computes it. This opened up a new avenue for designing quantum algorithms, i.e., by constructing feasible solutions to the adversary SDP in a smart way.

    In what followed, several independent lines of research emerged, designing algorithmic frameworks that build on Reichardt's result. We give a brief overview of the most prominent examples here:

    \begin{enumerate}[noitemsep]
        \item \textit{Span programs} \cite{reichardt2009span,reichardt2012span,ito2019approximate,cornelissen2020span}. The span program framework represents quantum algorithms as two subspaces in a Hilbert space, where reflecting through one of the subspaces represents the query to the input, and the other reflection is input-independent. Span programs are complete, i.e., they can generate any solution in the SDP's feasible region.
        \item \textit{The (adaptive/extended) learning graph frameworks}~\cite{belovs2012span-learning-graphs,belovs2012learning,carette2020extended}. These frameworks represent algorithms by rooted directed graphs, where every edge represents a query to the input. As one progresses through the graph, one learns progressively more about the input. The complexity of an instance is determined by the effective resistance between the root of the graph and all the $1$-certificates.
        \item \textit{The bomb-testing/guessing-tree/weighted-decision-tree frameworks}~\cite{lin2016upper,beigi2020quantum,cornelissen2022improved}. These frameworks represent algorithms as decision trees, just like (classical) deterministic query algorithms. In order to improve over the classical setting, weights are associated to each of the edges, and the shape of the tree determines how much of a speed-up can be obtained.
        \item \textit{The $st$-connectivity framework}~\cite{belovs2012span,jeffery2017quantum,jarret2018quantum}. This framework represents algorithms as connectivity problems on undirected graphs, where the availability of every edge is determined by a query to the input. The complexity of the instance is determined by the effective resistance of the graph restricted to the edges that are available (unavailable) whenever the function value is positive (negative).
        \item \textit{Quantum divide-and-conquer}~\cite{childs2022quantum,allcock2023quantum,jeffery2024multidimensional}. This framework\footnote{In this work, we only consider Strategy~1 of \cite{childs2022quantum}.} solves a computational problem by a divide-and-conquer approach, where the conquer step combines the subproblems with a boolean formula. The complexity of an instance is determined by a recurrence relation over the complexities of the subproblems.
        \item \textit{Transducers}~\cite{belovs2023one,belovs2024taming}. Transducers represent quantum algorithms as unitary operations acting on the direct sum of an action and a catalyst space. The action space is where the desired operation occurs, and the unitary acts as identity on a particular catalyst vector in the catalyst space. The norm of this catalyst vector determines the complexity of the instance. Transducers are also complete, in the sense that they can represent any solution in the SDP's feasible region.
    \end{enumerate}

    All of these frameworks have the property that they allow for efficient composition of several of their instances. In particular, composition on the level of these frameworks is more efficient than in the circuit model as it avoids two bottlenecks that are present in the latter. These are:

    \begin{enumerate}[noitemsep]
        \item \textit{Build-up of errors.} When composing bounded-error quantum algorithms in the circuit model, every level of the recursion adds some failure probability. Mitigating these typically requires polylogarithmic multiplicative overhead, for every level of composition. This is especially costly when we want to do recursive compositions (e.g., in divide-and-conquer algorithms). These frameworks all have a way to avoid paying this overhead when composing several of their instances together.
        \item \textit{Amortization.} When composing quantum algorithms at the circuit level, we typically analyze them using the worst-case query complexity for each of the components. This is because when we run the circuit, we cannot know what input the circuit is running on, and hence we don't know if we can terminate the circuit early. These frameworks all provide ways to get around this limitation, when composing several instances together.
    \end{enumerate}

    These composition properties motivate phrasing quantum algorithmic building blocks in terms of these frameworks instead of on the quantum circuit level, as they can subsequently be more efficiently composed into more complicated quantum procedures.

    The striking similarities between these independent lines of research naturally lead to the following question that we investigate in this work:

    \begin{center}
        \bf\emph{``Can these independent lines of research be unified?''}
    \end{center}

    We answer this question affirmatively in this work, in the form of five main results.

    \subsection{Relations between quantum algorithmic frameworks}

    First, we show that the $st$-connectivity framework plays a central role in unifying several of the existing frameworks:

    \begin{result}[Informal version of \Cref{thm:extended-learning-graphs,thm:weighted-decision-trees}]
        \label{res:st-lg-wdt}
        The $st$-connectivity framework subsumes the (adaptive/extended) learning graph frameworks, and the bomb-testing/guessing-tree/weighted-decision-tree frameworks.
    \end{result}

    In order to prove the above result, we prove that the complexity of an algorithm designed through the extended learning graph framework, or the weighted-decision-tree framework, is captured by the effective resistance of a corresponding $st$-connectivity instance. We remark that the above result implies in particular that every (classical) deterministic query algorithm can be turned into an instance of the $st$-connectivity framework, with at most the same complexity.

    This motivates the question whether we can also turn (classical) \textit{randomized} algorithms into instances of the $st$-connectivity framework. Since deciding connectivity in an undirected graph is inherently a deterministic problem, it is not immediately clear how one encodes randomness in the $st$-connectivity graph. Nevertheless, in this work, we arrive at the surprising observation that this is indeed possible:

    \begin{result}[Informal version of \Cref{thm:st-vs-R}]
        \label{res:randomized}
        Any bounded-error randomized query algorithm can be turned into an instance of the $st$-connectivity framework, with the same complexity up to a constant factor.
    \end{result}

    We prove this result by observing that we can solve the gapped majority function (i.e., $f(x) = 1$ if and only if $|x| \geq 2n/3$, and $f(x) = 0$ if and only if $|x| \leq n/3$) in constant complexity using the $st$-connectivity framework~(\Cref{lem:gapped-majority}). Next, recall that every randomized query algorithm is described by a family of decision trees, of which at least a two-third fraction accept (if the function value is $1$), or at most a one-third fraction accept (if the function value is $0$). Thus, we can view the randomized algorithm as evaluating the gapped majority function on a bit string, where each bit is computed by one of the decision trees. We subsequently use the composition properties of the $st$-connectivity framework to conclude that every randomized algorithm can be framed as an $st$-connectivity problem with the appropriate complexity.

    We remark that the above result situates the $st$-connectivity framework in between that of randomized and quantum algorithms. That is, the $st$-connectivity framework can represent a strictly larger set of algorithms than the set of randomized ones, and a strictly smaller set than the quantum ones, indicating that it's an intermediate model of computation between classical and quantum. To the best of our knowledge, this is the only quantum algorithmic framework that possesses this property.

    Next, we introduce a generalized version of the $st$-connectivity framework in this work, where instead of putting single bit queries on each of the edges of the graph, we allow for associating span programs to each of the edges. We refer to the resulting framework as the \textit{graph composition framework}. It is clear that it subsumes the $st$-connectivity framework, and we show two additional relations.

    \begin{result}[Informal version of~\Cref{thm:divide-and-conquer,thm:subspace-graph}]
        \label{res:graph-composition}
        The graph composition framework subsumes the quantum divide-and-conquer framework, and it is in turn subsumed by the subspace-graph framework (introduced in~\cite{jeffery2024multidimensional}).
    \end{result}

    The subspace-graph framework~\cite{jeffery2024multidimensional} is a formalization of the multidimensional quantum walk introduced in \cite{jeffery2025multidimensional}. The quantum algorithmic framework in this line of works is developed independently from the adversary bound, and \Cref{res:graph-composition} establishes how it is related to adversary-bound-based approaches.

    We summarize \Cref{res:st-lg-wdt,res:randomized,res:graph-composition} in \Cref{fig:framework-relations}, together with other known connections between quantum algorithmic frameworks.

    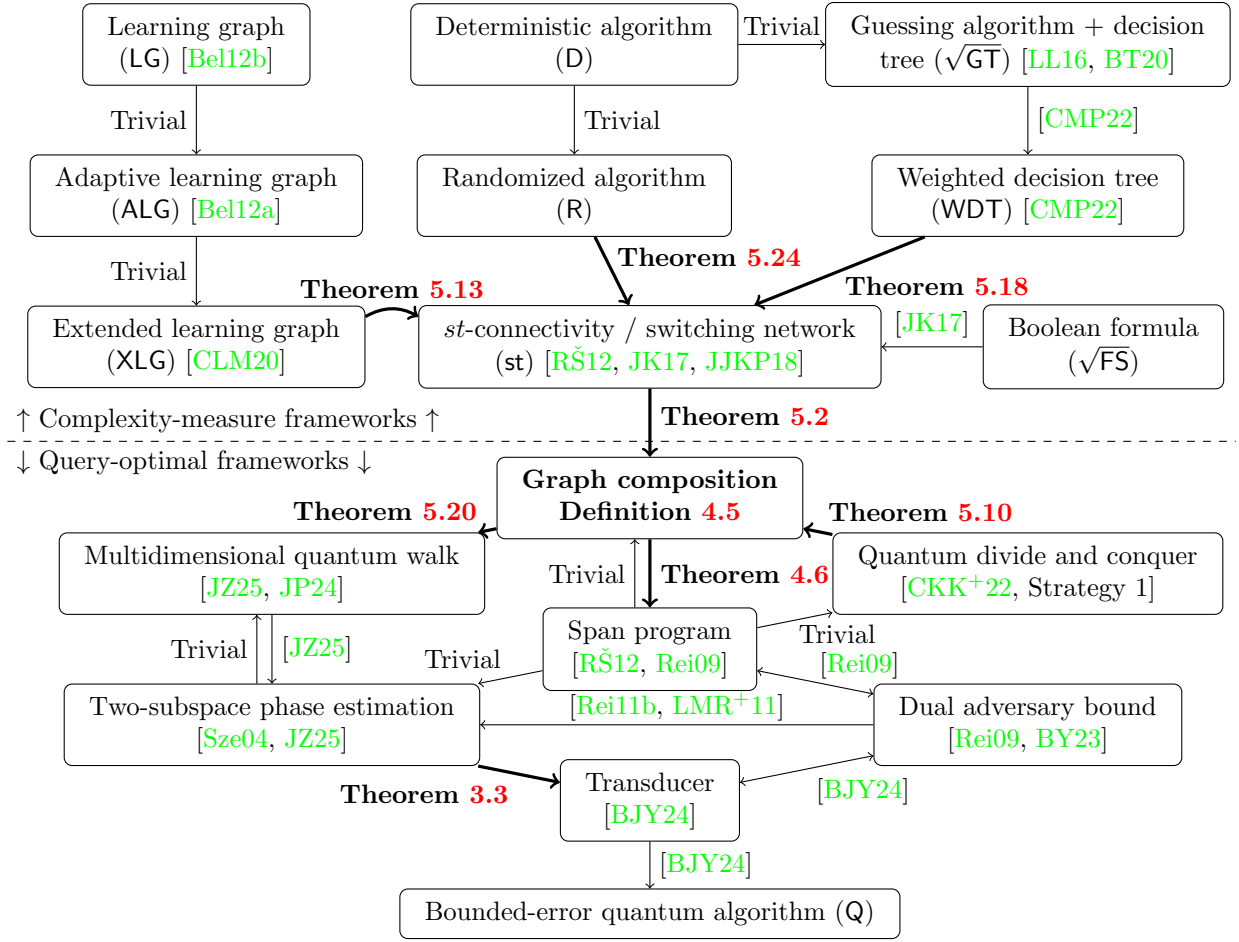
\begin{figure}[!ht]
        \centering
        \begin{tikzpicture}[vertex/.style = {draw, rounded corners = .3em}]
            \node[vertex] (QA) at (0,-1.5) {\begin{tabular}{c}
                Bounded-error quantum algorithm ($\mathsf{Q}$)
            \end{tabular}};
            \node[vertex] (T) at (0,0) {\begin{tabular}{c}
                Transducer \\
                \cite{belovs2024taming}
            \end{tabular}};
            \node[vertex] (ADV) at (5,1) {\begin{tabular}{c}
                Dual adversary bound \\
                \cite{reichardt2009span,belovs2023one}
            \end{tabular}};
            \node[vertex] (PEA) at (-5,1) {\begin{tabular}{c}
                Two-subspace phase estimation \\
                \cite{szegedy2004quantum,jeffery2025multidimensional}
            \end{tabular}};
            \node[vertex] (SP) at (0,2) {\begin{tabular}{c}
                Span program \\
                \cite{reichardt2012span,reichardt2009span}
            \end{tabular}};
            \node[vertex] (MQW) at (-5,3) {\begin{tabular}{c}
                Multidimensional quantum walk \\
                \cite{jeffery2025multidimensional,jeffery2024multidimensional}
            \end{tabular}};
            \node[vertex] (GC) at (0,4) {\bf\begin{tabular}{c}
                Graph composition \\
                \cref{def:graph-composition}
            \end{tabular}};
            \node[vertex] (QDC) at (5,3) {\begin{tabular}{c}
                Quantum divide and conquer \\
                \cite[Strategy~1]{childs2022quantum}
            \end{tabular}};
            \node[vertex] (st) at (0,6) {\begin{tabular}{c}
                $st$-connectivity / switching network \\
                ($\st$) \cite{reichardt2012span,jeffery2017quantum,jarret2018quantum}
            \end{tabular}};
            \node[vertex] (ALG) at (-6,8) {\begin{tabular}{c}
                Adaptive learning graph \\
                ($\ALG$) \cite{belovs2012learning}
            \end{tabular}};
            \node[vertex] (XLG) at (-6,6) {\begin{tabular}{c}
                Extended learning graph \\
                ($\mathsf{XLG}$) \cite{carette2020extended}
            \end{tabular}};
            \node[vertex] (LG) at (-6,10) {\begin{tabular}{c}
                Learning graph \\
                ($\LG$) \cite{belovs2012span-learning-graphs}
            \end{tabular}};
            \node[vertex] (FE) at (6,6) {\begin{tabular}{c}
                Boolean formula \\
                ($\sqrt{\FS}$)
            \end{tabular}};
            \node[vertex] (WDT) at (5,8) {\begin{tabular}{c}
                Weighted decision tree \\
                ($\WDT$) \cite{cornelissen2022improved}
            \end{tabular}};
            \node[vertex] (GT) at (5,10) {\begin{tabular}{c}
                Guessing algorithm + decision \\
                tree ($\sqrt{\GT}$) \cite{lin2016upper,beigi2020quantum}
            \end{tabular}};
            \node[vertex] (D) at (-1,10) {\begin{tabular}{c}
                Deterministic algorithm \\
                ($\mathsf{D}$)
            \end{tabular}};
            \node[vertex] (R) at (-1,8) {\begin{tabular}{c}
                Randomized algorithm \\
                ($\mathsf{R}$)
            \end{tabular}};

            \draw[<->] (ADV) to node[below right] {\cite{belovs2024taming}} (T);
            \draw[<->] (ADV) to node[above right=-.2em] {\cite{reichardt2009span}} (SP);
            \draw[->] (ADV) to node[above=-.2em] {\cite{reichardt2011reflections,lee2011quantum}} (PEA);
            \draw[->] (SP) to node[above left] {Trivial} (PEA);
            \draw[->] (SP) to node[below right=-.2em] {Trivial} (QDC);
            \draw[->] (MQW) to node[right] {\cite{jeffery2025multidimensional}} (PEA);
            \draw[->] ([shift={(-.2,0)}]PEA.north) to node[left] {Trivial} ([shift={(-.2,0)}]MQW.south);
            \draw[->] ([shift={(-.2,0)}]SP.north) to node[left] {Trivial} ([shift={(-.2,0)}]GC.south);
            \draw[->] (FE) to node[above] {\cite{jeffery2017quantum}} (st);
            \draw[->] (LG) to node[left] {Trivial} (ALG);
            \draw[->] (ALG) to node[left] {Trivial} (XLG);
            \draw[->] (GT) to node[right] {\cite{cornelissen2022improved}} (WDT);
            \draw[->] (D) to node[above] {Trivial} (GT);
            \draw[->] (D) to node[right] {Trivial} (R);
            \draw[->] (T) to node[right] {\cite{belovs2024taming}} (QA);

            \draw[very thick,->] (GC) to node[right] {\bf\cref{thm:graph-composition}} (SP);
            \draw[very thick,->] (PEA) to node[below left] {\bf\cref{thm:phase-estimation-algorithm-to-transducer}} (T);
            \draw[very thick,->] (GC) to node[above left] {\bf\cref{thm:subspace-graph}} (MQW);
            \draw[very thick,->] (QDC) to node[above right] {\bf\cref{thm:divide-and-conquer}} (GC);
            \draw[very thick,->] (st) to node[pos = .4, right] {\bf\cref{thm:st-connectivity}} (GC);
            \draw[very thick,->, bend left] ([shift={(0,.4)}]XLG.east) to node[above] {\bf\cref{thm:extended-learning-graphs}} ([shift={(0,.4)}]st.west);
            \draw[very thick,->] (WDT) to node[below right=-.2em] {\bf\cref{thm:st-to-WDT0}} (st);
            \draw[very thick,->] (R) to node[above right=-.2em] {\bf\cref{thm:st-vs-R}} (st);

            \draw[dashed] (-8.5,4.75) node[above right] {$\uparrow$ Complexity-measure frameworks $\uparrow$} node[below right] {$\downarrow$ Query-optimal frameworks $\downarrow$} to (7.75,4.75);
        \end{tikzpicture}
        \caption{Relations between quantum algorithmic frameworks. Framework A points to B, if it is generically possible to turn an instance of framework A into one of B. The results in bold are new in this work. The frameworks below the dashed line are all complete in terms of the quantum query complexity of $\mathsf{Q}$, i.e., one can always devise query-optimal algorithms in all these frameworks. Above the dashed line, it is not (known to be) possible to generically devise query-optimal algorithms, and thus it makes sense to define a complexity measure as the minimal number of queries made by a quantum query algorithm designed through said framework. These complexity measures are denoted in parentheses.}
        \label{fig:framework-relations}
    \end{figure}

    \subsection{Frameworks as complexity measures of boolean functions}

    Next, we investigate the power of the quantum algorithmic frameworks by investigating the optimal algorithm that it can produce, in terms of the number of queries to the input. We observe that there are two types of algorithmic frameworks, namely those that can always produce a query-optimal algorithm (up to constant factors), and those that are inherently limited. Since graph composition can encode any span program in a single edge between $s$ and $t$, it is trivially query-optimal, and therefore all the frameworks that subsume it must also be query-optimal. These are indicated below the dashed line in \Cref{fig:framework-relations}.

    For the frameworks that are not (known to be) query-optimal, we define a complexity measure $\mathsf{M}(f)$ that captures the minimum number of queries it can attain for a particular boolean function $f : \{0,1\}^n \supseteq \D \to \{0,1\}$. For instance, we write $\WDT(f)$ for the minimum number of queries that an algorithm computing $f$ makes, when it is designed through the weighted-decision-tree framework. The symbols for the other complexity measures are included in \Cref{fig:framework-relations} for convenience.

    We can now compare the frameworks by relating their complexity measures to each other, as well as to the more well-known complexity measures like quantum, randomized and deterministic query complexity ($\mathsf{Q}$, $\mathsf{R}$ and $\mathsf{D}$). We observe that an arrow from $A$ to $B$ in \Cref{fig:framework-relations} implies that the complexity measure associated to $A$ is bigger than the one associated to $B$. As such, \Cref{res:randomized} implies that $\st(f) \in O(\mathsf{R}(f))$, for all boolean functions $f : \{0,1\}^n \supseteq \D \to \{0,1\}$.

    We know from existing work that for all \textit{total} boolean functions, i.e., those for which $\D = \{0,1\}^n$, $\mathsf{D}$ and $\mathsf{Q}$ are polynomially related. Indeed, we have $\mathsf{D}(f) \in O(\mathsf{Q}(f)^4)$ for total functions~\cite{aaronson2021degree}, and this is tight~\cite{ambainis2017separations}. From the relations displayed in \Cref{fig:framework-relations}, we infer that the complexity measures $\mathsf{Q}$, $\st$, $\WDT$, $\sqrt{\GT}$, $\mathsf{R}$ and $\mathsf{D}$ are all polynomially related to one another.

    For \textit{partial} boolean functions, i.e., where we merely have $\D \subseteq \{0,1\}^n$, the situation is very different. There, for instance, we know that $\mathsf{Q}$ and $\mathsf{R}$ can be unboundedly separated~\cite{bernstein1997quantum}, i.e., there exists a partial function $f$ for which $\mathsf{Q}(f)$ is constant, and $\mathsf{R}(f)$ is not, and hence we can never meaningfully upper bound $\mathsf{R}$ with $\mathsf{Q}$. A similar unbounded separation is known between $\mathsf{R}$ and $\mathsf{D}$~\cite{deutsch1992rapid}.

    This naturally raises the question how the other complexity measures in \Cref{fig:framework-relations} are related to one another. We showcase all known relations between the complexity measures considered in this work in a Hasse diagram in \Cref{fig:compl-meas}, including the new relations we prove here:

    \begin{figure}[!ht]
        \centering
        \begin{tikzpicture}[
            vertex/.style = {draw, blue, rounded corners = .5em},
            sep/.style = {dashed, red},
            scale = 1.8
            ]

            \clip (-4,-.25) rectangle (4,5.5);

            % Regions
            \draw[dashed, blue, fill=blue!10] plot [smooth cycle] coordinates {(-1,-1) (-.6,3) (4,6) (6,6) (6,-2.5) (-1,-2.5)};
            \node[blue] at (2.75,.25) {$\begin{array}{c}
                    \text{Polynomially related for} \\
                    \text{all total functions}
                \end{array}$};
            \draw[dashed, teal, fill=teal!10] plot [smooth cycle] coordinates {(4,1) (1.5,1) (1.5,4) (4,5) (6,5) (6,-2.5) (6,2)};
            \node[teal] at (2.75,1.25) {$\begin{array}{c}
                    \text{Polynomially related for} \\
                    \text{all partial functions}
                \end{array}$};

            % Complexity measures
            \node[vertex] (Q) at (0,.25) {$\mathsf{Q}$};
            \node[vertex] (st) at (0,1) {$\st$};
            \node[vertex] (R) at (0,3) {$\mathsf{R}$};
            \node[vertex] (D) at (2,4) {$\mathsf{D}$};
            \node[vertex] (WDT) at (2,2) {$\WDT$};
            \node[vertex] (GT) at (2,3) {$\sqrt{\GT}$};
            \node[vertex] (XLG) at (-2.75,1.5) {$\mathsf{XLG}$};
            \node[vertex] (ALG) at (-2.75,2.5) {$\ALG$};
            \node[vertex] (LG) at (-2.75,3.5) {$\LG$};
            \node[vertex] (FS) at (-2,2.5) {$\sqrt{\FS}$};

            % Direct relations
            \draw (Q) to node[left] {$\begin{array}{c}
                    \cite{reichardt2012span} \\
                    \cite{jeffery2017quantum}
                \end{array}$} (st);
            \draw[very thick] (st) to node[below, rotate = 90] {\bf Thm~\ref{thm:st-vs-R}} (R);
            \draw (R) to (D);
            \draw (WDT) to (GT) to (D);
            \draw[very thick] (st) to node[below, pos = .6, rotate = {-atan(1/6)}] {\bf Thm~\ref{thm:extended-learning-graphs}} (XLG);
            \draw (XLG) to (ALG) to (LG);
            \draw (st) to node[above, rotate = {-atan(4/5)}] {\cite{jeffery2017quantum}} (FS);
            \draw[very thick] (st) to node[below, rotate = {atan(2.2/4.2)}] {\bf Thm~\ref{thm:st-to-WDT0}} (WDT);

            % Upper bounds
            \draw (D) to (1.25,4.75) node[left] {$\FS$};
            \draw (D) to (1.5,5) node[above left] {$n$};
            \draw[dashed] (D) to (2,4.6) node[above] {$\mathsf{R}^3$};
            \draw[dashed] (D) to (2.25,4.75) node[above] {$\mathsf{Q}^4$};
            \draw[very thick] (D) to node[above, rotate = {atan(1/2)}] {\bf Thm~\ref{thm:weighted-decision-tree-complexity}} (3,4.5) node[above right] {$\mathsf{WDT}^2$};
            \draw (LG) to (-2.75,4) node[above] {$n$};
            \draw (FS) to node[above, rotate = 90] {\cite{cornelissen2022improved}} (-2,3.5) node[above] {$\sqrt{2^{\mathsf{D}}}$};
            \draw[very thick] (GT) to node[above, rotate = {atan(1/2)}] {\bf Thm~\ref{thm:GT-ub-WDT}} (3,3.5) node[above right] {$\WDT^{3/2}$};

            % Polynomial seperations
            \draw[sep] (Q) to node[above] {$\oplus$} (.5,0) node[right] {$n$};
            \draw[sep] (R) to node[below, rotate = {-atan(1/3)}] {$\begin{array}{c}
                    \cite{sherstov2021optimal} \\
                    \cite{bansal2021k}
                \end{array}$} (.75,2.75) node[right] {$\mathsf{Q}^3$};
            \draw[sep] (R) to node[left] {$\lor$} (-.25,2.5) node[below] {$\mathsf{GT}$};
            \draw[sep] (R) to node[above left] {$\lor$} (-1.3,2.5) node[below] {$\mathsf{LG}^2,\mathsf{FS}$};
            \draw[sep] (D) to node[above, rotate = {atan(1/4)}] {\cite{ambainis2017separations}} (1,3.75) node[left] {$\mathsf{Q}^4$};
            \draw[sep,thick] (GT) to node[above, rotate = {-atan(1/2)}] {\bf Thm~\ref{thm:GT-sep-WDT}} (3,2.5) node[below] {$\WDT^{3/2}$};
            \draw[sep] (FS) to (-2,2) node[below] {$\sqrt{\frac{2^n}{\log(n)}}$};
            \draw[sep] (LG) to node[above, rotate = {atan(7/4)}] {$\mathrm{Threshold}$} node[below, rotate = {atan(7/4)}] {\cite{belovs2014power}} (-3.375,2.4) node[below] {$(\mathsf{GT},R)^{1-o(1)}$};

            % Unbounded separations
            \draw[dotted,red] (WDT) to (1,1.8) node[left] {$\mathsf{R}$};
            \draw[dotted,red] (R) to (-.5,2.7) node[below left] {$\mathsf{Q}$};
        \end{tikzpicture}
        \caption{Hasse diagram of the relationships between complexity measures of boolean functions. When two measures are connected by a solid black line, then the upper complexity measure is bigger than the lower one, for every boolean function $f : \{0,1\}^n \supseteq \D \to \{0,1\}$. On the other hand, if two measures are connected by a dashed black line, then this relation only holds for total boolean functions. If, they are connected by a dashed red line, then there exists a total boolean function $f : \{0,1\}^n \to \{0,1\}$ for which the upper measure is bigger than the lower one. And finally, a dotted red line indicates that there exists an unbounded separation.}
        \label{fig:compl-meas}
    \end{figure}
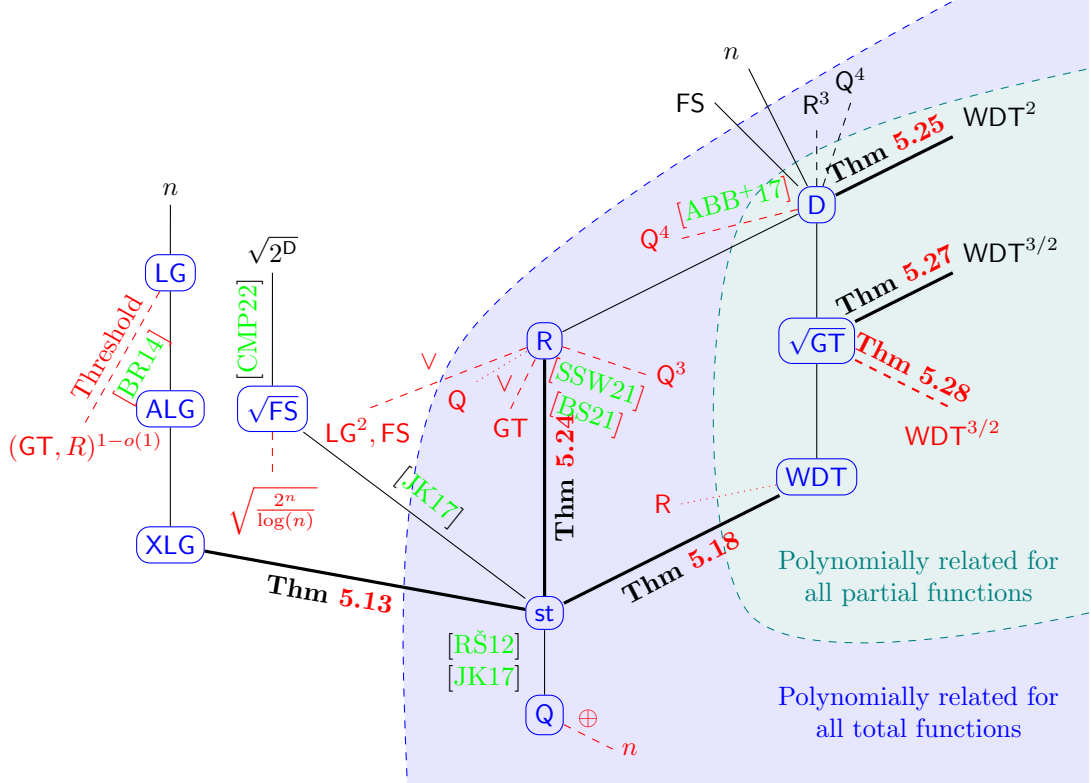

    \begin{result}[Informal version of \Cref{thm:GT-ub-WDT,thm:GT-sep-WDT}]
        \label{res:D-vs-WDT}
        For all boolean functions $f$, we have $\sqrt{\GT}(f) \in O(\WDT(f)^{3/2})$, and $\mathsf{D}(f) \in O(\WDT(f)^2)$. Moreover, both relations are tight, even for total boolean functions.
    \end{result}

    We conclude from \Cref{res:D-vs-WDT} that $\WDT$, $\sqrt{\GT}$ and $\mathsf{D}$ are polynomially related, for all partial boolean functions. This implies in particular that these frameworks can never obtain a more-than-quadratic speed-up over deterministic algorithms, placing a barrier on how much these models of computation can gain over their classical counterpart.

    We now naturally arrive at to the most important open question that arises from this work:

    \begin{open-question}
        Are $\st$ and $\mathsf{R}$ polynomially related for all (partial) boolean functions?
    \end{open-question}

    A positive answer to this question would put a similar barrier on the maximum quantum speed-up that can be obtained using the $st$-connectivity framework and any of the frameworks it subsumes. If, moreover, $\st$ and $\mathsf{R}$ are quadratically related, it could also help explain why quadratic speedups are plentiful, but super-quadratic speedups are few and far between. We leave this question for future work.

    We also remark that as of now, we don't have a proof that $\st$ and $\mathsf{Q}$ are unboundedly separated. A possible partial boolean function for showing such a separation could be a construction based on the Bernstein-Vazirani problem~\cite{bernstein1997quantum} (or perhaps the forrelation problem~\cite{aaronson2015forrelation}). Such a function indeed has constant quantum query complexity, so it remains to prove that its $\st$-connectivity complexity is non-constant. We leave this for future work as well.

    \subsection{Time-efficient implementations}

    Finally, we turn our attention to the time complexity of implementing algorithms designed through the $st$-connectivity and graph composition frameworks. To that end, we observe from \Cref{fig:framework-relations} that we can convert all instances of either framework into an instance of the two-subspace phase estimation framework, as implicitly used in \cite{szegedy2004quantum} for the first time, and formally introduced in \cite{jeffery2025multidimensional}.

    In order to turn an instance of the two-subspace phase estimation framework into a bounded-error quantum algorithm, we first of all observe that we can convert an instance of the two-subspace phase estimation algorithm into a transducer (\Cref{thm:phase-estimation-algorithm-to-transducer}). This is already implicitly remarked in \cite{jeffery2024multidimensional}, but we make the connection explicit here.

    A fortunate consequence of performing the conversion from a two-subspace phase estimation instance into a transducer is that it removes the need for running phase estimation in the resulting quantum algorithm. This also removes the need of the ``effective spectral gap lemma''~\cite[Lemma~4.2]{lee2011quantum}, greatly simplifying the analysis of the algorithm. We comment more on this in \Cref{sec:span-program-algorithm}, where we highlight how this simplifies the span program algorithm in particular.

    To implement the resulting transducer time-efficiently, one has to efficiently implement the reflection through both of the subspaces that make up the instance of the two-subspace phase estimation framework. By tracing back how these subspaces are constructed by instances of the $st$-connectivity framework, we observe that one of the two subspaces is a single query, and the other is completely independent of the input. Thus, the main difficulty in finding a time-efficient implementation of an algorithm constructed from the $st$-connectivity framework lies in implementing a reflection through this input-independent subspace.

    The core observation in this work is that the cost of implementing the reflection through the input-independent subspace is dominated by the cost of implementing a reflection through the \textit{circulation space} of the graph. For any undirected graph $G = (V,E)$, a circulation is an assignment of flows $f \in \C^E$, such that at every node $v \in V$, the net-flow (i.e., the difference between the total incoming and outgoing flow) is $0$. Then, if the graph has resistances $(r_e)_{e \in E} \subseteq \R_{>0}$, the circulation space is defined to be the subspace $\mathcal{C}_{G,r} \subseteq \C^E$, defined as
    \[\mathcal{C}_{G,r} := \Span\left\{\sum_{e \in E} f_e\sqrt{r_e}\ket{e} : f \in \C^E \text{ circulation in } G\right\} \subseteq \C^E.\]

    In order to implement the reflection through this subspace efficiently, we introduce a novel recursive way to decompose the graph, which we refer to as the tree-parallel decomposition. In every decomposition step, we partition the edges in the graph into disjoint subsets, such that if we contract every subset, we are left with either a tree, or a parallel graph with just two nodes. See also \cref{fig:tree-parallel-intro}.

    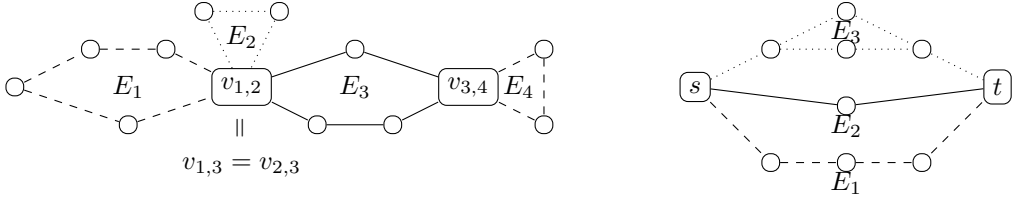
\begin{figure}[!ht]
        \centering
        \begin{tikzpicture}[vertex/.style = {draw, rounded corners = .3em}]
            \begin{scope}
                \node[vertex] (s) at (0,0) {};
                \node[vertex] (1) at (1,.5) {};
                \node[vertex] (3) at (1.5,-.5) {};
                \node[vertex] (2) at (2,.5) {};
                \node[vertex] (v) at (3,0) {$v_{1,2}$};
                \node[below] at (v.south) {$\begin{array}{c}
                        \rotatebox{90}{=} \\
                        v_{1,3} = v_{2,3}
                    \end{array}$};
                \node[vertex] (4) at (4,-.5) {};
                \node[vertex] (6) at (4.5,.5) {};
                \node[vertex] (5) at (5,-.5) {};
                \node[vertex] (t) at (6,0) {$v_{3,4}$};
                \node[vertex] (7) at (2.5,1) {};
                \node[vertex] (8) at (3.5,1) {};
                \node[vertex] (9) at (7,.5) {};
                \node[vertex] (10) at (7,-.5) {};

                \node at (1.5,0) {$E_1$};
                \node at (3,{2/3}) {$E_2$};
                \node at (4.5,0) {$E_3$};
                \node at ({6+2/3},0) {$E_4$};

                \draw[dashed] (s) to (1) to (2) to (v);
                \draw[dashed] (s) to (3) to (v);
                \draw (v) to (4) to (5) to (t);
                \draw (v) to (6) to (t);
                \draw[dotted] (v) to (7) to (8) to (v);
                \draw[dashed] (t) to (9) to (10) to (t);
            \end{scope}
            \begin{scope}[shift = {(9,0)}]
                \node[vertex] (s) at (0,0) {$s$};
                \node[vertex] (1) at (1,-1) {};
                \node[vertex] (2) at (2,-1) {};
                \node[vertex] (3) at (3,-1) {};
                \node[vertex] (4) at (2,-.25) {};
                \node[vertex] (5) at (1,.5) {};
                \node[vertex] (6) at (2,1) {};
                \node[vertex] (7) at (2,.5) {};
                \node[vertex] (8) at (3,.5) {};
                \node[vertex] (t) at (4,0) {$t$};

                \draw[dashed] (s) to (1) to (2) node[below] {$E_1$} to (3) to (t);
                \draw (s) to (4) node[below] {$E_2$} to (t);
                \draw[dotted] (s) to (5) to (6) node[below] {$E_3$} to (8) to (t);
                \draw[dotted] (5) to (7) to (8);
            \end{scope}
        \end{tikzpicture}
        \caption{Examples of the tree decomposition (left), and the parallel decomposition (right). The dashed, dotted and solid sets of edges represent the disjoint edge sets $E_1, \dots, E_k$.}
        \label{fig:tree-parallel-intro}
    \end{figure}

    For a tree decomposition, it is clear that the circulation space decomposes as a direct sum over its components, as any circulation through the original graph must necessarily decompose into circulations over the components. On the other hand, in a parallel decomposition, we can precompute the minimum unit flows $\ket{f_1}, \dots, \ket{f_k}$ between the two nodes through each of the edge sets $E_1, \dots, E_k$. We then observe that any circulation in the original graph can be decomposed as circulation within the respective components, and a vector in $\Span\{\ket{f_1}\} \oplus \cdots \oplus \Span\{\ket{f_k}\}$ that is orthogonal to $\ket{f_1} \oplus \cdots \oplus \ket{f_k}$.

    Next, we observe that if we have access to \textit{quantum read-only memory}~(QROM), then we can implement the state-preparation routines for the unit-flows $\ket{f_1}, \dots, \ket{f_k}$ in time logarithmic in the size of the graph. This subsequently also enables us to implement the reflection through $\ket{f_1} \oplus \cdots \oplus \ket{f_k}$ time-efficiently, and therefore decomposes the problem of reflecting through the circulation space into $k$ smaller components. We also remark here that the quantum read-only memory model (QROM) is a much less stringent assumption than the more common quantum random-access memory model (QRAM). We refer for a more elaborate discussion to~\Cref{subsec:model}.

    We use these decomposition steps recursively, to decompose the edge set $E$ of any undirected graph $G = (V,E)$, until the final decomposed edge sets are merely simple trees. The resulting construction gives rise to a \textit{tree-parallel decomposition tree}~(\cref{def:tree-parallel-decomposition}), where the root node is labeled by $E$, and all the leaf nodes are labeled by the simple trees. We stress that this entire procedure can be precomputed, i.e., it can computed beforehand and hard-coded into the quantum circuit and memory, thus not contributing to the cost of the quantum computation.

    We arrive at the following statement on the time overhead of implementing a reflection through the circulation space.

    \begin{result}[Informal version of \cref{thm:circulation-space-reflection-implementation}]
        \label{res:circulation-space}
        Let $G = (V,E)$ be an undirected graph. Suppose we have a tree-parallel decomposition tree of $G$ with depth $d$, and such that in the $(\ell-1)$th layer, the maximum number of children is $k_{\ell}$. Let $K = (k_1+1) \cdots (k_d+1)$. Then, we can implement the reflection through the circulation space of $G$ in the QROM-model, using a QROM with $\widetilde{O}(|E|K)$ bits, and $\widetilde{O}(d\log(K))$ gates.
    \end{result}

    We remark that even though the time overhead is typically small, the space overhead as stated in the previous theorem typically isn't. However, if the graph that we use in our construction is particularly well-structured, one can typically store the resulting decomposition more efficiently in QROM, drastically reducing the required memory overhead. We work out an example where this is the case in \Cref{subsec:var-time-search-formula-evaluation-divide-conquer}.

    Finally, we remark that for many usecases, the depth $d$ is polylogarithmic in the query complexity, in which case \Cref{res:circulation-space} gives rise to essentially time-optimal algorithms. We showcase several applications for which this is the case in \Cref{tbl:applications}.

    \begin{table}[!ht]
        \centering
        \begin{tabular}{r|ccl}
            \textbf{Problem} & \textbf{Query complexity} & \textbf{Time complexity} & \textbf{Result} \\\hline
            Pattern matching with known pattern & $O(\sqrt{n\log(p)})$ & $\widetilde{O}(\sqrt{n})$ & \Cref{thm:pattern-matching} \\
            {\bf OR $\circ$ pSEARCH} & $O(\sqrt{T\log(T)})$ & $\widetilde{O}(\sqrt{T})$ & \Cref{thm:or-psearch} \\
            $\Sigma^*20^*2\Sigma^*$ & $O(\sqrt{n\log(n)})$ & $\widetilde{O}(\sqrt{n})$ & \Cref{thm:202} \\
            {\bf Dyck-language with depth $3$} & $O(\sqrt{n\log(n)})$ & $\widetilde{O}(\sqrt{n})$ & \Cref{thm:dyck-3} \\
            {\bf$3$-increasing subsequence} & $O(\sqrt{n\log(n)})$ & $\widetilde{O}(\sqrt{n})$ & \Cref{thm:3-is}
        \end{tabular}
        \caption{The obtained results by applying the $st$-connectivity framework to various string problems. The time complexity results hold in the QROM-model~(see \Cref{subsec:model}). In the pattern matching problem, $p$ is the length of the pattern's period. For the problems dislpayed in bold, we obtain slight improvements over the best-known time-efficient implementations.}
        \label{tbl:applications}
    \end{table}

    \subsection{Organization}

    We fix notation, discuss the computational model, and recall relevant existing results in \Cref{sec:preliminaries}. Subsequently, in \Cref{sec:span-program-algorithm}, we present the improved span program algorithm. In \Cref{sec:graph-composition}, we introduce the graph composition framework. In \Cref{sec:relations}, we relate the graph composition framework to existing frameworks. Finally, in \Cref{sec:applications}, we apply the graph composition framework to several concrete computational problems.

    \section{Preliminaries}
    \label{sec:preliminaries}

    \subsection{Notation}

    We start by fixing some notation. $\N = \{1,2,\dots\}$ is the set of natural numbers. We assume throughout the paper that $a/\infty = 0$, for all $a > 0$, and similarly that $a/0 = \infty$, for all $a > 0$.

    Let $d \in \N$ and $f,g : \R_{\geq 0}^d \to \R_{\geq 0}$. We write $f \in O(g)$ if there exist $C,M > 0$ such that for all $x \in \R_{\geq 0}^d$ with $\norm{x} \geq M$, we have $f(x) \leq Cg(x)$. We write $f \in \Omega(g)$ if $g \in O(f)$, and we write $f \in \Theta(g)$, if $f \in O(g) \cap \Omega(g)$. We further write $f \in \widetilde{O}(g)$, exists $k > 0$ such that $f \in O(g \cdot \log^k(g))$. We similarly write $f \in \widetilde{\Omega}(g)$ if $g \in \widetilde{O}(f)$, and we write $f \in \widetilde{\Theta}(g)$ if $f \in \widetilde{O}(f) \cap \widetilde{\Omega}(f)$. Finally, we write $f \in O(g \cdot \polylog(x_j))$ for some $j \in [d]$, if there exists a $k \in \N$ such that $f \in O(g(x) \cdot \log^k(x_j))$.

    \subsection{Quantum algorithms and the computational model}
    \label{subsec:model}

    We give a very brief introduction into quantum algorithms here. For a more elaborate introduction, we refer to more entry-level texts, e.g., \cite{nielsen2010quantum, wolf2019quantum}.

    Quantum algorithms act on a complex finite-dimensional Hilbert space $\H$, referred to as the state space. The algorithm starts in a unit vector in the Hilbert space, referred to as the initial state, and subsequently modifies this state by applying unitary operations to it. At the end of the computation, the algorithm can produce a sample from a finite outcome set $O$, where all the outcomes are associated to mutually orthogonal subspaces of $\H$. The probability of obtaining $o \in O$ is the norm squared of the projection of the final state onto the subspace corresponding to $o$.

    A quantum algorithm can access the some computational problem's input $x \in \D$, where $\D$ is the domain, i.e., the set of allowed inputs, by means of a unitary $O_x$ that encodes this input. We typically refer to this unitary as the oracle. A quantum query algorithm can make queries to this oracle, and all its other unitary operations cannot depend on the input. We say that a quantum query algorithm computes a function $f : \D \to \mathcal{O}$ with high probability, if it outputs $f(x)$ on input $x \in \D$ with probability at least $2/3$. The minimal number of queries that any quantum query algorithm must make to $O_x$ in order to compute $f$ with high probability, is referred to as the quantum query complexity of $f$.

    We can also characterize the cost of implementing the unitary operations that do not depend on the input, known as their \textit{time complexity}. This is somewhat more tricky, because it might depend heavily on the specific architecture on which the algorithm is implemented. We assume that every Hilbert space has a special orthonormal basis, referred to as the computational basis. We assume that every unitary that in the computational basis acts as identity on all but a constant number of dimensions, takes constant time to implement (we refer to these as \textit{elementary operations}). We also assume that implementing a unitary $U$ on a Hilbert space $\H_1$ has the same cost as implementing a unitary $U \otimes I$ on $\H_1 \otimes \H_2$. If we allow for an additional polylogarithimic overhead in the dimension of the state space, we expect to be able to map our approaches to actual implementations, and so we phrase all our time complexity results with $\widetilde{O}$-notation, where the tilde always hides factors that are polylogartihmic in the dimension of the Hilbert space. These assumptions are broadly in line with the \textit{circuit model}, traditionally considered in quantum computation, see e.g.~\cite{nielsen2010quantum}.

    In addition to the above model, we also assume that our quantum algorithms have the option to interface coherently with random-access memory. There are broadly two types of memory we can assume to have access to, read-only memory (QROM) and read-write memory (QRAM). In this paper, we only require a memory register with read-only access. That is, we assume to have access to a QROM of size $N \in \N$, initialized in some immutable precomputed bitstring $x \in \{0,1\}^N$. Then, suppose we have a Hilbert space $\C^N \otimes (\C^2)^{\otimes N}$, where the first and second registers are the index and data registers, respectively. We assume to be able to perform the operation $\QROG$ (\textit{quantum read-only gate}) in unit cost that acts as
    \[\QROG : \ket{j} \otimes \ket{x} \mapsto (-1)^{x_j}\ket{j} \otimes \ket{x}.\]
    Note that this is a fundamentally less-powerful assumption than having access to a quantum random access gate, $\QRAG$, as is more common in the literature. This model appeared in \cite[Section~6.2]{ambainis2007quantum}, and it's also used in e.g.~\cite{buhrman2022memory,akmal2023near,wang2024quantum,belovs2024taming}. The distinction between quantum read-only memory and quantum read-write memory is mentioned in several places in the existing literature, e.g., in \cite{van2019quantum,naya2020optimal,allcock2023quantum}.\footnote{In \cite{van2019quantum}, these are referred to as the QCRAM- and \textit{full}-QRAM models. In \cite{naya2020optimal}, these are referred to as the QACM and QAQM-models. In \cite{allcock2023quantum}, these models are referred to as the QRAM- and QRAG-models.}

    In this model, suppose we denote the minimal cost of implementing a unitary operation $U$ up to constant operator norm error by $\mathsf{T}(U)$. We remark here that with polylogarithmic overhead in both the precision and the dimension of the Hilbert space, it possible to store the description of $U$'s implementation in quantum read-only memory, and then to read this description and apply this operation concurrently. Thus, if $U_1, \dots, U_n$ are unitary operations acting on $\H_1, \dots, \H_n$, we can implement the operation $U = \sum_{j=1}^n \ket{j}\bra{j} \otimes U_j$ on $\H := \H_1 \oplus \cdots \oplus \H_n$ in time
    \begin{equation}
        \label{eq:time-complextiy-composition}
        \mathsf{T}\left(\sum_{j=1}^n \ket{j}\bra{j} \otimes U_j\right) \in \widetilde{O}\left(\max_{j \in [n]} \mathsf{T}(U_j) \cdot \polylog(\dim(\H))\right).
    \end{equation}
    This model mirrors the setting in the classical case, and it is also implicitly used in \cite{akmal2023near}, and explicitly stated in \cite[Section~4.5]{belovs2024taming}.

    \subsection{Quantum subroutines}

    We start by recalling a subroutine for quantum state preparation.

    \begin{theorem}[Quantum state preparation {(see, e.g., \cite[Claim~2.2.3]{prakash2014quantum})}]
        \label{thm:quantum-state-preparation}
        Let $\H$ be a Hilbert space, and let $\ket{\psi} \in \H$ be a state. Let $\ket{\bot}$ be any computational basis state. We can implement an operation $C_{\ket{\bot},\ket{\psi}}$ that implements $\ket{\bot} \mapsto \ket{\psi}$ in time $\widetilde{O}(\log(\dim(\H)))$, using a QROM of size $\widetilde{O}(\dim(\H))$.
    \end{theorem}

    Next, we observe that we can use this subroutine to reflect through arbitrary one-dimensional subspaces, in time polylogarithmic in the dimension of the Hilbert space.

    \begin{theorem}[Reflection through a one-dimensional subspace]
        \label{thm:quantum-state-reflection}
        Let $\H$ be a Hilbert space and $\ket{\psi} \in \H$ be a state. We can reflect through $\Span\{\ket{\psi}\}$ in time $\widetilde{O}(\log(\dim(\H)))$, and using a QROM of size $\widetilde{O}(\dim(\H))$.
    \end{theorem}

    \begin{proof}
        We prepare an ancilla register, with the same state space $\H$, and we prepare $\ket{\psi}$ in that register using \cref{thm:quantum-state-preparation}. Next, we use the quantum SWAP-test, as introduced in \cite[Section~4]{barenco1997stabilization}, to flip the sign of the state in the original register, if it is orthogonal to $\ket{\psi}$. Finally, we uncompute the state-preparation routine in the ancilla register.
    \end{proof}

    Note that these subroutines are very general, in the sense that they can implement and reflect through any quantum state, and they are time-efficient, in the sense that they require time polylogarithmic in the dimension of the Hilbert space. However, they are not very space-efficient, since they require a QROM of size linear in the dimension of the Hilbert space. Thus, if the space complexity is of concern, it can sometimes still be beneficial to circumvent using these results by coming up with a more ad hoc construction.

    One example of such a case is where we want to prepare a uniform superposition. This routine was considered folklore, but was recently written up by \cite{shukla2024efficient}.

    \begin{theorem}[Uniform state preparation~(see, e.g., \cite{shukla2024efficient})]
        \label{thm:uniform-state-preparation}
        Let $n,m \in \N$, let $\ket{\bot}$ be a computational basis state in $\C^n$, and let $\ket{\psi} = \frac{1}{\sqrt{m}}\sum_{j=1}^m \ket{j} \in \C^n$. Then, we can implement an operation that maps $\ket{\bot} \mapsto \ket{\psi}$ in time $\widetilde{O}(\log(n))$.
    \end{theorem}

    \subsection{Span programs}

    Span programs exist in many different formulations in the literature. Here, we mostly follow the presentation from~\cite[Chapter~6]{cornelissen2023quantum}. The primary difference from related works, e.g., \cite{reichardt2009span,belovs2012span,ito2019approximate}, is that here we don't assume that the input is a string from some alphabet. In fact, we allow the inputs to the span program to come from an arbitrary finite set $\D$. Additionally, in contrast to earlier works, we drop the constraint that the initial vector has to be of unit norm. To highlight the distinction, we rename $\ket{w_0}$ to be the \textit{initial vector}, rather than the \textit{initial state}.

    \begin{definition}[Span programs]
        \label{def:span-program}
        A span program consists of the following mathematical objects:
        \begin{enumerate}[nosep]
            \item The \textit{state space}: A Hilbert space $\H$ of finite dimension.
            \item The \textit{domain}: A finite set $\D$, whose elements are referred to as \textit{inputs}.
            \item The \textit{input-dependent subspaces}: to every input $x \in \D$, we associate an input-dependent subspace $\H(x) \subseteq \H$.
            \item The \textit{input-independent subspace}: A subspace $\K \subseteq \H$.
            \item The \textit{initial vector}: $0 \neq \ket{w_0} \in \K^{\perp}$.
        \end{enumerate}
        Then, $\mathcal{P} = (\H, x \mapsto \H(x), \K, \ket{w_0})$ is a span program on $\D$. We make a distinction between positive and negative inputs, as such:
        \begin{enumerate}[nosep]
            \item $x \in \D$ is a \textit{positive input}, if $\ket{w_0} \in \K + \H(x)$.
            \item $x \in \D$ is a \textit{negative input}, if $\ket{w_0} \not\in \K + \H(x)$.
        \end{enumerate}
        Finally, we let $f : \D \to \{0,1\}$ be the function that evaluates to $1$ if and only if the input $x \in \D$ is positive. We say that $\mathcal{P}$ computes $f$.
    \end{definition}

    Next, we define the witness sizes.

    \begin{definition}[Span program witnesses]
        \label{def:witnesses}
        Let $\mathcal{P} = (\H, x \mapsto \H(x), \K, \ket{w_0})$ be a span program on $\D$ that computes $f$. Then,
        \begin{enumerate}[nosep]
            \item If $x \in \D$ is a positive input, then every vector $\ket{w_x} \in \H(x)$ that satisfies $\ket{w_x} - \ket{w_0} \in \K$ is a positive witness for $x$. We write $\W_+(x,\mathcal{P})$ for the set of these vectors. The positive witness complexity for $x$ is the minimal norm-squared of such vectors, and it's denoted by $w_+(x,\mathcal{P})$, which is $\infty$ if such a vector does not exist.
            \item If $x \in \D$ is a negative input, then every vector $\ket{w_x} \in \K^{\perp} \cap \H(x)^{\perp}$ for which $\braket{w_x}{w_0} = 1$ is a negative witness for $x$. We write $\W_-(x,\mathcal{P})$ for the set of these vectors. The negative witness complexity for $x$ is the minimal norm-squared of such vectors, and it's denoted by $w_-(x,\mathcal{P})$, which is $\infty$ if such a vector does not exist.
            \item We define the positive and negative witness complexity, $W_+(\mathcal{P})$ and $W_-(\mathcal{P})$, respectively, as
            \[W_+(\mathcal{P}) = \max_{x \in f^{-1}(1)} w_+(x,\mathcal{P}), \qquad \text{and} \qquad W_-(\mathcal{P}) = \max_{x \in f^{-1}(0)} w_-(x,\mathcal{P}),\]
            and we define the span program complexity, $C(\mathcal{P})$, as
            \[C(\mathcal{P}) = \sqrt{W_+(\mathcal{P}) \cdot W_-(\mathcal{P})}.\]
        \end{enumerate}
    \end{definition}

    For a more intuitive interpretation of these objects, we refer to the more elaborate introduction in~\cite[Chapter~6]{cornelissen2023quantum}.

    \subsection{Elementary span program manipulations}

    We can manipulate span programs in several elementary ways. We start with scalar multiplication, where we multiply the initial state with a constant factor.

    \begin{definition}[Scalar multiplication of span programs]
        Let $\mathcal{P} = (\H, x \mapsto \H(x), \K, \ket{w_0})$ be a span program, and $\alpha > 0$. Then, we write $\alpha\mathcal{P} = (\H, x \mapsto \H(x), \K, \sqrt{\alpha}\ket{w_0})$ as the $\alpha$-scalar multiple of $\mathcal{P}$.
    \end{definition}

    \begin{theorem}[Properties of scalar multiplication of span programs]
        Let $\mathcal{P}$ be a span program on $\D$, and $\alpha > 0$. Then, for all $x \in \D$
        \begin{enumerate}[nosep]
            \item $\W_+(x,\alpha\mathcal{P}) = \W_+(x,\mathcal{P}) \cdot \alpha$, and so $w_+(x,\alpha\mathcal{P}) = w_+(x,\mathcal{P}) \cdot \alpha$.
            \item $\W_-(x,\alpha\mathcal{P}) = \W_-(x,\mathcal{P})/\alpha$, and so $w_-(x,\alpha\mathcal{P}) = w_-(x,\mathcal{P})/\alpha$.
            \item $W_+(\alpha\mathcal{P}) = \alpha W_+(\mathcal{P})$, $W_-(\alpha\mathcal{P}) = W_-(\mathcal{P})/\alpha$, and so $C(\alpha\mathcal{P}) = C(\mathcal{P})$.
            \item If $\mathcal{P}$ computes $f$, then so does $\alpha\mathcal{P}$.
        \end{enumerate}
    \end{theorem}

    \begin{proof}
        For the first claim, observe that
        \begin{align*}
            \ket{w} \in \mathcal{W}_+(x,\alpha\mathcal{P}) &\Leftrightarrow \ket{w} \in \H(x) \land \ket{w} - \sqrt{\alpha}\ket{w_0} \in \K \Leftrightarrow \frac{\ket{w}}{\sqrt{\alpha}} \in \H(x) \land \frac{\ket{w}}{\sqrt{\alpha}} - \ket{w_0} \in \K \\
            &\Leftrightarrow \frac{\ket{w}}{\sqrt{\alpha}} \in \mathcal{W}_+(x,\mathcal{P}),
        \end{align*}
        and similarly,
        \begin{align*}
            \ket{w} \in \mathcal{W}_-(x,\alpha\mathcal{P}) &\Leftrightarrow \ket{w} \in \K^{\perp} \cap \H(x)^{\perp} \land \bra{w} \sqrt{\alpha}\ket{w_0} = 1 \Leftrightarrow \sqrt{\alpha}\ket{w} \in \K^{\perp} \cap \H(x)^{\perp} \land (\sqrt{\alpha}\ket{w})^{\dagger}\ket{w_0} = 1 \\
            &\Leftrightarrow \sqrt{\alpha}\ket{w} \in \mathcal{W}_-(x,\mathcal{P}).
        \end{align*}
        The final claims follow directly from the first two.
    \end{proof}

    We also recall how to perform the negation of a span program. This operation appeared implicitly in several works, see e.g., \cite{reichardt2011reflections,ito2019approximate,jeffery2017quantum}, but we phrase it explicitly here.

    \begin{definition}[Span program negation]
        Let $\mathcal{P} = (\H, x \mapsto \H(x), \K, \ket{w_0})$ be a span program on $\D$. We define:
        \begin{enumerate}[nosep]
            \item For all $x \in \D$, $\H'(x) = \H(x)^{\perp}$,
            \item $\K' = (\K \oplus \Span\{\ket{w_0}\})^{\perp}$,
            \item $\ket{w_0'} = \ket{w_0}/\norm{\ket{w_0}}^2$.
        \end{enumerate}
        Then $\lnot\mathcal{P} = (\H, x \mapsto \H'(x), \K', \ket{w_0'})$ is the negation of $\mathcal{P}$.
    \end{definition}

    \begin{theorem}[Properties of span program negation]
        \label{thm:negation}
        Let $\mathcal{P}$ be a span program on $\D$.
        \begin{enumerate}[nosep]
            \item For all $x \in \D$, $\mathcal{W}_-(x,\lnot\mathcal{P}) = \mathcal{W}_+(x,\mathcal{P})$, and so $w_-(x,\lnot\mathcal{P}) = w_+(x,\mathcal{P})$.
            \item For all $x \in \D$, $\mathcal{W}_+(x,\lnot\mathcal{P}) = \mathcal{W}_-(x,\mathcal{P})$, and so $w_+(x,\lnot\mathcal{P}) = w_-(x,\mathcal{P})$.
            \item $W_-(\lnot\mathcal{P}) = W_+(\mathcal{P})$, $W_+(\lnot\mathcal{P}) = W_-(\mathcal{P})$, and $C(\lnot\mathcal{P}) = C(\mathcal{P})$.
            \item If $\mathcal{P}$ computes $f$, then $\lnot\mathcal{P}$ computes $\lnot f$.
        \end{enumerate}
    \end{theorem}

    \begin{proof}
        For the first claim, it suffices to prove that $\mathcal{W}_-(x,\lnot\mathcal{P}) = \mathcal{W}_+(x,\mathcal{P}) \cdot \norm{\ket{w_0}}^2$. To that end, we find the following chain of equivalences:
        \begin{align*}
            &\ket{w} \in \mathcal{W}_-(x,\lnot\mathcal{P}) \Leftrightarrow \ket{w} \in (\K')^{\perp} \cap \H'(x)^{\perp} \land \braket{w_0'}{w} = 1 \\
            &\quad\Leftrightarrow \ket{w} \in (\K \oplus \Span\{\ket{w_0}\}) \cap \H(x) \land \braket{w_0'}{w} = 1 \\
            &\quad\Leftrightarrow \ket{w} \in \H(x) \land \exists \alpha \in \C : \ket{w} - \alpha\ket{w_0} \in \K \land \braket{w_0}{w} = \norm{\ket{w_0}}^2 \\
            &\quad \Leftrightarrow \ket{w} \in \H(x) \land \ket{w} - \ket{w_0} \in \K \\
            &\quad \Leftrightarrow \ket{w} \in \mathcal{W}_+(x,\mathcal{P}).
        \end{align*}

        For the second claim, observe direction from the definition that $\lnot(\lnot\mathcal{P})$. Consequently, we find that $\mathcal{W}_-(x,\mathcal{P}) = \mathcal{W}_-(x,\lnot(\lnot\mathcal{P})) = \mathcal{W}_+(x,\lnot\mathcal{P})$, completing the proof of the second claim. Finally, the third and fourth claims follow directly from the first two.
    \end{proof}

    Finally, for future reference, we refer to a one-dimensional span program as a \textit{trivial span program}. In such a trivial span program on $\D$, without loss of generality we have $\H = \Span\{\ket{0}\}$, $\K = \{0\}$, $\ket{w_0} = \ket{0}$, and for all $x \in \D$, we have $\H(x) = \{0\}$ or $\H(x) = \H$. Thus, we can canonically embed any function $f : \D \to \{0,1\}$, by choosing
    \[\H(x) = \begin{cases}
        \{0\}, & \text{if } f(x) = 0, \\
        \H, & \text{otherwise}.
    \end{cases}\]
    Then, the trivial span program evaluates $f$, and the complexity is $1$. This construction can be useful if the function $f$ we compute can be implemented in $\Theta(1)$ queries and time. We will see several examples of this in \cref{sec:applications}.

    \section{Improved span program algorithm}
    \label{sec:span-program-algorithm}

    In this section, we show how a span program computing $f$ can be converted into a bounded-error quantum algorithm that computes $f$. We build on ideas from Belovs and Yolcu~\cite[Section~7.4]{belovs2023one}, who convert a dual adversary bound solution into a quantum algorithm. We start by revising the definition of the dual adversary bound.

    \begin{definition}[Dual adversary bound with general oracles~{\cite[Equation~(7.2)]{belovs2023one}}]
        \label{def:dual-adversary-bound}
        Let $f : \D \to \{0,1\}$ be a function, and for all $x \in \D$, let $O_x$ be a unitary operation on $\H$, referred to as the oracle. Then, the adversary bound for $f$ with input oracles $\{O_x\}_{x \in \D}$ is the following optimization program:
        \begin{align*}
            \min\quad & \max_{x \in \D} \norm{\ket{w_x}}^2, \\
            \text{s.t.}\quad & \bra{w_x}((I - O_x^{\dagger}O_y) \otimes I_{\W})\ket{w_y} = \delta_{f(x) \neq f(y)}, & \forall x,y \in \D, \\
            & \ket{w_x} \in \H \otimes \mathcal{W}, & \forall x \in \D, \\
            & \W \text{ Hilbert space}.
        \end{align*}
    \end{definition}

    Now, if we have a span program and we consider the reflection through $\H(x)$ as the oracle belonging to input $x$, then the corresponding span program witnesses form a solution to the above optimization program. We prove this in the following theorem. This is a simplification of \cite[Theorem~3.39]{belovs2014applications}.

    \begin{theorem}
        \label{thm:feasible-sln-adv-bound}
        Let $\mathcal{P} = (\H, x \mapsto \H(x), \K, \ket{w_0})$ be a span program on $\D$, computing the function $f : \D \to \{0,1\}$. For all $x \in \D$, let $\ket{w_x}$ be a span program witness, i.e., if $f(x) = 0$, let $\ket{w_x} \in \W_-(x,\mathcal{P})$, and if $f(x) = 1$, let $\ket{w_x} \in \W_+(x,\mathcal{P})$, cf.\ \cref{def:witnesses}. Then, with $\mathcal{W} = \C$, the set $\{\ket{w_x}/\sqrt{2}\}_{x \in \D}$ forms a feasible solution to the dual adversary bound for $f$ with oracles $\{R_{\H(x)}\}_{x \in \D}$.
    \end{theorem}

    \begin{proof}
        Let $x \in \D$. Observe that if $f(x) = 1$, we have $\ket{w_x} \in \H(x)$, and so $R_{\H(x)}\ket{w_x} = \ket{w_x}$. On the other hand, if $f(x) = 0$, we have $\ket{w_x} \in \H(x)^{\perp}$, and so $R_{\H(x)}\ket{w_x} = -\ket{w_x}$. Thus, we conclude that $R_{\H(x)}\ket{w_x} = -(-1)^{f(x)}\ket{w_x}$, for all $x \in \D$.

        Now, let $x,y \in \D$. We obtain
        \begin{align*}
            \bra{w_x}(I - O_x^{\dagger}O_y)\ket{w_y} &= \braket{w_x}{w_y} - \bra{w_x}R_{\H(x)}R_{\H(y)}\ket{w_y} = \left[1 - (-1)^{f(x) + f(y)}\right]\braket{w_x}{w_y} \\
            &= 2\delta_{f(x) \neq f(y)}\braket{w_x}{w_y}.
        \end{align*}
        If $f(x) = f(y)$, the right-hand evaluates to $0$, so it remains to consider the case where $f(x) \neq f(y)$. Without loss of generality, assume that $f(x) = 0$ and $f(y) = 1$. Then, we find that $\ket{w_x} \in \K^{\perp}$, $\braket{w_x}{w_0} = 1$, and $\ket{w_y} - \ket{w_0} \in \K$. Thus, we find
        \[\braket{w_x}{w_y} = \braket{w_x}{w_0} + \bra{w_x}(\ket{w_y} - \ket{w_0}) = 1 + 0 = 1.\qedhere\]
    \end{proof}

    We observe that span programs are in fact special instances of the more general two-subspace phase estimation algorithms, as formalized by \cite[Definition~3.1]{jeffery2025multidimensional}. Such an object is a $4$-tuple $(\H, x \mapsto \H_A(x), x \mapsto \H_B(x), x \mapsto \ket{\psi_0^x})$, where $\H_A(x),\H_B(x) \subseteq \H$, and $\ket{\psi_0^x} \in \H_B(x)^{\perp}$. It computes a function $f : \D \to \{0,1\}$, where $f(x) = 1$ if and only if $\ket{\psi_0^x} \in \H_A(x) + \H_B(x)$, in which case a vector $\ket{\psi_A^x} \in \H_A(x)$ is a positive witness if $\ket{\psi_0^x} - \ket{\psi_A^x} \in \H_B(x)$. Similarly, a negative witness is a vector $\ket{\psi_{\perp}^x} \in \H_A(x)^{\perp} \cap \H_B(x)^{\perp}$, such that $\braket{\psi_0^x}{\psi_{\perp}^x} = 1$. It is clear that any span program is also a two-subspace phase estimation algorithm, since we can set $\H_A(x) = \H(x)$, $\H_B(x) = \K$, and $\ket{\psi_0^x} = \ket{w_0}$, and then the witnesses are the same for both.

    Next, we prove that every instance of the two-subspace phase estimation algorithm can be turned into a transducer, as follows.

    \begin{theorem}
        \label{thm:phase-estimation-algorithm-to-transducer}
        Let $\mathcal{P} = (\H, x \mapsto \H_A(x), x \mapsto \H_B(x), x \mapsto \ket{\psi_0^x})$ be a two-subspace phase estimation algorithm computing a function $f : \D \to \{0,1\}$. We define
        \[U = -(2\Pi_R - I)(I \oplus (2\Pi_{\H_B(x)} - I))(I \oplus (2\Pi_{\H_A(x)} - I)),\]
        with $R = \Span\{\ket{\psi^x}\}$, and $\ket{\psi^x} := \ket{-} \oplus -\ket{\psi_0^x}$. Then, $U$ is a transducer with witnesses equal to those in the phase estimation algorithm, and transduction action $\ket{-} \overset{U}{\rightsquigarrow} (-1)^{f(x)}\ket{-}$.
    \end{theorem}

    \begin{proof}
        Let $x \in \D$ be a positive instance, and $\ket{\psi_B^x} \in \H_B(x)$ be a corresponding witness vector. We write $\ket{\psi_A^x} = \ket{\psi_0^x} - \ket{\psi_B^x}$, and we observe that
        \begin{align*}
            \begin{bmatrix}
                \ket{-} \\
                \ket{\psi_A^x}
            \end{bmatrix} &\overset{I \oplus (2\Pi_{\H_A(x)} - I)}{\mapsto} \begin{bmatrix}
                \ket{-} \\
                \ket{\psi_A^x}
            \end{bmatrix} = \begin{bmatrix}
                \ket{-} \\
                \ket{\psi_0^x} - \ket{\psi_B^x}
            \end{bmatrix} \overset{I \oplus (2\Pi_{\H_B(x)} - I)}{\mapsto} \begin{bmatrix}
                \ket{-} \\
                -\ket{\psi_0^x} - \ket{\psi_B^x}
            \end{bmatrix} \\
            &\overset{-(2\ket{\psi^x}\!\!\bra{\psi^x} - I)}{\mapsto} \begin{bmatrix}
                -\ket{-} \\
                \ket{\psi_0^x} - \ket{\psi_B^x}
            \end{bmatrix} = \begin{bmatrix}
                -\ket{-} \\
                \ket{\psi_A^x}
            \end{bmatrix}
        \end{align*}
        On the other hand, let $x \in \D$ be a negative instance, and $\ket{\psi_{\perp}^x} \in \H_A(x)^{\perp} \cap \H_B(x)^{\perp}$ be a corresponding witness. Then, we observe that $(\bra{-} \oplus \bra{\psi_{\perp}^x})\ket{\psi^x} = 1 - \braket{\psi_{\perp}^x}{\psi_0^x} = 1 - 1 = 0$. Hence,
        \[\begin{bmatrix}
            \ket{-} \\
            \ket{\psi_{\perp}^x}
        \end{bmatrix} \overset{I \oplus (2\Pi_{\H_A(x)} - I)}{\mapsto} \begin{bmatrix}
            \ket{-} \\
            -\ket{\psi_{\perp}^x}
        \end{bmatrix} \overset{I \oplus (2\Pi_{\H_B(x)} - I)}{\mapsto} \begin{bmatrix}
            \ket{-} \\
            \ket{\psi_{\perp}^x}
        \end{bmatrix} \overset{-(2\ket{\psi^x}\!\!\bra{\psi^x} - I)}{\mapsto} \begin{bmatrix}
            \ket{-} \\
            \ket{\psi_{\perp}^x}
        \end{bmatrix}.\qedhere\]
    \end{proof}

    The previous theorem is an explicit implementation of a transducer. This strongly contrasts with the existing constructions that produce transducers from dual adversary bound solutions, since they merely prove existence and are not constructive~\cite{belovs2023one,belovs2024taming}. One way to interpret the above theorem, thus, is that the objects that make up span programs, or two-subspace phase estimation algorithms, highlight some hidden structure that is not easily recovered directly from the adversary bound solution, and which moreover can be exploited to obtain an explicit implementation of the resulting quantum algorithm.

    We have now provided all the ingredients for an explicit implementation of the span program, namely, one can interpret it as a two-subspace phase estimation algorithm, use \cref{thm:phase-estimation-algorithm-to-transducer} to turn it into a transducer, and then use \cite[Section~7.4]{belovs2023one} to turn it into a quantum algorithm. For ease of reference, we concretely present the resulting algorithm in \cref{alg:span-program-algorithm}.

    \begin{algorithm}[!ht]
        \caption{The span program algorithm}
        \label{alg:span-program-algorithm}
        \textbf{Input:}
        \begin{enumerate}[nosep]
            \item A span program $\mathcal{P} = (\H, x \mapsto \H(x), \K, \ket{w_0})$ on $\D$.
            \item Upper bounds $W_+ > 0$ and $W_- > 0$ on $W_+(\mathcal{P})$ and $W_-(\mathcal{P})$, respectively.
            \item $R_{\H(x)}$: a (controlled) operation that reflects through the subspace $\H(x) \subseteq \H$.
            \item $R_{\K}$: a (controlled) operation that reflects through the subspace $\K \subseteq \H$.
            \item $C_{\ket{w_0}}$: a (controlled, inverse) operation that implements $\ket{\bot} \mapsto \ket{w_0}/\norm{\ket{w_0}}$, for some computational basis state $\ket{\bot} \in \H$.
        \end{enumerate}
        \textbf{Derived parameter:} $K = 18\sqrt{W_+W_-}$.

        \textbf{Output:} $1$ if $x \in \D$ is a positive input for $\mathcal{P}$, $0$ otherwise.

        \textbf{Success probability:} At least $2/3$.

        \textbf{Cost:}
        \begin{enumerate}[nosep]
            \item $O(\sqrt{W_+W_-})$ queries to $R_{\H(x)}$, $R_{\K}$ and $C_{\ket{w_0}}$.
            \item $\widetilde{O}(\sqrt{W_+W_-})$ elementary operations.
            \item $O(\log(W_+W_-) + \log(\dim(\H)))$ qubits.
        \end{enumerate}

        \textbf{Procedure:} $\texttt{SpanProgramAlgorithm}(\mathcal{P}, W_+, W_-, R_{\H(x)}, R_{\K}, C_{\ket{w_0}})$:

        We use the state space $\C^{[K]} \otimes \C^2 \oplus \H$.
        \begin{enumerate}[nosep]
            \item Prepare the state $\ket{\psi_0} := \frac{1}{\sqrt{K}} \sum_{j=1}^k \ket{j}\ket{0} \in \C^{[K]} \otimes \C^2$.
            \item For $j = 1,\dots,K$, perform the following operations:
            \begin{enumerate}[nosep]
                \item Perform $I \oplus R_{\H(x)}$.
                \item Perform $I \oplus R_{\K}$.
                \item Perform $R$, i.e., a reflection through $\Span\{\ket{j}\ket{-} \oplus -\left(W_-/W_+\right)^{1/4}\ket{w_0}\}^{\perp}$.
            \end{enumerate}
            \item Measure according to the projection operator $I_{[K]} \otimes \ket{1}\bra{1}$ and output the result.
        \end{enumerate}
    \end{algorithm}

    \begin{theorem}
        \cref{alg:span-program-algorithm} is a quantum query algorithm in the circuit model that evaluates a span program $\mathcal{P}$ making $O(\sqrt{W_+W_-})$ queries to $R_{\H(x)}$, $R_{\K}$ and $C_{\ket{w_0}}$, with $\widetilde{O}(\sqrt{W_+W_-})$ elementary operations, and using $O(\log(W_+W_-) + \log(\dim(\H)))$ qubits.
    \end{theorem}

    \begin{proof}
        The claims on the costs are readily verified, so it remains to analyze the success probability. Let $f : \D \to \{0,1\}$ be the function that $\mathcal{P}$ computes. We observe from \cref{thm:phase-estimation-algorithm-to-transducer} that Step~2 of \cref{alg:span-program-algorithm} implements the transducer formed from the span program $\sqrt{W_-/W_+}\mathcal{P}$, and so it has transduction action $\ket{-} \rightsquigarrow (-1)^{f(x)}\ket{-}$. We also easily observe that it acts as identity on $\ket{+}$, with no witness vector, and so by linearity it implements the transduction action $\ket{0} \rightsquigarrow \ket{f(x)}$, where the witnesses are $\sqrt{2}$ smaller. Using \cite[Section~7.4]{belovs2023one}, we obtain that the total norm error made by the algorithm is upper bounded by
        \[\sqrt{\max\left\{\max_{x \in f^{-1}(1)} \frac{w_+(x,\sqrt{W_-/W_+}\mathcal{P})}{2 \cdot 18\sqrt{W_+W_-}}, \max_{x \in f^{-1}(0)} \frac{w_+(x,\sqrt{W_-/W_+}\mathcal{P})}{2 \cdot 18\sqrt{W_+W_-}}\right\}} \leq \sqrt{\max\left\{\frac{W_+(\mathcal{P})}{36W_+}, \frac{W_-(\mathcal{P})}{36W_-}\right\}} \leq \frac16,\]
        and so the total success probability is reduced by at most $1/3$.
    \end{proof}

    We end this section by remarking that this algorithm's analysis does not rely on the effective spectral gap lemma. This is somewhat surprising, since all previous constructions for the span program algorithm, or the two-subspace phase estimation algorithm, made use of this lemma. We believe that this simplified analysis can help to adapt this construction to a wider variety of settings, but we leave that for future research.

    \section{The graph composition framework}
    \label{sec:graph-composition}

    In this section, we introduce the graph composition framework. To that end, we first introduce some necessary background on electrical networks, in \cref{subsec:electrical-networks}, before formally stating our main result, i.e., \cref{thm:graph-composition}, in \cref{subsec:graph-composition}.

    \subsection{Electrical networks}
    \label{subsec:electrical-networks}

    We first give an intuitive sketch of the setting that we consider in this setting. The idea is to consider an undirected graph $G = (V,E)$, where each edge $e \in E$ represents an electrical wire with resistance $r_e \in [0,\infty]$. We can think of $r_e = 0$ as $e$ being a shortcut in the circuit, and $r_e = \infty$ as a missing wire, or cut. Next, we consider a source node $s$ and a target node $t$, and send unit current, or flow, into $s$. Then, Kirchhoff's laws predict how the flow distributes over the graph, i.e., how much flow $f_e$ passes through the edge $e \in E$. These laws are derived from the physical principle of ``path of least resistance'', i.e., that the network will find a steady state in which the energy dissipation is minimized.

    We can phrase Kirchoff's laws entirely in linear algebraic terms, by considering this minimization of energy as a shortest vector problem over an affine subspace encoding all current-preserving flows in the network. We formalize this idea in the following definition.

    \begin{definition}[Electrical networks]
        \label{def:electrical-networks}
        Let $G = (V,E)$ be an undirected graph, with \textit{resistances} $r : E \to [0,\infty]$. With every edge $e \in E$, we associate a default direction, i.e., it has a head $e_+$ and a tail $e_-$. We also write $N_+(v), N_-(v) \subseteq E$ to be the set of outgoing and incoming edges to $V$, respectively. Now, we define the following objects:
        \begin{enumerate}[nosep]
            \item The Hilbert space $\H_G = \C^E$ is the \textit{flow space} of $G$, i.e., the set $\{\ket{e} : e \in E\}$ forms an orthonormal basis of the complex Hilbert space $\H_G$.
            \item Every function $f : E \to \C$ that satisfies $f_e = 0$ if $r_e = \infty$ is a \textit{flow} on $G$. To every flow $f$ on $G$, we associate a \textit{flow state} $\ket{f_{G,r}} = \sum_{e \in E} f_e\sqrt{r_e}\ket{e} \in \H_G$, where here and in the following we take $0 \cdot \infty = 0$. Sometimes, we write $\ket{f}$ instead of $\ket{f_{G,r}}$, if the graph and the resistances are clear from context.
            \item If a flow $f$ on $G$ satisfies for all $v \in V$, $\sum_{e \in N_+(v)} f_e - \sum_{e \in N_-(v)} f_e = 0$, then $f$ is a \textit{circulation}. We denote the set of circulations by $C_G$, and the \textit{circulation space} is $\mathcal{C}_{G,r} = \{\ket{f_{G,r}} : f \in C_G\} \subseteq \H_G$.
            \item Let $s,t \in V$ with $s \neq t$. If a flow $f$ on $G$ satisfies
            \[\sum_{e \in N_+(v)} f_e - \sum_{e \in N_-(v)} f_e = \begin{cases}
                1, & \text{if } v = s, \\
                -1, & \text{if } v = t, \\
                0, & \text{otherwise},
            \end{cases}\]
            then $f$ is a \textit{unit $st$-flow} on $G$. We let $ F_{G,s,t}$ be the set of all unit $st$-flows on $G$, and the \textit{unit-$st$-flow space} on $G$ is $\mathcal{F}_{G,s,t,r} = \{\ket{f_{G,r}} : f \in F_{G,s,t}\} \subseteq \H_G$. We say that $s$ and $t$ are \textit{connected} in $G$ if $F_{G,s,t}$ is non-empty.
            \item Let $f$ be a flow on $G$. The norm squared of $\ket{f}$, i.e., $\norm{\ket{f_{G,r}}}^2 = \sum_{e \in E} |f_e|^2r_e$ is the \textit{energy} of $f$. If $s$ and $t$ are connected in $G$, we define the \textit{minimum-energy unit $st$-flow} and the \textit{effective resistance} between $s$ and $t$, respectively, by
            \[f^{\min}_{G,s,t,r} := \underset{f \in F_{G,s,t}}\argmin \norm{\ket{f_{G,r}}}^2, \qquad \text{and} \qquad R_{G,s,t,r} := \norm{\ket{f^{\min}_{G,s,t,r}}}^2.\]
            On the other hand, if $s$ and $t$ are not connected in $G$, then $R_{G,s,t,r} = \infty$.
        \end{enumerate}
    \end{definition}

    We start by proving several immediate properties of this definition.

    \begin{lemma}[Electric network properties]
        \label{lem:electric-network-properties} We have the following properties:
        \begin{enumerate}[nosep]
            \item $\mathcal{C}_{G,r}$ is a linear subspace of $\H_G$.
            \item $\mathcal{F}_{G,s,t,r} = \ket{f^{\min}_{G,s,t,r}} + \mathcal{C}_{G,r}$ and $\ket{f^{\min}_{G,s,t,r}} \in \mathcal{C}_{G,r}^{\perp}$.
        \end{enumerate}
    \end{lemma}

    \begin{proof}
        We start by observing that the embedding of flows into the flow space is linear. Furthermore, it is also clear that the set of circulations is closed under addition and scalar multiplication, from which it follows that $\mathcal{C}_{G,r}$ is a linear subspace of $\H_G$, proving the first claim.

        For the second claim, suppose we have two unit $st$-flows on $G$, $f$ and $f'$. Then, it is easily verified that their difference $f-f' \in C_G$ is a circulation. Thus, the set of all unit-$st$-flow states is an affine subspace of $\H_G$. Moreover, the unique vector with smallest norm in this affine subspace, which is $\ket{f^{\min}_{G,s,t,r}}$ by definition, is indeed orthogonal to $\mathcal{C}_{G,r}$, completing the proof.
    \end{proof}

    The above lemma gives us a direct connection between the circulation space and the unit $st$-flow space. It turns out we can also neatly characterize the orthogonal complement of the circulation space, through potential functions. We define them here.

    \begin{definition}[Potential functions]
        Let $G = (V,E)$ be an undirected graph with resistances $r : E \to [0,\infty]$. A \textit{potential function} on $G$ is a function $U : V \to \C$, such that for all $e \in E$ where $r_e = 0$, we have $U_{e_-} = U_{e_+}$. We define a flow $f_{G,U,r}$ derived from the potential $U$, such that for all $e \in E$,
        \[(f_{G,U,r})_e := \frac{U_{e_-} - U_{e_+}}{r_e},\]
        where we use the convention that $0/0 = 0$. We refer to $\ket{f_{G,U,r}}$ as the \textit{potential state}.
    \end{definition}

    The core observation we make is that a vector in the flow space is a potential state, if and only if it is orthogonal to the circulation subspace. We make this observation precise in the following theorem, and we show that we can also characterize the effective resistance in terms of a minimization over potential functions.

    \begin{theorem}[Properties of potential functions]
        \label{thm:potential-function-properties}
        Let $G = (V,E)$ be an undirected graph with resistances $r : E \to [0,\infty]$. We have the following properties of potential functions:
        \begin{enumerate}[nosep]
            \item A flow state is a potential state if and only if it is orthogonal to the circulation subspace.
            \item For any potential function $U$ on $G$ and any unit $st$-flow $f$ on $G$, we have $\braket{f_{G,r}}{f_{G,U,r}} = U_s - U_t$.
            \item We have
            \[\min_{\substack{U \text{ potential on } G \\ U_s - U_t = 1}} \norm{\ket{f_{G,U,r}}}^2 = R^{-1}_{G,s,t,r}.\]
        \end{enumerate}
    \end{theorem}

    \begin{proof}
        For the first claim, we show the two implications separately. First, let $U$ be a potential function on $G$, and $f \in C_G$. Then,
        \[\braket{f_{G,r}}{f_{G,U,r}} = \sum_{e \in E} \overline{f_e}(f_{G,U,r})_er_e = \sum_{e \in E} (U_{e_-} - U_{e_+})\overline{f_e} = \sum_{v \in V} U_v\left(\sum_{e \in N_+(v)} \overline{f_e} - \sum_{e \in N_-(v)} \overline{f_e}\right) = 0,\]
        and so indeed $\ket{f_{G,U,r}} \in \mathcal{C}_{G,r}^{\perp}$.

        For the other direction, let $f$ be a flow on $G$, and suppose that $\ket{f_{G,r}} \in \mathcal{C}_{G,r}^{\perp}$. For the moment, we remove edges $e \in E$ for which $r_e = \infty$. In every connected component of $G$, we pick a vertex $v$ arbitrarily, and set $U(v) = 0$. Next, for every $w \in V$ in that connected component, we find a path to $v$, and we denote the edges by $e_1, \dots, e_k$, and the corresponding directions by $m_1, \dots m_k$, such that $m_j = -1$ if we traverse $e_j$ in the right direction, and $1$ otherwise. Now, we let
        \[U(w) = \sum_{j=1}^k m_jf_{e_j}r_{e_j},\]
        To check whether this defines a well-defined potential function, let $U'(w)$ defined through a different path $e_1', \dots, e_{k'}'$, with signs $m_1', \dots, m_{k'}'$. Now, we observe that $e_1, \dots, e_k, e_{k'}, \dots, e_1'$ with corresponding signs $m_1, \dots, m_k, -m_{k'}', \dots, -m_1'$ is a cycle, and as such, if we send a unit flow $f'$ along this cycle, we obtain a circulation state
        \[\ket{f'_{G,r}} = \sum_{j=1}^k m_j\sqrt{r_{e_j}} - \sum_{j=1}^{k'} m_j'\sqrt{r_{e_j'}} \in \mathcal{C}_{G,r}.\]
        Since we assumed that $\ket{f_{G,r}} \in \mathcal{C}_{G,r}^{\perp}$, we must have in particular that $\braket{f'_{G,r}}{f_{G,r}} = 0$. As such, we find
        \[U(w) - U'(w) = \sum_{j=1}^k m_jf_{e_j}r_{e_j} - \sum_{j=1}^{k'} m_j'f_{e_j'}r_{e_j'} = \braket{f'_{G,r}}{f_{G,r}} = 0,\]
        and hence $U(w) = U'(w)$, which implies that $U$ is well-defined. Moreover, for any edge $e \in E$ where $r_e = 0$, we observe from the definition that indeed $U_{e_-} = U_{e_+}$, and so $U$ is indeed a potential function. Finally, for all $e \in E$ for which $r_e \neq \infty$, we indeed find that
        \[f_e = \frac{U_{e_-} - U_{e_+}}{r_e},\]
        and for all edges $e \in E$ for which $r_e = \infty$, we have $f_e = 0$ by definition, which means that the above relation also holds in that case. As such, $\ket{f_{G,r}} = \ket{f_{G,U,r}}$, completing the proof of the first claim.

        For the second claim, let $f \in F_{G,s,t}$ be a unit $st$-flow, and $U$ a potential function on $G$. We have
        \[\braket{f_{G,r}}{f_{G,U,r}} = \sum_{e \in E} \overline{f_e}(f_{G,U,r})_er_e = \sum_{e \in E} (U_{e_-} - U_{e_+})\overline{f_e} = \sum_{v \in V} U_v\left(\sum_{e \in N_+(v)} \overline{f_e} - \sum_{e \in N_-(v)} \overline{f_e}\right) = U_s - U_t,\]
        which proves the second claim.

        For the final claim, we observe that the states $\ket{f_{G,U,r}}$ we are minimizing over, are exactly the states in $\mathcal{C}_{G,r}^{\perp}$ that have inner product $1$ with all vectors in $\mathcal{F}_{G,s,t,r}$, by the first and second claim. Thus, we conclude
        \[\min_{\substack{U \text{ potential on } G \\ U_s - U_t = 1}} \norm{\ket{f_{G,U,r}}}^2 = \min_{\substack{\ket{f} \in \mathcal{C}_{G,r}^{\perp} \\ \forall \ket{f'} \in \mathcal{F}_{G,s,t,r}, \braket{f'}{f} = 1}} \norm{\ket{f}}^2 = \min_{\substack{\ket{f} \in \mathcal{C}_{G,r}^{\perp} \\ \braket{f^{\min}_{G,s,t,r}}{f} = 1}} \norm{\ket{f}}^2 = \frac{1}{\norm{\ket{f^{\min}_{G,s,t,r}}}^2} = R_{G,s,t,r}^{-1}.\qedhere\]
    \end{proof}

    \subsection{Span program composition through graphs}
    \label{subsec:graph-composition}

    Inspired by the theory of electrical networks, we show how to compose span programs on a common domain $\D$, by associating them to the edges of an undirected graph $G = (V,E)$ with distinct source and sink nodes $s,t \in V$. The idea is that every input $x \in \D$ will accept on the span programs associated to some of the edges $e \in E$, and we denote the subset of these edges $E(x) \subseteq E$. The resulting composed span program, now, will accept if and only if there is a path from $s$ to $t$ exclusively using these accepting edges $E(x)$.

    We start by formally introducing the construction, in the following definition.

    \begin{definition}[Graph composition]
        \label{def:graph-composition}
        Let $G = (V,E)$ be a connected, undirected graph. Let $\D$ be a finite set, and for every edge $e \in E$, let $\mathcal{P}^e = (\H^e, x \mapsto \H^e(x), \K^e, \ket{w_0^e})$ be a span program on $\D$. Let $s,t \in V$ with $s \neq t$, such that $s$ and $t$ are connected in $G$. For all $e \in E$, we write $r_e = \norm{\ket{w_0^e}}^2$. For all $x \in \D$, we let
        \[\H = \bigoplus_{e \in E} \H^e, \qquad \text{and} \qquad \H(x) = \bigoplus_{e \in E} \H^e(x).\]
        Now, we define $\mathcal{E} : \H_G \to \H$ as a linear isometric embedding that for all $e \in E$ maps $\ket{e} \mapsto \ket{w_0^e}/\norm{\ket{w_0^e}}$. We let
        \[\K = \bigoplus_{e \in E} \K^e \oplus \mathcal{E}(\mathcal{C}_{G,r}), \qquad \text{and} \qquad \ket{w_0} = \mathcal{E}\left(\ket{f^{\min}_{G,s,t,r}}\right).\]
        The span program $\mathcal{P} = (\H, x \mapsto \H(x), \K, \ket{w_0})$ is the graph composition of $G$ with source node $s$, sink node $t$, and span programs $(\mathcal{P}^e)_{e \in E}$.
    \end{definition}

    In order to check the well-definedness of this definition, we must check that $\ket{w_0} \in \K^{\perp}$, cf.\ \cref{def:span-program}. Indeed, since for every $e \in E$, we have $\ket{w_0^e} \in (\K^e)^{\perp}$, we find that $\img(\mathcal{E}) \subseteq \oplus_{e \in E} (\K^e)^{\perp}$, and so it remains to check that $\ket{w_0} \in \mathcal{E}(\mathcal{C}_{G,r})^{\perp}$. To that end, observe that $\mathcal{E}$ is an isometric embedding, and so it is sufficient to check whether $\ket{f^{\min}_{G,s,t,r}} \in \mathcal{C}_{G,r}^{\perp}$, which we know to be true from \cref{lem:electric-network-properties}.

    Next, we compute the witness sizes of graph-composed span programs.

    \begin{theorem}[Graph composition witness sizes]
        \label{thm:graph-composition}
        Let $G = (V,E)$ be an undirected graph with $s,t \in V$, $s \neq t$, such that $s$ and $t$ are connected in $G$. Let $\mathcal{P}$ be the graph composition of $G$ with source node $s$, sink node $t$, and span programs $(\mathcal{P}^e)_{e \in E}$. Let $x \in \D$. Then,
        \begin{enumerate}[nosep]
            \item $w_+(x,\mathcal{P}) = R_{G,s,t,r^+}$, where for all $e \in E$, $r^+_e = w_+(x,\mathcal{P}^e)$.
            \item $w_-(x,\mathcal{P}) = R_{G,s,t,r^-}^{-1}$, where for all $e \in E$, $r^-_e = w_-(x,\mathcal{P}^e)^{-1}$.
        \end{enumerate}
    \end{theorem}

    \begin{proof}
        We start with the first claim. Let $\ket{w} \in \H$. We write $\ket{\overline{w}_e} = \Pi_{\H^e}\ket{w}$. Moreover, we let $f_e := \braket{w_0^e}{\overline{w}_e}/\norm{\ket{w_0^e}}^2$. If $f_e \neq 0$, we also write $\ket{w^e} := \ket{\overline{w}_e}/f_e$. Now, we observe the following sequence of equivalences:
        \begin{align*}
            &\ket{w} \in \mathcal{W}_+(x,\mathcal{P}) \Leftrightarrow \ket{w} \in \H(x) \land \ket{w} - \ket{w_0} \in \K \\
            &\quad \Leftrightarrow \left[\forall e \in E, \ket{\overline{w}_e} \in \H^e(x) \land \Pi_{\Span\{\ket{w_0^e}\}^{\perp}}\ket{\overline{w}_e} \in \K^e\right] \land \sum_{e \in E} \Pi_{\Span\{\ket{w_0^e}\}}\ket{\overline{w}_e} \in \ket{w_0} + \mathcal{E}(\mathcal{C}_{G,r}) \\
            &\quad \Leftrightarrow \left[\forall e \in E, \ket{\overline{w}_e} \in \H^e(x) \land \ket{\overline{w}_e} - \frac{\braket{w_0^e}{\overline{w}_e}}{\norm{\ket{w_0^e}}^2}\ket{w_0^e} \in \K^e\right] \land \sum_{e \in E} \frac{\braket{w_0^e}{\overline{w}_e}}{\norm{\ket{w_0^e}}^2}\ket{w_0^e} \in \mathcal{E}(\ket{f^{\min}_{G,s,t,r}}) + \mathcal{E}(\mathcal{C}_{G,r}) \\
            &\quad \Leftrightarrow \left[\forall e \in E, \ket{\overline{w}_e} \in \H^e(x) \land \ket{\overline{w}_e} - f_e\ket{w_0^e} \in \K^e\right] \land \sum_{e \in E} f_e\ket{w_0^e} \in \mathcal{E}(\ket{f^{\min}_{G,s,t,r}} + \mathcal{C}_{G,r}) \\
            &\quad \Leftrightarrow \left[\forall e \in E, \begin{cases}
                \ket{w_e} \in \H^e(x) \land \ket{w_e} - \ket{w_0^e} \in \K^e, & \text{if } f_e \neq 0, \\
                \ket{\overline{w}_e} \in \H^e(x) \cap \K^e, & \text{otherwise}
            \end{cases}\right] \land \sum_{e \in E} f_e\norm{\ket{w_0^e}}\ket{e} \in \ket{f^{\min}_{G,s,t,r}} + \mathcal{C}_{G,r} \\
            &\quad \Leftrightarrow \left[\forall e \in E, \begin{cases}
                \ket{w_e} \in \mathcal{W}_+(x,\mathcal{P}^e), & \text{if } f_e \neq 0, \\
                \ket{\overline{w}_e} \in \H^e(x) \cap \K^e, & \text{otherwise}
            \end{cases}\right] \land f \in F_{G,s,t},
        \end{align*}
        where in the last step we used the second claim of \cref{lem:electric-network-properties}. Next, by computing the norm of these vectors $\ket{w} \in \mathcal{W}_+(x,\mathcal{P})$, we can compute the positive witness size of $x$. We find
        \begin{align*}
            w_+(x,\mathcal{P}) &= \min_{\ket{w} \in \mathcal{W}_+(x,\mathcal{P})} \norm{\ket{w}}^2 = \min_{\ket{w} \in \mathcal{W}_+(x,\mathcal{P})} \left[\sum_{\substack{e \in E \\ f_e \neq 0}} |f_e|^2\norm{\ket{w_e}}^2 + \sum_{\substack{e \in E \\ f_e = 0}} \norm{\ket{\overline{w}_e}}^2\right] \\
            &= \min_{f \in F_{G,s,t}} \left[\sum_{\substack{e \in E \\ f_e \neq 0}} |f_e|^2w_+(x,\mathcal{P}^e)\right] = \min_{f \in F_{G,s,t}} \norm{\ket{f_{G,r^+}}}^2 = R_{G,s,t,r^+},
        \end{align*}
        where we observed that the optimal choices for all the $\ket{\overline{w}_e}$'s are simply $0$, and the optimal choices for all $\ket{w_e}$'s are the minimal positive witnesses for $x$ in $\mathcal{P}^e$, which have norm squared $w_+(x,\mathcal{P}^e)$. This completes the proof of the first claim.

        For the second claim, again let $\ket{w} \in \H$ and let $\ket{\overline{w}_e} = \Pi_{\H^e}\ket{w}$. Let $f_e = \braket{w_0^e}{\overline{w}_e}/\norm{\ket{w_0^e}}^2$, and if $f_e \neq 0$, let $\ket{w_e} = \ket{\overline{w}_e}/(f_e\norm{\ket{w_0^e}}^2)$. Now,
        \begin{align*}
            &\ket{w} \in \mathcal{W}_-(x,\mathcal{P}) \Leftrightarrow \ket{w} \in \H(x)^{\perp} \cap \K^{\perp} \land \braket{w}{w_0} = 1 \\
            &\quad \Leftrightarrow \left[\forall e \in E, \ket{\overline{w}_e} \in \H^e(x)^{\perp} \cap (\K^e)^{\perp}\right] \land \sum_{e \in E} \frac{\braket{w_0^e}{\overline{w}_e}}{\norm{\ket{w_0^e}}^2}\ket{w_0^e} \in \mathcal{E}(\mathcal{C}_{G,r})^{\perp} \land \sum_{e \in E} (f^{\min}_{G,s,t,r})_e\sqrt{r_e}\frac{\braket{w_0^e}{\overline{w}_e}}{\norm{\ket{w_0^e}}} = 1 \\
            &\quad \Leftrightarrow \left[\forall e \in E, \ket{\overline{w}_e} \in \H^e(x)^{\perp} \cap (\K^e)^{\perp}\right] \land \sum_{e \in E} f_e \ket{w_0^e} \in \mathcal{E}(\mathcal{C}_{G,r})^{\perp} \land \sum_{e \in E} (f^{\min}_{G,s,t,r})_e f_er_e = 1 \\
            &\quad \Leftrightarrow \left[\forall e \in E, \begin{cases}
                \ket{w_e} \in \H^e(x)^{\perp} \cap (\K^e)^{\perp} \land \braket{w_e}{w_0^e} = 1, & \text{if } f_e \neq 0, \\
                \ket{\overline{w}_e} \in \H^e(x)^{\perp} \cap (\K^e)^{\perp} \cap \Span\{\ket{w_0^e}\}^{\perp}, & \text{otherwise}
            \end{cases}\right] \\
            &\qquad \land \sum_{e \in E} f_e\sqrt{r_e}\ket{e} \in \mathcal{C}_{G,r}^{\perp} \land \braket{f^{\min}_{G,s,t,r}}{f_{G,r}} = 1 \\
            &\quad \Leftrightarrow \left[\forall e \in E, \begin{cases}
                \ket{w_e} \in \mathcal{W}_-(x,\mathcal{P}^e), & \text{if } f_e \neq 0, \\
                \ket{\overline{w}_e} \in \H^e(x)^{\perp} \cap (\K^e)^{\perp} \cap \Span\{\ket{w_0^e}\}^{\perp}, & \text{otherwise}
            \end{cases}\right] \land \ket{f_{G,r}} \in \mathcal{C}_{G,r}^{\perp} \land \braket{f^{\min}_{G,s,t,r}}{f_{G,r}} = 1 \\
            &\quad \Leftrightarrow \left[\forall e \in E, \begin{cases}
                \ket{w_e} \in \mathcal{W}_-(x,\mathcal{P}^e), & \text{if } f_e \neq 0, \\
                \ket{\overline{w}_e} \in \H^e(x)^{\perp} \cap (\K^e)^{\perp} \cap \Span\{\ket{w_0^e}\}^{\perp}, & \text{otherwise}
            \end{cases}\right] \land \left[\begin{array}{l}
                \exists U \text{ potential on } G : \\
                \quad f = f_{G,U,r} \land U_s - U_t = 1
            \end{array}\right],
        \end{align*}
        where in the last step we used claims 1 and 2 from \cref{thm:potential-function-properties}.

        By computing the norms of the resulting vectors $\ket{w} \in \mathcal{W}_-(x,\mathcal{P})$, we can compute the negative witness size for $x$. We find
        \begin{align*}
            w_-(x,\mathcal{P}) &= \min_{\ket{w} \in \mathcal{W}_-(x,\mathcal{P})} \norm{\ket{w}}^2 = \min_{\ket{w} \in \mathcal{W}_-(x,\mathcal{P})} \left[\sum_{\substack{e \in E \\ f_e \neq 0}} |f_e|^2\norm{\ket{w_0^e}}^4\norm{\ket{w_e}}^2 + \sum_{\substack{e \in E \\ f_e = 0}} \norm{\ket{\overline{w}_e}}^2\right] \\
            &= \min_{\substack{U \text{ potential on } G \\ U_s - U_t = 1}} \sum_{\substack{e \in E \\ f_e \neq 0}} |(f_{U,r})_e|^2\norm{\ket{w_0^e}}^4w_-(x,\mathcal{P}^e) \\
            &= \min_{\substack{U \text{ potential on } G \\ U_s - U_t = 1}} \sum_{\substack{e \in E \\ f_e \neq 0}} \left|(f_{U,r^-})_e \cdot \frac{r^-_e}{r_e}\right|^2\norm{\ket{w_0^e}}^4w_-(x,\mathcal{P}^e) \\
            &= \min_{\substack{U \text{ potential on } G \\ U_s - U_t = 1}} \sum_{\substack{e \in E \\ f_e \neq 0}} |(f_{U,r^-})_e|^2 r^-_e = \min_{\substack{U \text{ potential on } G \\ U_s - U_t = 1}} \norm{\ket{f_{G,U,r^-}}}^2 = R_{G,s,t,r^-}^{-1},
        \end{align*}
        where we used that the optimal choice for all $\ket{\overline{w}_e}$'s is $0$, and for all $\ket{w_e}$'s the optimal choice is the minimal negative witness for $x$ in $\mathcal{P}^e$, with norm squared $w_-(x,\mathcal{P}^e)$. Finally, in the last inequality, we used the third claim from \cref{thm:potential-function-properties}.
    \end{proof}

    The characterization of the witness sizes in \cref{thm:graph-composition} can be usefully upper bounded by ``paths'' and ``cuts'' in the graph $G$, which can usually be more tractable than computing effective resistances.

    \begin{theorem}
        \label{thm:witness-upper-bounds}
        Let $G = (V,E)$ be an undirected graph, and let $s,t \in V$ with $s \neq t$, such that $s$ and $t$ are connected in $G$. Let $\mathcal{P}$ be the graph composition of $G$ with source node $s$, sink node $t$, and span programs $(\mathcal{P}_e)_{e \in E}$ on a common domain $\D$.
        \begin{enumerate}[nosep]
            \item Suppose that $x \in \D$ is a positive instance for $\mathcal{P}$. Let $P \subseteq E$ be an $st$-path in $G$ such that for all $e \in E$, $x$ is a positive instance for $\mathcal{P}_e$. Then, $w_+(x,\mathcal{P}) \leq \sum_{e \in P} w_+(x,\mathcal{P}_e)$.
            \item Suppose that $x \in \D$ is a negative instance for $\mathcal{P}$. Let $C \subseteq E$ be an $st$-cut in $G$, i.e., a set of edges such that every path from $s$ to $t$ has to intersect $C$ at least once, such that for all $e \in C$, $x$ is a negative instance for $\mathcal{P}_e$. Then, $w_-(x,\mathcal{P}) \leq \sum_{e \in C} w_-(x,\mathcal{P}_e)$.
        \end{enumerate}
    \end{theorem}

    \begin{proof}
        For the first claim, we send unit flow along the path $P$, i.e., we let $f : E \to \R_{\geq 0}$ satisfy $f_e = 1$ iff $e \in P$, and $f_e = 0$ otherwise. Then, $f \in F_{G,s,t}$, and so using \cref{thm:graph-composition} and \cref{def:electrical-networks}, we find that
        \[w_+(x,\mathcal{P}) = R_{G,s,t,r^+} \leq \norm{\ket{f_{G,r}}}^2 = \sum_{e \in P} r_e^+ = \sum_{e \in P} w_+(x,\mathcal{P}_e).\]

        For the second claim, we define the potential function $U : V \to \C$, by for all $v \in V$, setting $U_v = 1$ if we can reach $v$ from $s$ without crossing $C$, and $U_v = 0$ otherwise. Now, $U_s = 1$ and $U_t = 0$, and so from \cref{thm:graph-composition} and \cref{thm:potential-function-properties}, we observe that
        \[w_-(x,\mathcal{P}) = R_{G,s,t,r^-}^{-1} \leq \norm{\ket{f_{G,U,r^-}}}^2 \leq \sum_{e \in C} \frac{1}{r_e^-} = \sum_{e \in C} w_-(x,\mathcal{P}_e).\qedhere\]
    \end{proof}

    \subsection{Time-efficient implementation}
    \label{subsec:time-efficiency}

    To turn a graph composition into a quantum algorithm, we can run the span program algorithm that we presented in \cref{alg:span-program-algorithm}. The witness size analysis in \cref{thm:graph-composition} already analyzes the number of queries that this algorithm needs to make to the input routines, i.e., the reflections through $\H(x)$, $\K$, and a routine that constructs $\ket{w_0} / \norm{\ket{w_0}}$. In order to analyze the time complexity of the resulting algorithm, it remains to provide implementations for these subroutines.

    We start with a general statement about the implementation of these routines in the circuit model, in the following theorem.

    \begin{theorem}
        \label{thm:implementation-graph-composition}
        Let $\mathcal{P} = (\H, x \mapsto \H(x), \K, \ket{w_0})$ on $\D$ be a graph composition of $G = (V,E)$, with source node $s$, sink node $t$, and edge span programs $(\mathcal{P}_e)_{e \in E}$, where for all $e \in E$, $\mathcal{P}_e = (\H^e, x \mapsto \H^e(x), \K^e, \ket{w_0^e})$. Suppose we have implementations of the following operations:
        \[\overline{R}_{\H(x)} : \ket{e}\ket{\psi} \mapsto \ket{e}(2\Pi_{\H^e(x)} - I)\ket{\psi}, \qquad \overline{R}_\K : \ket{e}\ket{\psi} \mapsto \ket{e}(2\Pi_{\K^e} - I)\ket{\psi},\]
        and
        \[\overline{C}_{\ket{w_0}} : \ket{e}\ket{\bot} \mapsto \ket{e}\frac{\ket{w_0^e}}{\norm{\ket{w_0^e}}}, \qquad R_{\mathcal{C}_{G,r}} = 2\Pi_{\mathcal{C}_{G,r}} - I, \qquad \text{and} \qquad C_{\ket{f_{G,s,t,r}^{\min}}} : \ket{\bot} \mapsto \frac{\ket{f_{G,s,t,r}^{\min}}}{\norm{\ket{f_{G,s,t,r}^{\min}}}}.\]
        Then, we have
        \[R_{\H(x)} = \overline{R}_{\H(x)}, \qquad R_{\K} = -\overline{R}_{\K}R_{\mathcal{C}_{G,r}}, \qquad \text{and} \qquad C_{\ket{w_0}} = \overline{C}_{\ket{w_0}}(C_{\ket{f^{\min}_{G,s,t,r}}} \otimes I),\]
        and as such the costs of implementing the routines $R_{\H(x)}$, $R_{\K}$ and $C_{\ket{w_0}}$ in the circuit model are
        \begin{center}
            \begin{tabular}{r|cccccc}
                & $\overline{R}_{\H(x)}$ & $\overline{R}_{\K}$ & $\overline{C}_{\ket{w_0}}$ & $R_{\mathcal{C}_{G,r}}$ & $C_{\ket{f_{G,s,t,r}^{\min}}}$ & Additional gates \\\hline
                $R_{\H(x)}$ & $1$ & $0$ & $0$ & $0$ & $0$ & $0$ \\
                $R_{\K}$ & $0$ & $1$ & $0$ & $1$ & $0$ & $O(1)$ \\
                $C_{\ket{w_0}}$ & $0$ & $0$ & $1$ & $0$ & $1$ & $0$
            \end{tabular}
        \end{center}
    \end{theorem}

    \begin{proof}
        First, note that $\overline{C}_{\ket{w_0}}$ is an embedding of all the spaces $\H^e$ into $\H$, so it plays the same role as $\mathcal{E}$ in \cref{thm:graph-composition}. Thus, we find that
        \[R_{\H(x)} = \bigoplus_{e \in E} R_{\H^e(x)}, \qquad R_\K = -\left[\bigoplus_{e \in E} R_{\K^e}\right] \cdot \mathcal{E}R_{\mathcal{C}_{G,r}}\mathcal{E}^{\dagger}, \qquad \text{and} \qquad C_{\ket{w_0}} = \mathcal{E}C_{\ket{f^{\min}_{G,s,t,r}}}.\]
        The results follow by writing out the actions of the consecutive operations.
    \end{proof}

    We observe that the routines $\overline{R}_{\H(x)}$, $\overline{R}_{\K}$ and $\overline{C}_{\ket{w_0}}$ perform several operations in parallel. This is particularly interesting in the QROM-model, since we know from \cref{eq:time-complextiy-composition} that we can implement these parallel operations time-efficiently, as long as we have access to their descriptions in QROM. Furthermore, we recall from \cref{thm:quantum-state-preparation} that preparing a state can be done time-efficiently as well in the QROM-model, which guarantees the existence of a time-efficient implementation of $C_{\ket{f_{G,s,t,r}^{\min}}}$, i.e., in time polylogarithmic in the number of edges in the graph.

    The reflection through the circulation space $R_{\mathcal{C}_{G,r}}$, though, can potentially be very costly to implement. There exists a generic approach to implementing this reflection, through what we will refer to as the ``spectral method''. This idea was first featured in \cite[Lemma~32]{jeffery2017quantum}, and we restate it here for convenience.\footnote{Note that the proof of \cite[Corollary~26 in the arXiv version]{jarret2018quantum} seems to give an $O(1)$-time implementation of this approach. This, however, was based on an earlier faulty version of \cite[Lemma~32]{jeffery2017quantum}, which has subsequently been fixed.}

    \begin{lemma}[{\cite[Lemma~32]{jeffery2017quantum}}]
        \label{thm:spectral-method}
        Let $G = (V,E)$ be an undirected graph with resistances $r : E \to \R_{>0}$. For all $v \in V$, let $N(v) \subseteq E$ be the set of edges incident to $v$. Let
        \[U_{G,r} : \ket{v}\ket{\bot} \mapsto \ket{v} \left[\sum_{e \in N(v)} \frac{1}{r_e}\right]^{-\frac12}\sum_{e \in N(v)} \frac{1}{\sqrt{r_e}}\ket{e}.\]
        Then, we can implement the reflection through $\mathcal{C}_{G,r}$ with $\widetilde{O}(1/\sqrt{\delta})$ call to $U_{G,r}$ and $\widetilde{O}(1/\sqrt{\delta})$ additional gates, where $\delta$ is the smallest non-zero eigenvalue of the symmetrically normalized Laplacian of $G$, which is defined as $L \in \R^{V \times V}$, with
        \[L[v,w] = \frac{\sum_{e \in N(v) \cap N(w)} \frac{1}{r_e}}{\sqrt{\sum_{e \in N(v)} \frac{1}{r_e} \cdot \sum_{e \in N(w)} \frac{1}{r_e}}}.\]
    \end{lemma}

    We specifically note that the operation $U_{G,r}$ is essentially $|V|$ state-preparation unitaries in parallel, and so using \cref{thm:quantum-state-preparation}, it can be implemented efficiently in the QROM-model too. Thus, we now have a generic way to implement all the operations from \cref{thm:implementation-graph-composition} in the QROM-model.

    There are, however, cases where the aforementioned spectral method is severely suboptimal. For instance, if we take $G$ to be a line graph with $n$ vertices, then the spectral gap $\delta$ of the symmetrically normalized Laplacian is $\Theta(1/n)$, and so we obtain an overhead of $\widetilde{O}(\sqrt{n})$ in \cref{thm:spectral-method}. The circulation space $\mathcal{C}_{G,r}$ of the line graph, however, is empty, and so implementing a reflection through it is trivial, suggesting that the $\widetilde{O}(1/\sqrt{\delta})$-overhead is not always necessary.

    We present a second way to implement the reflection through the circulation space, based on graph decomposition, and as such we will refer to it as the ``decomposition method''. To that end, we first introduce these decompositions formally.

    \begin{definition}[Circulation space decompositions]
        \label{def:tree-parallel-decomposition}
        Let $G = (V,E)$ be an undirected graph with resistances $r : E \to \R_{>0}$. Let $E_1, \dots, E_k$ be a partition of $E$ such that each of the $E_j$'s is not empty, and for all $j \in [k]$, let $V_j = \{v \in V : N(v) \cap E_j \neq \varnothing\}$.
        \begin{enumerate}[nosep]
            \item \textit{Tree decomposition.} Suppose that for all $j,j' \in [k]$ with $j \neq j'$, we have $|V_j \cap V_{j'}| \leq 1$. Let $G' = (V',E')$ be the graph where we contract all of the edge sets $E_1, \dots, E_k$ into a single edge and then prune the edges that have a loose end, i.e., formally we write $V' = \{v_{j,j'} : \{j,j'\} \subseteq [k]^2 \land j \neq j' \land V_j \cap V_{j'} = \{v_{j,j'}\}\}$ and $E' = \{\{v_{j,j'},v_{\ell,\ell'}\} \in (V')^2 : |\{j,j',\ell,\ell'\}| = 3 \land v_{j,j'} \neq v_{\ell,\ell'}\}$ as a multiset. Suppose that $G'$ is a tree. Then, we refer to $G|_{E_1}, \dots, G|_{E_k}$ as a \textit{tree decomposition} of $G$.
            \item \textit{Parallel decomposition.} If there exist $s,t \in V$ with $s \neq t$, such that for all $j,j' \in [k]$ with $j \neq j'$, $V_j \cap V_{j'} = \{s,t\}$, then we refer to $G|_{E_1}, \dots, G|_{E_k}$ as a \textit{parallel decomposition} of $G$.
        \end{enumerate}
        Finally, let $T$ be a tree, where every node is labeled by a subset of $E$, such that:
        \begin{enumerate}[nosep]
            \item The root node is labeled by $E$.
            \item Every leaf is labeled by a single edge $\{e\}$, with $e \in E$.
            \item If an internal node is labeled by $E'$ and all its children are labeled by $E_1, \dots, E_k$, then $G|_{E_1}, \dots, G|_{E_k}$ is a tree or parallel decomposition of $G|_{E'}$.
        \end{enumerate}
        Then, we refer to $T$ as a \textit{tree-parallel decomposition} of $G$.
    \end{definition}

    We present pictorial representations of the tree and parallel decompositions in \cref{fig:graph-decompositions}.

    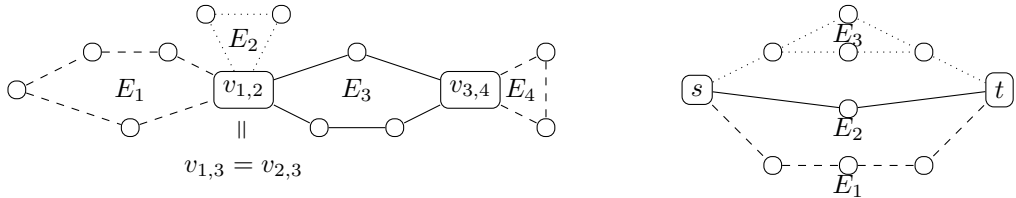
\begin{figure}[!ht]
        \centering
        \begin{tikzpicture}[vertex/.style = {draw, rounded corners = .3em}]
            \begin{scope}
                \node[vertex] (s) at (0,0) {};
                \node[vertex] (1) at (1,.5) {};
                \node[vertex] (3) at (1.5,-.5) {};
                \node[vertex] (2) at (2,.5) {};
                \node[vertex] (v) at (3,0) {$v_{1,2}$};
                \node[below] at (v.south) {$\begin{array}{c}
                    \rotatebox{90}{=} \\
                    v_{1,3} = v_{2,3}
                \end{array}$};
                \node[vertex] (4) at (4,-.5) {};
                \node[vertex] (6) at (4.5,.5) {};
                \node[vertex] (5) at (5,-.5) {};
                \node[vertex] (t) at (6,0) {$v_{3,4}$};
                \node[vertex] (7) at (2.5,1) {};
                \node[vertex] (8) at (3.5,1) {};
                \node[vertex] (9) at (7,.5) {};
                \node[vertex] (10) at (7,-.5) {};

                \node at (1.5,0) {$E_1$};
                \node at (3,{2/3}) {$E_2$};
                \node at (4.5,0) {$E_3$};
                \node at ({6+2/3},0) {$E_4$};

                \draw[dashed] (s) to (1) to (2) to (v);
                \draw[dashed] (s) to (3) to (v);
                \draw (v) to (4) to (5) to (t);
                \draw (v) to (6) to (t);
                \draw[dotted] (v) to (7) to (8) to (v);
                \draw[dashed] (t) to (9) to (10) to (t);
            \end{scope}
            \begin{scope}[shift = {(9,0)}]
                \node[vertex] (s) at (0,0) {$s$};
                \node[vertex] (1) at (1,-1) {};
                \node[vertex] (2) at (2,-1) {};
                \node[vertex] (3) at (3,-1) {};
                \node[vertex] (4) at (2,-.25) {};
                \node[vertex] (5) at (1,.5) {};
                \node[vertex] (6) at (2,1) {};
                \node[vertex] (7) at (2,.5) {};
                \node[vertex] (8) at (3,.5) {};
                \node[vertex] (t) at (4,0) {$t$};

                \draw[dashed] (s) to (1) to (2) node[below] {$E_1$} to (3) to (t);
                \draw (s) to (4) node[below] {$E_2$} to (t);
                \draw[dotted] (s) to (5) to (6) node[below] {$E_3$} to (8) to (t);
                \draw[dotted] (5) to (7) to (8);
            \end{scope}
        \end{tikzpicture}
        \caption{Examples of the tree decomposition (left), and the parallel decomposition (right). The dashed, dotted and solid sets of edges represent the disjoint edge sets $E_1, \dots, E_k$.}
        \label{fig:graph-decompositions}
    \end{figure}

    We remark that every graph admits a tree-parallel decomposition. To see this, observe that we can always take two adjacent vertices $s,t \in V$, and do a parallel decomposition with $E_1 = N(v) \cap N(w)$, and $E_2 = E \setminus E_1$. If this leads to $E_2$ being empty, then we can decompose $E_1$ into single edges with a parallel decomposition. Inductively, then, this means that we can always iterate this process until we have decomposed the entire graph into subsets that just contain a single edge. Moreover, we can always upper bound the depth of the tree-parallel decomposition by $|E|$.

    We now proceed to show how these decompositions relate to the circulation space.

    \begin{lemma}[Tree decomposition of the circulation space]
        \label{lem:tree-decomposition}
        Let $G = (V,E)$ be a graph with resistances $r : E \to \R_{>0}$. Let $E_1, \dots, E_k \subseteq E$ such that $G|_{E_1}, \dots, G|_{E_k}$ is a tree decomposition of $G$. Then,
        \[\mathcal{C}_{G,r} = \bigoplus_{j=1}^k \mathcal{C}_{G|_{E_j},r|_{E_j}}.\]
    \end{lemma}

    \begin{proof}
        The ``$\supseteq$''-direction is clear, so it remains to prove the ``$\subseteq$''-direction. To that end, let $f : E \to \R$ be a circulation on $G$. For every edge $\{v_{j,j'},v_{j,j''}\} = e' \in E'$, we can compute the net flow over that edge as
        \[f_{e'} = \sum_{e \in E_j \cap N^+(v_{j,j'})} f_e - \sum_{e \in E_j \cap N^-(v_{j,j'})} f_e.\]
        This defines a flow on $G'$, and since the net-flow on every vertex $v \in V$ is zero, so is the flow on all the vertices $v \in V'$. But as $G'$ is a tree, that means that $f_{e'} = 0$ for all $e' \in E'$. Thus, for every edge $\{v_{j,j'},v_{j,j''}\} = e' \in E'$, we have
        \[\sum_{e \in E_j \cap N^+(v_{j,j'})} - \sum_{e \in E_j \cap N^-(v_{j,j'})} = 0,\]
        and so we observe for all $j \in [k]$ that $f|_{E_j}$ is a circulation in $G|_{E_j}$. Since the $E_j$'s form a partition of $E$, we can write $f$ as a linear combination of circulations in $G|_{E_j}$.
    \end{proof}

    \begin{lemma}[Parallel decomposition of the circulation space]
        \label{lem:parallel-decomposition}
        Let $G = (V,E)$ be a graph with resistances $r : E \to \R_{>0}$. Let $E_1, \dots, E_k \subseteq E$, such that $G|_{E_1}, \dots, G|_{E_k}$ is a parallel decomposition of $G$. Let $\mathcal{E} : \C^k \to \H_E = \Span\{\ket{e} : e \in E\}$, that for all $j \in [k]$ performs the mapping
        \begin{equation}
            \label{eq:parallel-embedding}
            \mathcal{E}: \ket{j}\ket{\bot} \mapsto \ket{j}\frac{\ket{f_j}}{\norm{\ket{f_j}}}, \qquad \text{with} \qquad \ket{f_j} = \ket{f^{\min}_{G|_{E_j},s,t,r|_{E_j}}} \in \C^{E_j}.
        \end{equation}
        Then,
        \begin{equation}
            \label{eq:parallel-decomposition}
            \mathcal{C}_{G,r} = \mathcal{E}\left(\mathcal{C}'\right) \oplus \bigoplus_{j=1}^k \mathcal{C}_{G|_{E_j},r|_{E_j}}, \quad \text{with} \quad \mathcal{C}' = \Span\left\{\ket{\psi}\right\}^{\perp}, \quad \text{and} \quad \ket{\psi} = \sum_{j=1}^k \frac{1}{\norm{\ket{f_j}}}\ket{j}.
        \end{equation}
    \end{lemma}

    \begin{proof}
        For the ``$\supseteq$''-direction, it is clear that $\mathcal{C}_{G|_{E_j},s,t,r|_{E_j}}$ is contained in $\mathcal{C}_{G,r}$. On the other hand, let $\sum_{j=1}^k f_j'\norm{\ket{f^{\min}_{G|_{E_j},s,t,r|_{E_j}}}}\ket{j} = \ket{f'} \in \mathcal{C}'$. Then, we observe that
        \[0 = \bra{f'} \sum_{j=1}^k \frac{1}{\norm{\ket{f^{\min}_{G|_{E_j},s,t,r|_{E_j}}}}}\ket{j} = \sum_{j=1}^k \frac{f_j'\norm{\ket{f^{\min}_{G|_{E_j},s,t,r|_{E_j}}}}}{\norm{\ket{f^{\min}_{G|_{E_j},s,t,r|_{E_j}}}}} = \sum_{j=1}^k f_j'.\]
        We find immediately that the flow is conserved at every node besides $s$ and $t$, and for $s$ we find
        \begin{align*}
            &\sum_{e \in N^+(s)} \bra{e}\frac{1}{\sqrt{r_e}} \mathcal{E}(\ket{f'}) - \sum_{e \in N^-(s)} \bra{e}\frac{1}{\sqrt{r_e}} \mathcal{E}(\ket{f'}) \\
            &\qquad = \sum_{j=1}^k \left[\sum_{e \in E_j \cap N^+(s)} \frac{1}{\sqrt{r_e}} f_j'(f^{\min}_{G|_{E_j},s,t,r|_{E_j|}})_e\sqrt{r_e} - \sum_{e \in E_j \cap N^-(s)} \frac{1}{\sqrt{r_e}} f_j'(f^{\min}_{G|_{E_j},s,t,r|_{E_j|}})_e\sqrt{r_e}\right] \\
            &\qquad = \sum_{j=1}^k f_j'\left[\sum_{e \in E_j \cap N^+(s)} (f^{\min}_{G|_{E_j},s,t,r|_{E_j|}})_e - \sum_{e \in E_j \cap N^-(s)} (f^{\min}_{G|_{E_j},s,t,r|_{E_j|}})_e\right] = \sum_{j=1}^k f_j' = 0.
        \end{align*}
        The same argument can be applied to $t$, and so $\mathcal{E}(\ket{f'}) \in \mathcal{C}_{G,r}$.

        For the ``$\subseteq$''-direction, let $f$ be a circulation on $G$. For all $j \in [k]$, we write
        \[f_j' = \sum_{e \in E_j \cap N^+(s)} f_e - \sum_{e \in E_j \cap N^-(s)} f_e.\]
        Since $f$ is a circulation, we have net-zero flow on all the nodes, so in particular on those besides $s$ and $t$. This means that the flow flowing from $s$ into $E_j$, i.e., $f_j'$, is the same as the flow from $E_j$ into $t$. As such, we observe that we have an $st$-flow through $E_j$ with total flow $f_j'$, and so we can write
        \[\ket{f|_{E_j}} - f_j'\ket{f^{\min}_{G|_{E_j},s,t,r|_{E_j}}} \in \mathcal{C}_{G|_{E_j},r|_{E_j}}.\]
        As such, it remains to prove that $\sum_{j=1}^k f_j'\ket{f^{\min}_{G|_{E_j},s,t,r|_{E_j}}} \in \mathcal{E}(\mathcal{C}')$, for which it suffices to prove that $\ket{f'} := \sum_{j=1}^k f_j'\norm{\ket{f^{\min}_{G|_{E_j},s,t,r|_{E_j}}}}\ket{j} \in \mathcal{C}'$. We conclude by observing that
        \[\bra{f'} \sum_{j=1}^k \frac{1}{\norm{\ket{f^{\min}_{G|_{E_j},s,t,r|_{E_j}}}}}\ket{j} = \sum_{j=1}^k \frac{f_j'\norm{\ket{f^{\min}_{G|_{E_j},s,t,r|_{E_j}}}}}{\norm{\ket{f^{\min}_{G|_{E_j},s,t,r|_{E_j}}}}} = \sum_{j=1}^k f_j' = 0.\qedhere\]
    \end{proof}

    We observe that the tree and parallel decompositions effectively split up the circulation space into several mutually orthogonal components, through which we can reflect individually. These decompositions can be useful to implement the reflection through the circulation space time-efficiently in the circuit model. We prove this in the following theorem.

    \begin{theorem}
        \label{thm:implementation-circuit-model}
        Let $G = (V,E)$ be a graph with resistances $r : E \to \R_{>0}$, and let $T$ be a tree-parallel decomposition of the circulation space of $G$, with depth $d$. For every $\ell \in [d]$, let $k_{\ell}$ be the maximum number of children of nodes in the $(\ell-1)$th vertex layer. We use a Hilbert space $\C^{[k_1] \cup \{\bot\}} \otimes \cdots \otimes \C^{[k_d] \cup \{\bot\}}$, with dimension $K = (k_1+1) \cdots (k_d+1)$.

        Every node $u$ in $T$ is connected to its root by a unique path. If $u$ is in the $\ell$th vertex layer, we write $u = (j_1, \dots, j_\ell)$, if $u$ is reached through starting at the root node and taking the $j_{\ell'}$th child at the $(\ell'-1)$th vertex layer, for all $\ell' \in [\ell]$. Moreover, since every leaf $u = (j_1, \dots, j_d)$ is uniquely labeled by a singleton $\{e\}$, we identify
        \[\ket{e} := \ket{j_1, \dots, j_d}.\]

        Next, let an internal node $u \in V$ be labeled by $E_u$, and let $u_1, \dots, u_k$ be its children, labeled by $E_{u_1}, \dots, E_{u_k}$. Suppose that $G|_{E_{u_1}}, \dots, G|_{E_{u_k}}$ is a parallel decomposition of $G|_{E_u}$, with source and sink nodes $s$ and $t$. For all $j \in [k]$, we write
        \[\ket{\psi_u} := \sum_{j=1}^k \frac{1}{\norm{\ket{f_{u,j}}}}\ket{j}, \qquad \text{with} \qquad \ket{f_{u,j}} := \ket{f^{\min}_{G|_{E_{u_j}}, s, t, r|_{E_{u_j}}}}.\]

        Now, for every $\ell \in [d]$, we write the operation $\mathcal{E}_{\ell}$ that for all $u = (j_1, \dots, j_{\ell-1})$ in the $(\ell-1)$th layer of $T$ acts for all $j \in [k]$,
        \[\mathcal{E}_{\ell} : \ket{j_1,\dots,j_{\ell-1}}\ket{j}\ket{\bot}^{\otimes(d-\ell)} \mapsto \begin{cases}
            \ket{j_1,\dots,j_{\ell-1}}\ket{j}\ket{\bot}^{\otimes(d-\ell)}, & \text{if } u \text{ is the root of a tree decomposition}, \\
            \ket{j_1,\dots,j_{\ell-1}}\ket{j}\frac{\ket{f_{u,j}}}{\norm{\ket{f_{u,j}}}}, & \text{otherwise}.
        \end{cases}\]

        Similarly, for all $\ell \in [d]$, we write $U_{\ell}$ as the operation that for all $u = (j_1, \dots, j_{\ell-1})$ in the $(\ell-1)$th layer of $T$ acts as
        \[U_{\ell} : \ket{j_1,\dots,j_{\ell-1}}\ket{\bot}\ket{\bot}^{\otimes(d-\ell)} \mapsto \begin{cases}
            \ket{j_1,\dots,j_{\ell-1}}\ket{\bot}\ket{\bot}^{\otimes(d-\ell)}, & \text{if } u \text{ is the root of a tree decomposition}, \\
            \ket{j_1, \dots, j_{\ell-1}}\ket{\psi_u}\ket{\bot}^{\otimes(d-\ell)}, & \text{otherwise}.
        \end{cases}\]

        Then, we can implement the reflection through $\mathcal{C}_{G,r}$ in the circuit model with two calls to all the $\mathcal{E}_{\ell}$'s and $U_{\ell}$'s, and $\widetilde{O}(d\log(K))$ auxiliary elementary gates.
    \end{theorem}

    \begin{proof}
        We prove that we can implement the reflection through $\mathcal{C}_{G,r}$ by
        \[R_{\mathcal{C}_{G,r}} = \prod_{\ell=1}^d -\underbrace{\mathcal{E}_{\ell}U_{\ell} \left(I^{\otimes(\ell-1)} \otimes \left[(-R_{\ket{\bot}}) \otimes (\ket{\bot}\bra{\bot})^{\otimes(d-\ell)} + (-I) \otimes (I - (\ket{\bot}\bra{\bot})^{\otimes(d-\ell)})\right]\right) U_{\ell}^{\dagger} \mathcal{E}_{\ell}^{\dagger}}_{=: R_\ell}.\]
        The claim about the number of calls to the $\mathcal{E}_{\ell}$'s and the $U_{\ell}$'s then follows immediately. The reflection through $\ket{\bot}$ and its control structure can both be implemented can be implemented using $\widetilde{O}(\log(K))$ auxiliary gates.

        Let $\ell \in [d]$. We observe that $\mathcal{E}_{\ell}U_{\ell}$ acts as identity on $\ket{j_1, \dots, j_{\ell-1}} \otimes \C^{[k_{\ell}] \cup \{\bot\}} \otimes \cdots \otimes \C^{[k_d] \cup \{\bot\}}$, if $(j_1, \dots, j_{\ell-1})$ is the root node of a tree decomposition. Using \cref{lem:tree-decomposition,lem:parallel-decomposition}, it remains to show that if $u = (j_1, \dots, j_{\ell-1})$ is the root node of a parallel decomposition, $R_\ell$ reflects through all the spaces $\mathcal{E}_{\ell}(\ket{j_1, \dots, j_{\ell-1}} \otimes \Span\{\ket{\psi_u}\}^{\perp} \otimes \ket{\bot}^{\otimes (d-\ell)})$. This is equivalent to showing that $U_{\ell}^{\dagger}\mathcal{E}_{\ell}^{\dagger}R_{\ell}\mathcal{E}_{\ell}U_{\ell}$ reflects through $\C^{([k_1] \cup \{\bot\})} \otimes \cdots \otimes \C^{[k_{\ell-1}] \cup \{\bot\}} \otimes \Span\{\ket{\bot}\}^{\perp} \otimes \ket{\bot}^{\otimes (d-\ell)}$, which is indeed the action of the middle operation of $R_{\ell}$.
    \end{proof}

    In specific instances it is possible to implement the operations $\mathcal{E}_\ell$ and $U_\ell$ efficiently in the circuit model. However, we remark here that is possible to give a time-efficient implementation of these operations unconditionally, in the QROM-model.

    \begin{theorem}
        \label{thm:circulation-space-reflection-implementation}
        Let $G = (V,E)$ be a graph with resistances $r : E \to \R_{>0}$, and let $T$ be a tree-parallel decomposition of $G$, with depth $d$. For every $\ell \in [d]$, let $k_{\ell}$ be the maximum number of children nodes of vertices in the $(\ell-1)$th vertex layer, and let $K = (k_1 + 1) \cdots (k_d + 1)$. Then, we can implement $R_{\mathcal{C}_{G,r}}$ in the QROM-model with $\widetilde{O}(|E|K) \subseteq \widetilde{O}(|E|^{d+1})$ bits of QROM, and $\widetilde{O}(d\log(K)) \subseteq \widetilde{O}(d^2\log|E|)$ elementary gates.
    \end{theorem}

    \begin{proof}
        It remains to implement the operations $\mathcal{E}_{\ell}$ and $U_{\ell}$ in \cref{thm:implementation-circuit-model}. To that end, observe that both are parallel quantum state-preparation operations. For each internal node, we need to store the state $\ket{\psi_u}$, and the states $\ket{f_{u,j}}$, all of which requires $\widetilde{O}(K)$ bits of QROM to store. The number of nodes in the tree for which we have a parallel decomposition is at most $|E|$, and so we obtain that the QROM-size required is $\widetilde{O}(|E|K)$.

        Finally, every state-preparation routine requires $\widetilde{O}(\log(K))$ elementary gates to be implemented, and this is also the number of gates to implement it in parallel. Thus, the total number of gates required to implement the entire reflection through $\mathcal{C}_{G,r}$ is $\widetilde{O}(d\log(K))$.
    \end{proof}

    Note that even though the above implementations can be quite efficient when it comes to the number of elementary gates, they are typically still very space-inefficient. In general, we require a QROM of size $\widetilde{O}(|E|K)$ to implement the aforementioned operations, which, as we will see in \cref{sec:relations,sec:applications}, can be exponentially large for some applications. Moreover, it also takes time to classically precompute the contents of this memory, which might also take exponential time to do. Thus, when applying these techniques to specific applications, it might still be necessary to find problem-specific improvements over the generic techniques presented in this section, and we will see several examples of that in the following sections.

    \section{Relations to other quantum algorithmic frameworks}
    \label{sec:relations}

    In this section, we elaborate on how the graph composition framework relates to other ways of designing quantum algorithms that exist in various places in the literature.

    \subsection{$st$-connectivity and planar graphs}

    The $st$-connectivity problem in the adjacency matrix model has received quite some attention in the literature. The first quantum algorithm that solves it in an undirected graph of $n$ vertices was presented by D\"urr et al.~\cite{durr2006quantum}, who constructed an $O(n^{3/2})$-query algorithm. It was subsequently improved by Belovs and Reichardt~\cite{belovs2012span}, who constructed a span-program-based algorithm that makes $\widetilde{O}(n\sqrt{\ell})$ queries, if one is promised that if $s$ and $t$ are connected, there always exists a path between them of length at most $\ell$. Later, Jeffery and Kimmel improved the analysis of the algorithm from \cite{belovs2012span} in the special case where $G \cup \{\{s,t\}\}$ is a planar graph, relating the span program witness sizes to effective resistances~\cite{jeffery2017quantum}. Moreover, they also showed how the boolean formula evaluation problem fits in this framework. Finally, Jarret et al.~\cite{jarret2018quantum} generalized the approach taken in \cite{jeffery2017quantum} to non-planar graphs, leading to the current state-of-the-art algorithm for solving the $st$-connectivity problem in the adjacency matrix model~\cite[Corollary~18]{jarret2018quantum}.

    The graph composition framework is a strict generalization of the setting considered in the $st$-connectivity problem. In the $st$-connectivity problem, one can make direct queries whether a given edge $\{u,v\} \in V^2$ is present in the graph. In the graph composition framework, we instead allow for putting arbitrary span programs on the edges in the graph that compute whether the edge is present. Indeed, the main result from \cite{jarret2018quantum} for a graph $G = (V,E)$ can be recovered from \cref{thm:graph-composition} by setting all the span programs $(\mathcal{P}_e)_{e \in E}$ to scalar multiples of trivial span programs that evaluate the $e$th bit of a bitstring $x \in \{0,1\}^E$.

    We formalize the definition of an $st$-connectivity graph here.

    \begin{definition}
        Let $G = (V,E)$ be an undirected multigraph with resistances $r : E \to \R_{>0}$, and $s,t \in V$ with $s \neq t$. Let $\D \subseteq \{0,1\}^n$ and for all $e \in E$, let $j(e) \in [n]$ and $b(e) \in \{0,1\}$. We refer to $\mathcal{G} = (G,j,b,r)$ as an $st$-connectivity graph on $\D$, and we say that it computes the function $f : \D \to \{0,1\}$, where $f(x) = 1$ if and only if $s$ and $t$ are connected along the graph $G(x) = (V,E(x))$, where $E(x) = \{e \in E : x_{j(e)} = b(e)\}$. We let $w_+(x,\mathcal{G}) = R_{G(x),s,t,r}$, and $w_-(x,\mathcal{G}) = R_{G,s,t,r'}^{-1}$, where $r'(e) = r(e)^{-1}$ if $x_{j(e)} \neq b(e)$, and $0$ otherwise. We say that the complexity of $\mathcal{G}$ is
        \[C(\mathcal{G}) = \sqrt{\max_{x \in f^{-1}(1)} w_+(x,\mathcal{G}) \cdot \max_{x \in f^{-1}(0)} w_-(x,\mathcal{G})}.\]
        Finally, for any boolean function $f : \D \to \{0,1\}$ with $\D \subseteq \{0,1\}^n$, we write $\st(f)$ for the minimal complexity $C(\mathcal{G})$ of an $st$-connectivity graph $\mathcal{G}$ that computes $f$.
    \end{definition}

    \begin{theorem}
        \label{thm:st-connectivity}
        Let $n \in \mathbb{N}$ and $\D \subseteq \{0,1\}^n$. Every $st$-connectivity graph $\mathcal{G} = (G,j,e,r)$ on $\D$ can be turned into a graph composition $\mathcal{P}$ on $\D$, such that $C(\mathcal{P}) = C(\mathcal{G})$.
    \end{theorem}

    \begin{proof}
        Let $G = (V,E)$. For all $j \in [n]$, let $\mathcal{P}_j$ on $\D$ be a trivial span program that computes $j$th bit. For all the edges $e \in E$, let $\mathcal{P}_e = r_e\mathcal{P}_{j(e)}$ if $b(e) = 1$, and $\mathcal{P}_e = r_e(\lnot\mathcal{P}_{j(e)})$ if $b(e) = 0$. Now, we take $\mathcal{P}$ to be the graph composition of $(\mathcal{P}_e)_{e \in E}$ on $G$. The analyses of the witness sizes then follows directly from \cref{thm:graph-composition}.
    \end{proof}

    It turns out that in the case where $G \cup \{\{s,t\}\}$ is a planar graph, the span program negation of the graph composition of $G$ is related to a span program composition of the planar dual graph $G^{\dagger}$. This was observed in the unit-cost setting in \cite{jeffery2017quantum}, and here we extend this observation to our generalized setting. The core result that makes the special case of planar graphs interesting to us is the following characterization of the circulation space of dual graphs.

    \begin{lemma}
        \label{lem:planar-dual-circulation-space}
        Let $G = (V,E)$ be a connected, undirected, planar graph, where every edge $e \in E$ has an implicit direction meaning that it goes from $e_-$ to $e_+$. Let $G^{\dagger} = (F^{\dagger},E^{\dagger})$ be a planar dual of $G$, where every edge $e$ is associated to a dual edge $e^{\dagger} \in E^{\dagger}$, that points from the face on the left of $e$ to the face on the right of $e$, as seen from the perspective of traversing $e$ in its implicit direction. Moreover, let $r : E \to [0,\infty]$ be resistances on $G$, and let $r^{\dagger} : E^{\dagger} \to [0,\infty]$ be defined by $r^{\dagger}_{e^{\dagger}} = 1/r_e$. Then, by identifying $\ket{e^{\dagger}} = \ket{e}$, and as such identifying $\H_{G^{\dagger}} = \H_G$, we find that
        \[\H_G = \mathcal{C}_{G,r} \oplus \mathcal{C}_{G^{\dagger},r^{\dagger}} = \H_{G^{\dagger}}.\]
    \end{lemma}

    \begin{proof}
        Let $\ket{f} \in \mathcal{C}_{G^{\dagger},r^{\dagger}}$. Let $f^{\dagger} \in F^{\dagger}$ arbitrarily, and let $f'$ be the unit flow in $G$ around $f^{\dagger}$ in the counter-clockwise direction. Let $e_1, \dots, e_k$ and $e_1', \dots, e_{k'}'$ be the edges it traverses in the right and wrong direction, respectively. Note that $N_+(f^{\dagger}) = \{e_1^{\dagger}, \dots, e_k^{\dagger}\}$, and $N_-(f^{\dagger}) = \{(e_1')^{\dagger}, \dots (e_{k'}')^{\dagger}\}$. Thus, we observe that
        \[\braket{f'}{f} = \sum_{j=1}^k \sqrt{r_{e_j} \cdot r_{e_j^{\dagger}}^{\dagger}}f_{e_j^{\dagger}} - \sum_{j=1}^{k'} \sqrt{r_{e_j'} \cdot r_{(e_j')^{\dagger}}^{\dagger}}f_{(e_j')^{\dagger}} = \sum_{e \in N_+(f^{\dagger})} f_{e^{\dagger}} - \sum_{e \in N_-(f^{\dagger})} f_{e^{\dagger}} = 0,\]
        Next, observe that such cycle flows $f'$ span the circulation space of $G$, we obtain $\ket{f} \in \mathcal{C}_{G,r}^{\perp}$. We thus find $\mathcal{C}_{G^{\dagger},r^{\dagger}} \subseteq \mathcal{C}_{G,r}^{\perp}$, and by symmetry $\mathcal{C}_{G,r} \subseteq \mathcal{C}_{G^{\dagger},r^{\dagger}}^{\perp}$. Finally, observe that there are $|F^{\dagger}|-1$ linearly independent cycle flows in $G$, and similarly $|V|-1$ linearly independent cycle flows in $G^{\dagger}$, and so we find using Euler's formula that
        \[\dim(\mathcal{C}_{G,r}) + \dim(\mathcal{C}_{G^{\dagger},r^{\dagger}}) = |F^{\dagger}|-1 + |V|-1 = |E| = \dim(\H_G).\qedhere\]
    \end{proof}

    This result gives a nice relation between dual-graph-composed span programs and their negation. We note that this result elegantly recovers~\cite[Lemma~11]{jeffery2017quantum}.

    \begin{theorem}
        \label{thm:dual-graph-composition}
        Let $G = (V,E)$ be a connected, undirected graph with $s,t \in V$ and $s \neq t$, such that $G \cup \{(s,t)\}$ is planar. Let $G^{\dagger} = (F^{\dagger},E^{\dagger})$ be a planar dual of $G \cup \{(s,t)\}$ with the dual edge of $(s,t)$, connecting $s^{\dagger}$ and $t^{\dagger}$, removed. For all $e \in E$ we denote its dual edge by $e^{\dagger}$, we let $\mathcal{P}$ be the span program formed by taking the graph composition of $G$ with the span programs $(\mathcal{P}_e)_{e \in E}$, and we let $\mathcal{P}^{\dagger}$ be the graph composition of $G^{\dagger}$ with span programs $(\lnot\mathcal{P}_e)_{e^{\dagger} \in E^{\dagger}}$. Then, $\lnot\mathcal{P} = \mathcal{P}^{\dagger}$.
    \end{theorem}

    \begin{proof}
        Let $\ket{st}$ be the vector with weight $1$ on the edge between $s$ and $t$, and let the default direction of this edge be from $s$ to $t$. We define $r : E \to [0,\infty]$ as $r_e = \norm{\ket{w_0^e}}^2$, and $r_{st} = 1$. Now, we observe that
        \[\ket{f^{\min}_{G,s,t,r}} \oplus -\ket{st} \in \mathcal{C}_{G \cup \{(s,t)\}, r}, \qquad \text{and} \qquad \mathcal{C}_{G \cup \{(s,t)\},r} = \mathcal{C}_{G,r} \oplus \Span\{\ket{f^{\min}_{G,s,t,r}} \oplus -\ket{st}\},\]
        where we used that $\ket{f^{\min}_{G,s,t,r}} \in \mathcal{C}_{G,r}^{\perp}$. By symmetry, we find that
        \[\mathcal{C}_{G^{\dagger} \cup \{(s^{\dagger},t^{\dagger})\}, r^{\dagger}} = \mathcal{C}_{G^{\dagger},r^{\dagger}} \oplus \Span\{\ket{f^{\min}_{G^{\dagger},s^{\dagger},t^{\dagger},r^{\dagger}}} \oplus -\ket{s^{\dagger}t^{\dagger}}\},\]
        and recall by \cref{lem:planar-dual-circulation-space} that
        \[\mathcal{C}_{G \cup \{(s,t)\},r} \oplus \mathcal{C}_{G^{\dagger} \cup \{(s^{\dagger},t^{\dagger})\}, r^{\dagger}} = \H_G \oplus \Span\{\ket{st}\},\]
        and since $\ket{st}$ and $\ket{s^{\dagger}t^{\dagger}}$ are directed oppositely, we find $\ket{f^{\min}_{G,s,t,r}} = \ket{f^{\min}_{G^{\dagger},s^{\dagger},t^{\dagger},r^{\dagger}}}/\norm{\ket{f^{\min}_{G^{\dagger},s^{\dagger},t^{\dagger},r^{\dagger}}}}^2$.

        Now, we verify that $\lnot\mathcal{P} = \mathcal{P}^{\dagger}$. To that end, note that the embedding $\mathcal{E}$ is the same for both graphs, and so we have
        \[\ket{w_0'} = \frac{\ket{w_0}}{\norm{\ket{w_0}}^2} = \frac{\mathcal{E}(\ket{f^{\min}_{G,s,t,r}})}{\norm{\ket{f^{\min}_{G,s,t,r}}}^2} = \mathcal{E}\left(\ket{f^{\min}_{G^{\dagger},s^{\dagger},t^{\dagger},r^{\dagger}}}\right) = \ket{w_0^{\dagger}},\]
        and we observe from the definition that for all $x \in \D$, we have
        \[\H(x)^{\perp} = \left[\bigoplus_{e \in E} \H^e(x)\right]^{\perp} = \bigoplus_{e \in E} \H^e(x)^{\perp} = \H^{\dagger}(x).\]
        Thus, it remains to prove that $(\K \oplus \Span\{\ket{w_0}\})^{\perp} = \K^{\dagger}$. To that end, observe that
        \begin{align*}
            (\K \oplus \Span\{\ket{w_0}\})^{\perp} &= \left[\bigoplus_{e \in E} \K_e \oplus \mathcal{E}(\mathcal{C}_{G,r}) \oplus \Span\{\mathcal{E}(\ket{f^{\min}_{G,s,t,r}})\}\right]^{\perp} \\
            &= \bigoplus_{e \in E} (\K_e \oplus \Span\{\ket{w_0^e}\})^{\perp} \oplus \mathcal{E}((\mathcal{C}_{G,r} \oplus \Span\{\ket{f^{\min}_{G,s,t,r}}\})^{\perp}) \\
            &= \bigoplus_{e \in E} (\K_e \oplus \Span\{\ket{w_0^e}\})^{\perp} \oplus \mathcal{E}(\mathcal{C}_{G^{\dagger},r^{\dagger}}) = \K^{\dagger}.\qedhere
        \end{align*}
    \end{proof}

    \subsection{Variable-time search, formula evaluation and divide and conquer}
    \label{subsec:var-time-search-formula-evaluation-divide-conquer}

    We continue by making the fundamental observation that we can encode logic in the structure of the graph that we use to compose span programs.\footnote{The idea of encoding logic in series-parellel graphs is not new -- it was already considered by Shannon~\cite{shannon1938symbolic,shannon1949synthesis}.} Indeed, if we want to compute the AND or OR of several span programs, there is an easy way to do this using the graph composition framework. We briefly introduce these constructions here, analogously to earlier works, e.g., \cite{reichardt2012span,reichardt2009span,jeffery2017quantum}, \cite[Section~5.1]{cornelissen2020span} and \cite[Section~3.3]{jeffery2024multidimensional}.

    \begin{definition}[AND- and OR-composition]
        \label{def:logical-compositions}
        Let $n \in \N$, and $\mathcal{P}_1, \dots, \mathcal{P}_n$ be span programs on $\D$. We define two special types of graph compositions:
        \begin{enumerate}[nosep]
            \item Let $G$ be a line graph with $n$ edges, and let $s$ and $t$ be the endpoints. We let $\mathcal{P}$ be the graph composition of $G$ with span programs $\mathcal{P}_1, \dots, \mathcal{P}_n$ on the edges, and we write $\mathcal{P} = \bigwedge_{j=1}^n \mathcal{P}_j$. We refer to this as the AND-composition of $\mathcal{P}_1, \dots, \mathcal{P}_n$.
            \item Let $G$ be a graph with $2$ nodes, $s$ and $t$, and $n$ parallel edges between them. We let $\mathcal{P}$ be the graph composition with span programs $\mathcal{P}_1, \dots, \mathcal{P}_n$ on the edges, and we write $\mathcal{P} = \bigvee_{j=1}^n \mathcal{P}_j$. We refer to this construction as the OR-composition of $\mathcal{P}_1, \dots, \mathcal{P}_n$.
        \end{enumerate}
    \end{definition}

    The pictorial interpretations of these constructions are provided in \cref{fig:and-or-compositions}, and we compute the witness sizes in \cref{thm:and-or-witness-sizes}.

    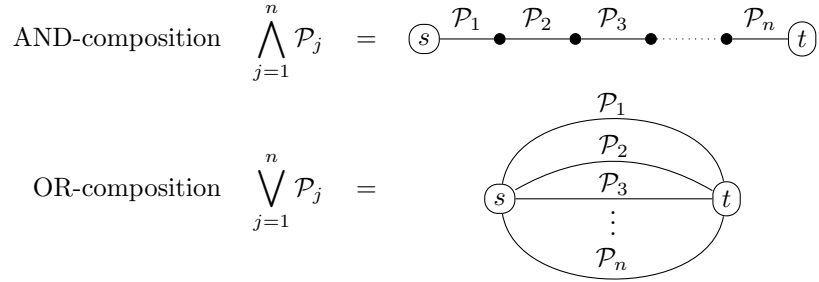
\begin{figure}[!ht]
        \centering
        \begin{tabular}{rccc}
            AND-composition & $\displaystyle\bigwedge_{j=1}^n \mathcal{P}_j$ & $=$ & \raisebox{-.4em}{\begin{tikzpicture}[label/.style = {draw, rounded corners = .5em}, bullet/.style = {draw, fill, minimum size = .4em, rounded corners = .2em, inner sep = 0pt}]
                \node[label] (s) at (0,0) {$s$};
                \node[label] (t) at (5,0) {$t$};
                \node[bullet] (1) at (1,0) {};
                \node[bullet] (2) at (2,0) {};
                \node[bullet] (3) at (3,0) {};
                \node[bullet] (4) at (4,0) {};
                \draw (s) to node[above] {$\mathcal{P}_1$} (1) to node[above] {$\mathcal{P}_2$} (2) to node[above] {$\mathcal{P}_3$} (3);
                \draw[dotted] (3) to (4);
                \draw (4) to node[above] {$\mathcal{P}_n$} (t);
            \end{tikzpicture}} \\
            OR-composition & $\displaystyle\bigvee_{j=1}^n \mathcal{P}_j$ & $=$ & \raisebox{-4em}{\begin{tikzpicture}[label/.style = {draw, rounded corners = .5em}, bullet/.style = {draw, fill, minimum size = .4em, rounded corners = .2em, inner sep = 0pt}]
                \node[label] (s) at (0,0) {$s$};
                \node[label] (t) at (3,0) {$t$};
                \draw[bend left=80] (s) to node[above=-.2em] {$\mathcal{P}_1$} (t);
                \draw[bend left=30] (s) to node[above=-.2em] {$\mathcal{P}_2$} (t);
                \draw (s) to node[above=-.2em] {$\mathcal{P}_3$} (t);
                \draw[bend right=80] (s) to node[above] {$\mathcal{P}_n$} (t);
                \node[below=-.5em] at (1.5,0) {$\vdots$};
            \end{tikzpicture}}
        \end{tabular}
        \caption{Pictorial representations of the AND- and OR-compositions of span programs.}
        \label{fig:and-or-compositions}
    \end{figure}

    \begin{theorem}[Witness sizes of AND- and OR-composition]
        \label{thm:and-or-witness-sizes}
        Let $\mathcal{P}_1, \dots, \mathcal{P}_n$ be span programs on a common domain $\D$. Then, for all $x \in \D$, we have
        \[w_+\left(x,\bigwedge_{j=1}^n \mathcal{P}_j\right) = \sum_{j=1}^n w_+(x,\mathcal{P}_j), \qquad \text{and} \qquad w_-\left(x,\bigwedge_{j=1}^n \mathcal{P}_j\right) = \left[\sum_{j=1}^n \frac{1}{w_-(x,\mathcal{P}_j)}\right]^{-1},\]
        and
        \[w_+\left(x,\bigvee_{j=1}^n \mathcal{P}_j\right) = \left[\sum_{j=1}^n \frac{1}{w_+(x,\mathcal{P}_j)}\right]^{-1}, \qquad \text{and} \qquad w_-\left(x,\bigvee_{j=1}^n \mathcal{P}_j\right) = \sum_{j=1}^n w_-(x,\mathcal{P}_j).\]
    \end{theorem}

    \begin{proof}
        The result follows directly from \cref{thm:graph-composition}, and the formulas for the effective resistance of series and parallel circuits.
    \end{proof}

    We can use logical compositions to recover a span-program version of variable-time search with known costs, which first appeared in \cite{ambainis2010quantum}.

    \begin{theorem}[Variable-time search]
        Let $\mathcal{P}_1, \dots, \mathcal{P}_n$ be span programs on a common domain $\D$. Then,
        \[\mathcal{P} = \bigvee_{j=1}^n \frac{\mathcal{P}_j}{W_+(\mathcal{P}_j)} \qquad \Rightarrow \qquad C\left(\mathcal{P}\right)^2 \leq \sum_{j=1}^n C(\mathcal{P}_j)^2.\]
    \end{theorem}

    \begin{proof}
        Let $x \in \D$. If $x$ is positive for $\mathcal{P}$, then it must be positive for at least one span program $\mathcal{P}_{j'}$. As such, we have
        \[w_+\left(x, \mathcal{P}\right) = \left[\sum_{j=1}^n \frac{1}{w_+\left(x,\frac{\mathcal{P}_j}{W_+(\mathcal{P}_j)}\right)}\right]^{-1} = \left[\sum_{j=1}^n \frac{W_+(\mathcal{P}_j)}{w_+(x,\mathcal{P}_j)}\right]^{-1} \leq \frac{w_+(x,\mathcal{P}_{j'})}{W_+(\mathcal{P}_{j'})} \leq 1,\]
        and so $W_+(\mathcal{P}) \leq 1$. On the other hand, if $x$ is negative for $\mathcal{P}$, then it must be negative for all $\mathcal{P}_j$'s, and so we have
        \[w_-(x,\mathcal{P}) = \sum_{j=1}^n w_-\left(x, \frac{\mathcal{P}_j}{W_+(\mathcal{P}_j)}\right) = \sum_{j=1}^n W_+(\mathcal{P}_j) \cdot w_-(x,\mathcal{P}_j) \leq \sum_{j=1}^n W_+(\mathcal{P}_j)W_-(\mathcal{P}_j) = \sum_{j=1}^n C(\mathcal{P}_j)^2.\qedhere\]
    \end{proof}

    These techniques have very direct consequences for the formula evaluation problem too. After a long line of research \cite{reichardt2012span,reichardt2009span,reichardt2011span,reichardt2011faster,ambainis2010any}, the formula evaluation problem with unit costs was eventually settled by Reichardt in~\cite[Corollary~1.6]{reichardt2010span}, who proved that any boolean formula of length $\ell$ can be evaluated with $O(\sqrt{\ell})$ quantum queries to the input variables.This result was later recovered by Jeffery and Kimmel in \cite[Theorem~16]{jeffery2017quantum}, using ideas very similar to those presented here.

    An interesting generalization is the question how efficiently we can evaluate boolean formulas if its bits cannot be queried directly, but can be evaluated by some subroutines of non-unit costs. This question was first considered in the context of the divide and conquer framework, which was introduced by Childs et al.~\cite[Lemma~1]{childs2022quantum}, who gave a query-efficient algorithm for the variable-time formula evaluation problem. Subsequently, Jeffery and Pass~\cite[Theorem~4.12]{jeffery2024multidimensional} gave a similar result when the boolean formula, when expressed as an AND-OR tree, is symmetric, and each input bit is evaluated by a subspace graph. They also extended their theorem to give a time-efficient implementation of this result if the input bits are evaluated by bounded-error quantum algorithms~\cite[Theorem~4.15]{jeffery2024multidimensional}. Here, we derive a similar statement using the graph composition framework where the input bits are computed by span programs.

    \begin{theorem}[Variable-time formula evaluation]
        \label{thm:variable-time-formula-evaluation}
        Let $n \in \N$, $\D \subseteq \{0,1\}^n$, and $\varphi : \D \to \{0,1\}$ be a boolean formula. Let $J(\varphi)$ be the multiset of the indices being queried by $\varphi$, e.g., $J((x_{j_1} \land x_{j_2}) \lor (\lnot x_{j_1} \land \lnot x_{j_2})) = \{j_1,j_2,j_1,j_2\}$. Let $\mathcal{P}_1, \dots, \mathcal{P}_n$ be span programs on $\D$, such that $\mathcal{P}_j$ computes the $j$th bit of $x \in \D$. Then we can build a graph composition $\mathcal{P}$ on $\D$ of a series-parallel graph $G$ that computes $\varphi$ with complexity
        \[C(\mathcal{P})^2 \leq \sum_{j \in J(\varphi)} C(\mathcal{P}_j)^2.\]
        Moreover, in the QROM-model, we can implement the reflection through the circulation space using a QROM with $\widetilde{O}(|J(\varphi)|^{d+1})$ bits, and with $\widetilde{O}\left(d^2\log|J(\varphi)|\right)$ elementary gates, where $d$ is the depth of the AND-OR-tree representation of $\varphi$.
    \end{theorem}

    \begin{proof}
        We start with the complexity. Because of De Morgan's law, we can assume without loss of generality that the boolean formula is entirely made up of $\lnot$'s and $\lor$'s. Moreover, recall from \cref{thm:dual-graph-composition} that we can take the negation of any series-parallel graph composition by considering its dual graph, which is again series-parallel. Next, we perform induction on the recursion depth of the boolean formula. \cref{thm:and-or-witness-sizes} handles the basis for induction, i.e., where the depth is $1$. Now, suppose that our theorem is true for all boolean formulas of depth at most $k-1$, and suppose $\varphi$ is of depth $k$. Now, we can write $\varphi$ as
        \[\varphi(x) = \bigvee_{m=1}^{m'} \varphi_m'(x),\]
        where all $\varphi_m'$'s are boolean formulas of depth at most $k-1$. Hence, by our induction hypothesis, we can find a graph composition $\mathcal{P}_m'$ that evaluates $\varphi_m'$ with complexity
        \[C(\mathcal{P}_m')^2 \leq \sum_{j \in J(\varphi_m')} C(\mathcal{P}_j)^2,\]
        and so using \cref{thm:and-or-witness-sizes}, we find that that
        \[C\left(\bigvee_{m=1}^{m'} \frac{\mathcal{P}_m'}{W_+(\mathcal{P}_m')}\right)^2 \leq \sum_{m=1}^{m'} C(\mathcal{P}_m')^2 \leq \sum_{m=1}^{m'} \sum_{j \in J(\varphi_m')} C(\mathcal{P}_j)^2 = \sum_{j \in J(\varphi)} C(\mathcal{P}_j)^2.\]

        It remains to check the time and space complexity in the QROM-model. To that end, we observe that the resulting graph is a series-parallel graph of depth $d$. This gives rise to a tree-parallel decomposition of depth $d$ as well, and so by \cref{thm:circulation-space-reflection-implementation}, we can implement the reflection through the circulation space with $\widetilde{O}(d|J(\varphi)|)$ bits of QROM, and $\widetilde{O}(d\polylog|J(\varphi)|)$ elementary gates.
    \end{proof}

    We now use variable-time formula evaluation can be used to show that the graph composition framework subsumes the first strategy of the quantum divide-and-conquer framework, as introduced by Childs et al.~\cite{childs2022quantum}. To that end, we briefly recall this strategy, and introduce it slightly more formally than in the previous work. A similar treatment can be found in \cite[Theorem~4.13]{jeffery2024multidimensional}.

    \begin{definition}[Quantum divide and conquer, {\cite[Strategy~1]{childs2022quantum}}]
        \label{def:divide-and-conquer}
        Let $a,b,n \in \N$, $\Sigma$ a finite alphabet, and $\Lambda$ a finite parameter space. For all $m \in \N$ and $\lambda \in \Lambda$, let $f_m^{\lambda} : \Sigma^n \to \{0,1\}$, $f_m^{\lambda,\aux} : \Sigma^n \to \{0,1\}$, $m_1, \dots, m_a \in \Theta(m/b)$, $\lambda_1, \dots, \lambda_a \in \Lambda$, and $\varphi_m^{\lambda} : \{0,1\}^{a+1} \to \{0,1\}$ be a boolean read-once formula, i.e., where each variable appears exactly once, such that there exists a $m_0 \in \N$, such that for all $m \geq m_0$ and $x \in \Sigma^n$,
        \[f_m^{\lambda}(x) = \varphi_m^{\lambda}(f_{m_1}^{\lambda_1}(x), \dots, f_{m_a}^{\lambda_a}(x), f_m^{\lambda,\aux}(x)).\]
        This defines a \textit{quantum divide and conquer strategy for $\{f_m^{\lambda}\}_{m \in \N, \lambda \in \Lambda}$}.
    \end{definition}

    Childs et al.\ prove that the query complexity of evaluating this family of functions $\{f_m^{\lambda}\}_{m \in \N, \lambda \in \Lambda}$ satisfies a recurrence relation~\cite[Equation~(10)]{childs2022quantum}. We recover this result here, and give a handle on a time-efficient implementation.

    \begin{theorem}
        \label{thm:divide-and-conquer}
        Consider a quantum divide and conquer strategy for a family of functions $\{f_m^{\lambda}\}_{m \in \N, \lambda \in \Lambda}$, in the sense of \cref{def:divide-and-conquer}. For all $m \in \N$, let $\mathcal{P}_m^{\lambda,\aux}$ be a span program that evaluates $f_m^{\lambda,\aux}$, and for all $m < m_0$, let $\mathcal{P}_m^{\lambda}$ be a span program evaluating $f_m^{\lambda}$. Now, we can build a family of graph compositions $(\mathcal{P}_m^{\lambda})_{m \in \N, \lambda \in \Lambda}$, where for all $m \geq m_0$, $\mathcal{P}_m^{\lambda}$ evaluates $f_m^{\lambda}$, and satisfies
        \[C(\mathcal{P}_m^{\lambda})^2 \leq \sum_{j=1}^a C(\mathcal{P}_{m_j}^{\lambda_j})^2 + C(\mathcal{P}_m^{\lambda,\aux})^2.\]
        Moreover, in the QROM-model, the complexity of implementing the reflection through the circulation space, in the graph composition for $\mathcal{P}_m^{\lambda}$, can be done with a QROM of size $\widetilde{O}(\mathrm{poly}(a^{\log(m)/\log(b)}))$ and using $\widetilde{O}((\log(m)/\log(b))^2\log(a))$ elementary gates.
    \end{theorem}

    \begin{proof}
        The complexity follows directly from plugging the span programs into \cref{thm:variable-time-formula-evaluation}. For the time complexity, we use the time characterization from \cref{thm:variable-time-formula-evaluation} for the reflection through the circulation space of the composed graph. We observe that the depth of the resulting formula is $\Theta(\log(m)/\log(b))$, since in every iteration we divide $m$ by $\Theta(b)$. Moreover, in every division step, we divide into $a+1$ different parts. We compute the quantities in \cref{thm:implementation-circuit-model}, and observe that $d \in \Theta(\log(m)/\log(b))$ and $K \in \widetilde{\Theta}(\mathrm{poly}(a^{\log(m)/\log(b)}))$. Plugging these into the expressions in \cref{thm:implementation-circuit-model} yields the result.
    \end{proof}

    We remark that a similar result to \cref{thm:divide-and-conquer} is obtained in \cite[Theorem~4.13]{jeffery2024multidimensional}. We state the resulting time complexity for general formulas $\varphi_m^{\lambda}$, and thereby incur extra overhead in the depth of the AND-OR-tree representation of the formula. Jeffery and Pass state the theorem for the more restricted class of symmetric formulas, which they can evaluate more efficiently.

    Finally, in \cite[Section~5]{jeffery2024multidimensional}, Jeffery and Pass apply the quantum divide and conquer technique to obtain a time-efficient implementation of Savitch's algorithm for the directed $st$-connectivity problem. We slightly improve on \cite[Theorem~5.3]{jeffery2024multidimensional}, by improving the time bound from $2^{\frac12\log^2(n) + O(\log(n))}$ to $O(2^{\frac12\log^2(n)} \cdot \polylog(n)) = \widetilde{O}(\sqrt{2n}^{\log(n)})$.

    \begin{theorem}
        \label{thm:savitch}
        There is an $st$-connectivity algorithm that solves the directed $st$-connectivity problem making $O(\sqrt{2n}^{\log(n)})$ queries, with $\widetilde{O}(\sqrt{2n}^{\log(n)})$ elementary gates, and using $O(\log^2(n))$ space.
    \end{theorem}

    \begin{proof}
        Without loss of generality, suppose that $n$ is a power of $2$. We use the divide and conquer framework to turn Savitch's algorithm~\cite{savitch1970relationships} into a boolean formula evaluation problem. To that end, we let $\varphi_{s,t}^{\ell} : \{0,1\}^E \to \{0,1\}$ be a boolean formula that evaluates whether there is a directed path from $s$ to $t$ of length at most $\ell$. We recursively define it as
        \[\varphi_{s,t}^1(x) = \begin{cases}
            1, & \text{if } s = t, \\
            x_{(s,t)}, & \text{otherwise},
        \end{cases} \qquad \text{and} \qquad \varphi_{s,t}^{\ell}(x) = \bigvee_{v \in V} \varphi_{s,v}^{\ell/2}(x) \land \varphi_{v,t}^{\ell/2}(x).\]
        The resulting formula $\varphi_{s,t}^n$ is of depth $\log(n)$, length $(2n)^{\log(n)}$, and all the input span programs are trivial. Now, using \cref{thm:implementation-graph-composition,thm:variable-time-formula-evaluation}, we obtain an implementation in the claimed number of queries, and time complexity
        \[\widetilde{O}\left(\sqrt{|J(\varphi_{s,t}^n)|} \cdot \left[1 + \log(n) \cdot \log|J(\varphi_{s,t}^n)|\right]\right) \subseteq \widetilde{O}\left(\sqrt{2n}^{\log(n)}\right).\]
        It remains to consider the space complexity. To that end, observe that in every level of the recursion, we have to keep track of which $v \in V$ we choose, and whether we check the existence of a path from $s$ to $v$ or from $v$ to $t$. Thus, we can embed our state space in $\H = (\C^V \otimes \C^2)^{\otimes \log(n)}$. Moreover, at each level of the recursion, all the initial states $\ket{w_0}$ are uniform superpositions over $\ket{v,b}$, for all $v \in V$ and $b \in \{0,1\}$. These can be implemented without the need for any data structure, cf.\ \cref{thm:uniform-state-preparation}, and so we don't need any QROM-overhead. Thus, the total space complexity is $O(\log(\dim(\H))) \subseteq O(\log^2(n))$.
    \end{proof}

    \subsection{Learning graphs}
    \label{subsec:learning-graphs}

    Belovs introduced the \textit{learning graph framework}, and used it to design a quantum algorithm for triangle finding~\cite{belovs2012span-learning-graphs}. This was later improved using a different learning graph in \cite{lee2017improved}, which was subsequently proven to be optimal in the learning graph framework~\cite{belovs2014power}. The learning graph framework was also used to find constant-sized subgraphs \cite{lee2012learning}. Subsequently, Belovs introduced a generalization, called the \textit{adaptive learning graph framework}, and used it to construct algorithms for the graph collision problem and the $k$-distinctness problem~\cite{belovs2012learning}. Adaptive learning graphs were recently also used to find hypergraph simplices \cite{yu2024quantum}. Finally, the framework was generalized again by Carette, Lauri\`ere and Magniez to the \textit{extended learning graph framework}, in \cite{carette2020extended}.

    Here, we prove that extended learning graphs can be turned into $st$-connectivity instances. To that end, we first recall the definition of extended learning graphs, as introduced in \cite{carette2020extended}.

    \begin{definition}[{Extended learning graphs~\cite[Definition~3]{carette2020extended}}]
        Let $n \in \N$, $\D \subseteq \{0,1\}^n$, and $f : \D \to \{0,1\}$. We need the following ingredients:
        \begin{enumerate}[nosep]
            \item A directed acyclic graph $G = (V,E)$.
            \item Labels $S : V \to 2^{[n]}$, such that:
            \begin{enumerate}[nosep]
                \item There is a unique $r(G) \in V$ with $S(r(G)) = \varnothing$.
                \item For all $(u,v) = e \in E$, there exists a $j \in [n]$ such that $j \not\in S(u)$ and $S(v) = S(u) \cup \{j\}$.
            \end{enumerate}
            \item A weight function $w : \D \times E \times \{0,1\} \to \R_{\geq 0}$, such that for all $(u,v) = e \in E$, $j \in [n]$ such that $S(v) = S(u) \cup \{j\}$ and $b \in \{0,1\}$:
            \begin{enumerate}[nosep]
                \item For all $x,y \in \D$, we have $x_{S(v)} = y_{S(v)}$ implies $w(x,e,b) = w(y,e,b)$.
                \item For all $(x,y) \in f^{-1}(1) \times f^{-1}(0)$, $x_{S(u)} = y_{S(u)}$ and $x_j \neq y_j$ implies $w(x,e,0) = w(y,e,1)$.
            \end{enumerate}
            \item Finally, for all $y \in f^{-1}(1)$, let $p_y : E \to \R_{\geq 0}$ be a unit flow on $G$, with:
            \begin{enumerate}[nosep]
                \item A single source node $r(G)$.
                \item The sink nodes are all the nodes $v \in V$ for which $S(v)$ contains a $1$-certificate for $y$.
                \item for all $e \in E$, $w(y,e,1) = 0$ implies $p_y(e) = 0$.
            \end{enumerate}
        \end{enumerate}
        Then $L = (G,S,w,\{p_y\}_{y \in f^{-1}(1)})$ is an \textit{adaptive learning graph} that computes $f$. For all $x \in f^{-1}(0)$ and $y \in f^{-1}(1)$, respectively, we define
        \[\ell_-(x,L) = \sum_{e \in E} w(x,e,0), \qquad \text{and} \qquad \ell_+(x,L) = \sum_{e \in E} \frac{p_y(e)^2}{w(y,e,1)},\]
        and we define
        \[\mathcal{L}_-(L) = \max_{x \in f^{-1}(0)} \ell_-(x,L), \qquad \mathcal{L}_+(L) = \max_{x \in f^{-1}(1)} \ell_+(x,L), \qquad \text{and} \qquad \mathcal{L}(L) = \sqrt{\mathcal{L}_-(L) \cdot \mathcal{L}_+(L)}.\]
        Finally, we write $\mathsf{XLG}(f)$ as the minimum of $\mathcal{L}(L)$ over all adaptive learning graphs $L$ that compute $f$.
    \end{definition}

    Now, we prove that all extended learning graphs are instantiations of the $st$-connectivity framework.

    \begin{theorem}
        \label{thm:extended-learning-graphs}
        Let $L$ be an extended learning graph. Then, we can construct an instance of the $st$-connectivity framework that evaluates the learning graph, with complexity at most $\mathcal{L}(L)$. As such, for all $f : \{0,1\}^n \supseteq \D \to \{0,1\}$, $\st(f) \leq \mathcal{L}(f)$.
    \end{theorem}

    \begin{proof}
        We write $L = (G,S,w,\{p_y\}_{y \in f^{-1}(1)})$, where $G = (V,E)$, $n \in \N$, $\D \subseteq \{0,1\}^n$ and $f : \D \to \{0,1\}$ is the function computed by $L$. Now, we construct the graph composition that evaluates $L$. To that end, let $\mathcal{P}_j$ be the trivial span program that evaluates the $j$th bit of an $n$-bit string.

        Next, we define an undirected graph $G' = (V',E')$, with
        \[V' = \{(v,z) : v \in V, z \in \{0,1\}^{S(v)}\}, \qquad \text{and} \qquad E' = \{\{(u,z),(v,z')\} : (u,v) \in E, z'|_{S(u)} = z\},\]
        where we use that $\{0,1\}^{\varnothing} = \{\varnothing\}$. Now, we let $s$ be the node $(\varnothing,\varnothing) \in V'$. We prune all the nodes $(v,z) \in V$ for which $z$ is a negative certificate for $f$ or for which there is no valid input $x \in \D$ that satisfies $z$. Furthermore, we contract all the nodes $(v,z) \in V'$ for which $z$ is a positive certificate for $f$ into a single node $t$, and we remove all self-loops that are created in the process. Finally, let $\{(u,z),(v,z')\} = e \in E'$, and let $j \in [n]$ be such that $S(v) = S(u) \cup \{j\}$. Now, take any input $x \in \D$ that satisfies $z'$. We label the edge with $\mathcal{P}_e = w(x,(u,v),1)^{-1}\mathcal{P}_j$ if $z'(j) = 1$, and $\mathcal{P}_e = w(x,(u,v),1)^{-1}(\lnot\mathcal{P}_j)$ if $z'(j) = 0$. We refer to the resulting graph composition as $\mathcal{P}$. See \cref{fig:learning-graph} for an example of this conversion.

        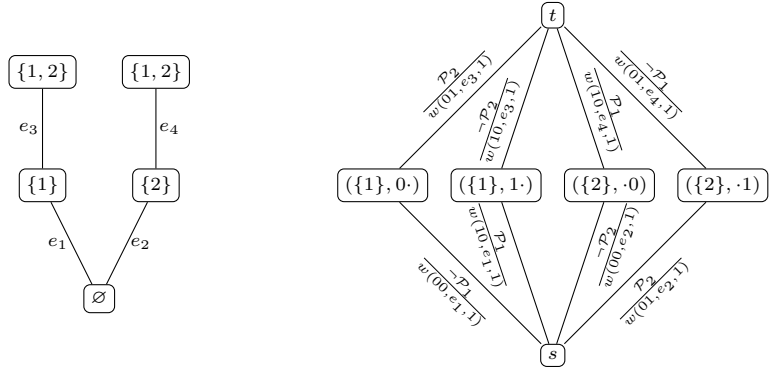
\begin{figure}[!ht]
            \centering\scriptsize
            \begin{tikzpicture}[vertex/.style = {draw, rounded corners = .3em}, scale = 1.5]
                \begin{scope}
                    \node[vertex] (s) at (0,0) {$\varnothing$};
                    \node[vertex] (1) at (-.5,1) {$\{1\}$};
                    \node[vertex] (2) at (.5,1) {$\{2\}$};
                    \node[vertex] (12) at (-.5,2) {$\{1,2\}$};
                    \node[vertex] (21) at (.5,2) {$\{1,2\}$};

                    \draw (s) to node[left=-.2em] {$e_1$} (1);
                    \draw (s) to node[right=-.2em] {$e_2$} (2);
                    \draw (1) to node[left=-.2em] {$e_3$} (12);
                    \draw (2) to node[right=-.2em] {$e_4$} (21);
                \end{scope}
                \begin{scope}[shift={(4,-.5)}]
                    \node[vertex] (s) at (0,0) {$s$};
                    \node[vertex] (10) at (-1.5,1.5) {$(\{1\},0\cdot)$};
                    \node[vertex] (11) at (-.5,1.5) {$(\{1\},1\cdot)$};
                    \node[vertex] (20) at (.5,1.5) {$(\{2\}, \cdot0)$};
                    \node[vertex] (21) at (1.5,1.5) {$(\{2\}, \cdot1)$};
                    \node[vertex] (t) at (0,3) {$t$};

                    \draw (s) to node[below, rotate = -45] {$\frac{\lnot\mathcal{P}_1}{w(00,e_1,1)}$} (10) to node[above, rotate = 45] {$\frac{\mathcal{P}_2}{w(01,e_3,1)}$} (t);
                    \draw (s) to node[below, near end, rotate = -{atan(3)}] {$\frac{\mathcal{P}_1}{w(10,e_1,1)}$} (11) to node[above, pos=.3, rotate = {atan(3)}] {$\frac{\lnot\mathcal{P}_2}{w(10,e_3,1)}$} (t);
                    \draw (s) to node[below, rotate = {atan(3)}, near end] {$\frac{\lnot\mathcal{P}_2}{w(00,e_2,1)}$} (20) to node[above, rotate = -{atan(3)}, pos=.3] {$\frac{\mathcal{P}_1}{w(10,e_4,1)}\hspace{1em}$} (t);
                    \draw (s) to node[below, rotate = 45] {$\frac{\mathcal{P}_2}{w(01,e_2,1)}$} (21) to node[above, rotate = -45] {$\frac{\lnot\mathcal{P}_1}{w(01,e_4,1)}$} (t);
                \end{scope}
            \end{tikzpicture}
            \caption{Conversion of an extended learning graph computing the parity function on $2$ bits to a graph composition.}
            \label{fig:learning-graph}
        \end{figure}

        It now remains to compute the witness sizes of this construction. To that end, suppose that $x \in \D$ is a positive input for the learning graph $L$. Let $p_x : E \to \R_{\geq 0}$ be the corresponding flow. We show how this flow naturally extends to a flow in the graph composition. To that end, let $(u,v) = e \in E$. We define $e_x = \{(u,z),(v,z')\}$, where $z$ and $z'$ are the partial assignments with support on $S(u)$ and $S(v)$ that agree with $x$. We let $p_x'(e_x) = p_x(e)$, and $p_x'(e') = 0$ for all other $e' \in E'$. This flow $p_x'$ is a unit $st$-flow in $G'$. We write $r_+(e') = w_+(x,\mathcal{P}_{e'})$, and upper bound the positive witness size by
        \[w_+(x,\mathcal{P}) \leq \min_{f \in F_{G',s,t}} \norm{\ket{f_{G,r^+}}}^2 \leq \sum_{e' \in E'} p_x'(e')^2w_+(x,\mathcal{P}_{e'}) = \sum_{e \in E} \frac{p_x(e)^2}{w(x,e,1)} = \ell_+(x,L).\]

        On the other hand, suppose that $x \in \D$ is a negative input to $L$. Now, we define a potential function $U_x : V' \to \R$ on $G'$, with $U_x((v,z)) = 1$ if $z$ agrees with $x$, and $0$ otherwise. We also take $U_x(t) = 0$, and we observe that $U_x(s) = 1$. Next, let $\{(u,z),(v,z')\} = e' \in E'$. We observe that $U_x((u,z)) - U_x((v,z'))$ is only non-zero if $z$ agrees with $x$ and $z'$ does not. For every edge $(u,v) = e \in E$, at most one such edge $\{(u,z),(v,z')\} \in E'$ exists. Moreover, we have that $x|_{S(u)} = z$ and $x_j \neq z'_j$, where $S(v) = S(u) \cup \{j\}$. Observe that there must exist an instance $y_{e'} \in f^{-1}(1)$ for which $z|_{S(v)} = z'$, because otherwise $e'$ would have been pruned earlier. Hence, for this pair $(x,y_{e'}) \in f^{-1}(0) \times f^{-1}(1)$, we have $x_{S(u)} = (y_{e'})_{S(u)}$ and $x_j \neq (y_{e'})_j$, and so $w(x,e,0) = w(y_{e'},e,1)$. Thus, we can write $r^-(e') = w_-(x,\mathcal{P}_{e'})^{-1} = w(y_{e'},e,1)^{-1} = w(x,e,0)^{-1}$, and upper bound the negative witness size by
        \[w_-(x,\mathcal{P}) \leq \norm{\ket{f_{G',U_x,r^-}}}^2 = \sum_{\{(u,z),(v,z')\} = e' \in E'} \frac{(U_x((u,z)) - U_x((v,z')))^2}{r^-(e')} \leq \sum_{e \in E} w(x,e,0) = \ell_-(x,L).\]

        We conclude that
        \[W_+(\mathcal{P}) = \max_{x \in f^{-1}(1)} w_+(x,\mathcal{P}) \leq \mathcal{L}_+(L), \qquad \text{and} \qquad W_-(\mathcal{P}) = \max_{x \in f^{-1}(0)} w_-(x,\mathcal{P}) \leq \mathcal{L}_-(L),\]
        and so $C(\mathcal{P}) \leq \mathcal{L}(L)$.
    \end{proof}

    \subsection{The bomb-testing model, and weighted decision trees}

    Another technique to develop quantum algorithms is through the conversion of classical decision trees into quantum algorithms. The first\footnote{Here, we mean ``first'' with respect to the quantum setting. In the classical literature, modeling decision trees as $st$-connectivity problems has long been considered~\cite{lee1959representation}.} such construction was introduced by Lin and Lin~\cite[Theorem~5.5]{lin2016upper}, in a framework that they referred to as the bomb-testing model. They constructed a quantum algorithm that evaluates a given decision tree with boolean queries using $O(\sqrt{\mathsf{GT}})$ queries, where $\mathsf{G}$ is the ``guessing complexity'' of the decision tree, and $\mathsf{T}$ is its depth. This construction was later generalized to the non-boolean case by Beigi and Taghavi~\cite[Theorem~4]{beigi2020quantum}. A time-efficient implementation was later provided in \cite{beigi2022time}. Finally, the construction in the boolean case was improved in \cite[Theorem~8]{cornelissen2022improved}.

    Here, we show that both constructions are a specific instantiation of the $st$-connectivity framework. We start by formalizing both as complexity measures.

    \begin{definition}[Decision tree + guessing algorithm complexity~{\cite[Sections~5.2 and 5.3]{lin2016upper}}]
        Let $n \in \N$, $\D \subseteq \{0,1\}^n$ and $f : \D \to \{0,1\}$. Let $T = (V,E)$ be a decision tree computing $f$, and let $G : E \to \{0,1\}$ be a guessing algorithm, such that for every internal node $v \in V$, exactly one of its children is connected with an edge $e \in E$ for which $G(e) = 1$. Then, we write $\mathsf{T}(T)$ for the depth of $T$, and the $\mathsf{G}(T)$ for the maximum number of edges $e \in E$ along a path from the root to a leaf that satisfy $G(e) = 0$. We call $\mathsf{G}(T)$ the guessing complexity of $T$. Finally, we write
        \[\sqrt{\mathsf{GT}}(f) = \min_{T \text{ computing } f} \sqrt{\mathsf{G}(T) \cdot \mathsf{T}(T)}.\]
    \end{definition}

    \begin{definition}[{Weighted-decision-tree complexity~\cite[Definition~1.5]{cornelissen2022improved}}]
        \label{def:weighted-decision-tree}
        Let $n \in \N$, $\D \subseteq \{0,1\}^n$ and $f : \D \to \{0,1\}$. Let $\T = (V,E)$ be a decision tree computing $f$, and $w : E \to \R_{\geq 0}$ a weight functions on the edges. For every input $x \in \D$, let $\mathcal{P}_x$ be the path taken on input $x$, and let $\overline{\mathcal{P}}_x$ be the set of edges that have exactly one node in common with the vertices traversed in $\mathcal{P}_x$. Then, we define
        \[\WDT_+(\T,w) = \max_{x \in f^{-1}(1)} \sum_{e \in \mathcal{P}_x} w_e, \qquad \WDT_-(\T,w) = \max_{x \in f^{-1}(0)} \sum_{e \in \overline{\mathcal{P}}_x} \frac{1}{w_e},\]
        and $\WDT(\T,w) = \sqrt{\WDT_+(\T,w) \cdot \WDT_-(\T,w)}$. Moreover, we let
        \[\WDT(\T) = \min_{w : E \to \R_{\geq0}} \WDT(\T,w), \qquad \text{and} \qquad \WDT(f) = \min_{\T \text{ computing } f} \WDT(\T).\]
    \end{definition}

    We can also phrase these complexity measures in the zero-error setting, as follows.

    \begin{definition}[Zero-error version of the complexity measures]
        \label{def:zero-error}
        Let $n \in \N$, $\D \subseteq \{0,1\}^n$, and $f : \D \to \{0,1\}$. Let $\mathcal{T} = \{T_1, \dots, T_k\}$ be a set of decision trees, where every leaf is labeled with $0$, $1$, or \texttt{?}. Let $ p= (p_1, \dots, p_k)$ be a probability distribution on $[k]$, such that for all $x \in \D$, if we sample $T_j$ with $j \sim p$, then the probability that it outputs $\lnot f(x)$ is $0$, and the probability that it outputs \texttt{?} is at most $1/2$. Then, we say that $\mathcal{T}$ computes $f$ with zero error, and we write
        \[\sqrt{\mathsf{GT}}(\mathcal{T}) = \max_{j \in [k]} \sqrt{\mathsf{G}(T_j) \cdot \mathsf{T}(T_j)}, \qquad \text{and} \qquad \sqrt{\mathsf{GT}}_0(f) = \min_{\mathcal{T} \text{ computes } f \text{ with zero error}} \sqrt{\mathsf{GT}}(\mathcal{T}).\]
        and
        \[\WDT(\mathcal{T}) = \max_{j \in [k]} \WDT(T_j), \qquad \text{and} \qquad \WDT_0(f) = \min_{\mathcal{T} \text{ computes } f \text{ with zero error}} \WDT(\mathcal{T}).\]
    \end{definition}

    It is immediate from the definitions above that $\mathsf{WDT}_0(f) \leq \mathsf{WDT}(f) \leq \sqrt{\mathsf{GT}}(f) \leq \mathsf{D}(f)$, and similarly that $\mathsf{WDT}_0(f) \leq \sqrt{\mathsf{GT}}_0(f) \leq \sqrt{\mathsf{GT}}(f)$, for all possibly partial functions $f : \D \to \{0,1\}$. We prove here that the smallest of all these measures, i.e., $\mathsf{WDT}_0(f)$, is lower bounded by the $st$-connectivity complexity of $f$. We start by showing that we can turn any weighted decision tree in an $st$-connectivity instance, in~\cref{thm:weighted-decision-trees}, and then we use the graph composition framework to compose the family of decision trees together, in~\cref{thm:st-to-WDT0}.

    \begin{theorem}
        \label{thm:weighted-decision-trees}
        Let $\T = (V,E)$ be a decision tree, making queries to an $n$-bit string, having outputs in $\{0,1,\texttt{?}\}$, and with a weight function $w : E \to \R_{>0}$ on the edges. Then, we can construct an instance of the $st$-connectivity framework that distinguishes whether the output of the tree is $1$ or an element of $\{0,\texttt{?}\}$ with complexity at most $\WDT(\T,w)$.
    \end{theorem}

    \begin{proof}
        Let $\mathcal{P}_j$ be the trivial span program that evaluates the $j$th bit of the bitstring. We construct our graph composition by modifying the decision tree as follows. First, for every non-leaf node in the decision tree, querying the $j$th bit of the bitstring, we label its outgoing legs by $w_e\mathcal{P}_j$ for the outcome-$1$ leg, and $w_e(\lnot\mathcal{P}_j)$ for the outcome-$0$ leg. Next, we label the root vertex of the decision tree by $s$, we prune all the $0$-leaves and \texttt{?}-leaves of the tree, and we contract all the $1$-leaves into a single vertex labeled by $t$. An example of this conversion is shown in \cref{fig:decision-tree-conversion}.

        \begin{figure}[!ht]
            \centering
            \begin{tikzpicture}[vertex/.style = {draw, rounded corners = .3em}, yscale=-1]
                \begin{scope}
                    \node[vertex] (1) at (0,0) {$x_1$};
                    \node[vertex] (2) at (1,1) {$x_2$};
                    \node[vertex] (3) at (2,2) {$x_3$};

                    \draw (1) to node[above left=-.2em, near start] {$1$} node[below right=-.2em] {$w_1$} (-1,1) node[below left=-.2em] {$1$};
                    \draw (2) to node[above left=-.2em, near start] {$1$} node[below right=-.2em] {$w_3$} (0,2) node[below left=-.2em] {$1$};
                    \draw (3) to node[above left=-.2em, near start] {$1$} node[below right=-.2em] {$w_5$} (1,3) node[below left=-.2em] {$1$};
                    \draw (1) to node[above right=-.2em, near start] {$0$} node[below left=-.2em] {$w_2$} (2) to node[above right=-.2em, near start] {$0$} node[below left=-.2em] {$w_4$} (3) to node[above right=-.2em, near start] {$0$} node[below left=-.2em] {$w_6$} (3,3) node[below right=-.2em] {$0$};
                \end{scope}

                \node at (4,1.5) {$\mapsto$};

                \begin{scope}[shift={(6,0)}]
                    \node[vertex] (s) at (0,0) {$s$};
                    \node[vertex] (1) at (1,1) {};
                    \node[vertex] (2) at (2,2) {};
                    \node[vertex] (t) at (0,3) {$t$};

                    \draw (s) to node[above right=-.2em] {$w_2(\lnot\mathcal{P}_1)$} (1) to node[above right=-.2em] {$w_4(\lnot\mathcal{P}_2)$} (2);
                    \draw (s) to node[left=-.2em] {$w_1\mathcal{P}_1$} (t);
                    \draw (1) to node[right=-.2em] {$w_3\mathcal{P}_2$} (t);
                    \draw (2) to node[below right=-.2em] {$w_5\mathcal{P}_3$} (t);
                \end{scope}
            \end{tikzpicture}
            \caption{An example of the conversion of a weighted decision tree into a graph composition.}
            \label{fig:decision-tree-conversion}
        \end{figure}
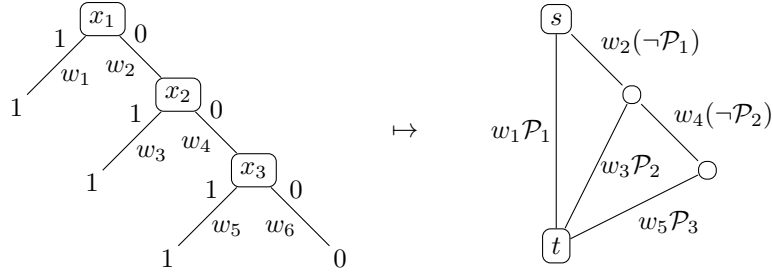

        It remains to compute the witness sizes of the resulting graph composition. To that end, suppose that $x$ is a positive input. Then, there is a unique path $\mathcal{P}_x \subseteq E$ that connects $s$ and $t$, corresponding to the path followed in the decision tree, and as such the positive witness size is $w_+(x,\mathcal{P}) = \sum_{e \in \mathcal{P}_x} w_e$.

        On the other hand, suppose that $x$ is a negative input. Then, for every non-leaf vertex in the decision tree that $x$'s path traverses, we consider the set $\overline{\mathcal{P}}_x$ of legs of these vertices that the path is not taking. We now observe that $\overline{\mathcal{P}}_x$ is a cut in the graph, and thus from \cref{thm:witness-upper-bounds}, we find that $w_-(x,\mathcal{P}) \leq \sum_{e \in \overline{\mathcal{P}}_x} \frac{1}{w_e}$.

        Combining these observations yields that
        \[W_+(\mathcal{P}) = \max_{x \in f^{-1}(1)} w_+(x,\mathcal{P}) = \WDT_+(\T,w), \qquad \text{and} \qquad W_-(\mathcal{P}) = \max_{x \in f^{-1}(0)} w_-(x,\mathcal{P}) \leq \WDT_-(\T,w),\]
        and so $C(\mathcal{P}) \leq \WDT(\T,w)$.
    \end{proof}

    The above theorem immediately proves that $\st(f) \leq \WDT(f)$. However, we can do slightly better and prove that $\st(f) \leq \WDT_0(f)$. This is the objective of the following theorem.

    \begin{theorem}
        \label{thm:st-to-WDT0}
        Let $f : \D \to \{0,1\}$ with $\D \subseteq \{0,1\}^n$. We have $\st(f) \leq \sqrt{2}\WDT_0(f)$.
    \end{theorem}

    \begin{proof}
        Let $\mathcal{T} = \{T_1, \dots, T_k\}$ be a family of decision trees that computes $f$ with zero-error, and that minimizes $\WDT(\mathcal{T})$. Then, we find that $\WDT(\mathcal{T}) = \WDT_0(f)$. Next, for all $j \in [k]$, we let $\mathcal{P}_j$ be the $st$-connectivity that is constructed from $T_j$ by \cref{thm:weighted-decision-trees}. We observe that $C(\mathcal{P}_j) = \WDT(T_j) \leq \WDT_0(f)$. Without loss of generality, we assume that $W_-(\mathcal{P}_j) = C(\mathcal{P}_j)$ and $W_+(\mathcal{P}_k) = C(\mathcal{P}_j)$, since we can always rescale the resistances accordingly.

        Next, we consider a parallel graph composition $\mathcal{P}$ with span programs $(\mathcal{P}_j/p_j)_{j=1}^k$ on the edges. See also \cref{fig:st-to-WDT}. We compute the witness sizes of the resulting instance of the $st$-connectivity framework.

        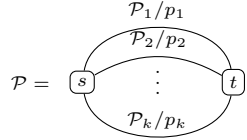
\begin{figure}[!ht]
            \centering\scriptsize
            \begin{tikzpicture}[vertex/.style = {draw, rounded corners = .3em}, scale = 2]
                \node[vertex] (s) at (0,0) {$s$};
                \node[left] at (s.west) {$\mathcal{P} = \;\;$};
                \node[vertex] (t) at (1,0) {$t$};
                \draw (s) to[bend left=80] node[above] {$\mathcal{P}_1/p_1$} (t);
                \draw (s) to[bend left=30] node[above] {$\mathcal{P}_2/p_2$} (t);
                \node at (.5,.05) {$\vdots$};
                \draw (s) to[bend right=80] node[above] {$\mathcal{P}_k/p_k$} (t);
            \end{tikzpicture}
            \caption{The reduction from any zero-error randomized algorithm to an $st$-connectivity instance.}
            \label{fig:st-to-WDT}
        \end{figure}

        To that end, let $x \in \{0,1\}^n$ be a positive instance, and let $J(x) \subseteq [k]$ be the subset of indices such that $j \in J(x)$ if and only if $T_j$ accepts $x$. Since this happens with probability at least $1/2$, we have $\sum_{j \in J(x)} p_j \geq 1/2$. Thus, we find that
        \[w_+(x,\mathcal{P}) = \frac{1}{\sum_{j \in J(x)} \frac{p_j}{w_+(x,\mathcal{P}_j)}} \leq \max_{j \in J(x)} C(\mathcal{P}_j) \cdot \left[\sum_{j \in J(x)} p_j\right]^{-1} \leq 2\WDT_0(f).\]

        On the other hand, let $x \in \{0,1\}^n$ be a negative instance. Then, all trees in $\mathcal{T}$ do not accept $x$, and so we obtain that
        \[w_-(x,\mathcal{P}) = \sum_{j=1}^k w_-\left(x,\frac{\mathcal{P}_j}{p_j}\right) = \sum_{j=1}^k p_jw_-(x,\mathcal{P}_j) \leq \WDT_0(f) \cdot \sum_{j=1}^k p_j = \WDT_0(f).\]

        We conclude that $\st(f) \leq C(\mathcal{P}) \leq \sqrt{2} \cdot \WDT_0(f)$.
    \end{proof}

    \subsection{Phase estimation algorithms and multidimensional quantum walks}

    We relate the graph composition framework to the framework of multidimensional walks. The latter framework was introduced in \cite{jeffery2025multidimensional}, and further formalized in \cite{jeffery2024multidimensional}.

    First, Jeffery and Zur define the notion of a phase estimation algorithm in~\cite[Definition~3.1]{jeffery2025multidimensional}. It was shown by Reichardt~\cite{reichardt2009span} that span programs can always be evaluated using the phase estimation algorithm. We reprove that result in the span program notation used here, and the phase estimation algorithm notation used in \cite{jeffery2025multidimensional}.

    \begin{theorem}
        \label{thm:phase-estimation-algo}
        Every span program, evaluated on an input $x$, is also a phase estimation algorithm, and the span program's positive and negative witnesses are negative and positive witnesses of the phase estimation algorithm, respectively.
    \end{theorem}

    \begin{proof}
        Let $\mathcal{P} = (\H, x \mapsto \H(x), \K, \ket{w_0})$ be a span program. Fix $x \in \D$, and let $\Psi_A$ and $\Psi_B$ be bases for $\H(x)$ and $\K$, respectively. Then, $(\H, \Psi_A, \Psi_B, \ket{w_0})$ defines a phase estimation algorithm according to \cite[Definition~3.1]{jeffery2025multidimensional}. It remains to check the witness sizes.

        To that end, according to \cref{def:span-program}, a negative witness for $x$ is a vector in $\K^{\perp} \cap \H(x)^{\perp} = (\Span(\Psi_A) + \Span(\Psi_B))^{\perp}$, which fits in the definition of a positive witness in \cite[Definition~3.5]{jeffery2025multidimensional}. On the other hand, according to \cref{def:span-program}, a positive witness for $x$ is a vector in $\ket{w_x} \in \H(x)$, such that $\ket{k} := \ket{w_x} - \ket{w_0} \in \K$. This is equivalent of finding a pair of vectors in $\Span(\Psi_A)$ and $\Span(\Psi_B)$, such that $\ket{w_0}$ is the sum of the two, which fits the definition of a negative witness in \cite[Definition~3.2]{jeffery2025multidimensional}.
    \end{proof}

    Jeffery and Zur then go on to evaluate the welded-trees problem and the $k$-distinctness problem using phase estimation algorithms. The algorithm for welded trees inherently relies on the alternative-neighborhood technique presented in \cite[Section~3.2.2]{jeffery2025multidimensional}, which seems not to be available in the graph composition framework. For the $k$-distinctness problem, we leave recovering a time-efficient algorithm using the graph composition framework for future work.

    Recently, Jeffery and Pass defined a formal version of the multidimensional quantum walk, referred to as a subspace graph, in~\cite[Definition~3.1]{jeffery2024multidimensional}. We show that every graph composition is also a subspace graph.

    \begin{theorem}
        \label{thm:subspace-graph}
        Every graph composition can be turned into a subspace graph with at most the same witness sizes.
    \end{theorem}

    \begin{proof}
        We embed span programs in subspace graphs in a similar way to the switching networks, as presented in \cite[Section~3.3]{jeffery2024multidimensional}. Let $G = (V,E)$ with $s,t \in V$ connected in $G$, and $s \neq t$, and let $\mathcal{P}$ be the graph composition of $G$ with span programs $(\mathcal{P}_e)_{e \in E}$ on $\D$. For all $e \in E$, we write $\lnot\mathcal{P}_e =: \mathcal{P}_e' = (\H_e', x \mapsto \H_e'(x), \K_e', \ket{(w_0^e)'})$, and we define
        \[\Xi_e = (\H_e' \cap \Span\{\ket{(w_0^e)'}\}^{\perp}) \oplus \Span\{\ket{(w_0^e)',\rightarrow},\ket{(w_0^e)',\leftarrow}\},\]
        and we define the isometric embedding $\mathcal{Z}_e : \H_e' \mapsto \Xi_e$ that acts as identity on $\H_e' \cap \Span\{\ket{(w_0^e)'}\}$, and as $\mathcal{Z}_e : \ket{(w_0^e)'} \mapsto (\ket{(w_0^e)',\rightarrow} - \ket{(w_0^e)',\leftarrow})/\sqrt{2}$. Inspired by \cite[Definition~3.5]{jeffery2024multidimensional}, we define
        \[\Xi_e^{\mathcal{A}} = \mathcal{Z}_e(\H_e'(x)), \qquad \text{and} \qquad \Xi_e^{\mathcal{B}} = \K_e' \oplus \Span\{\ket{(w_0^e)',\rightarrow} + \ket{(w_0^e)',\leftarrow}\}.\]
        Next, we take $B = \{s,t\}$, and we define
        \begin{align*}
            \Xi_s &= \Span\{\ket{s},\ket{s,\leftarrow}\}, \qquad \Xi_s^{\mathcal{A}} = \Span\{\ket{s} + \ket{s,\leftarrow}\}, \qquad \text{and} \qquad \Xi_s^{\mathcal{B}} = \{0\}, \\
            \Xi_t &= \Span\{\ket{t},\ket{t,\rightarrow}\}, \qquad \Xi_t^{\mathcal{A}} = \Span\{\ket{t} + \ket{t,\rightarrow}\}, \qquad \text{and} \qquad \Xi_t^{\mathcal{B}} = \{0\},
        \end{align*}
        and $\mathcal{V}_B = \Span\{\ket{s,\leftarrow} + \ket{t,\rightarrow}\}$. Next, in accordance with \cite[Definition~3.3]{jeffery2024multidimensional}, we define for all $v \in V$,
        \[\mathcal{V}_v = \begin{cases}
            \Span\left\{\ket{\psi_*(s)} + \ket{s,\leftarrow}\right\}, & \text{if } v = s, \\
            \Span\left\{\ket{\psi_*(t)} + \ket{t,\rightarrow}\right\}, & \text{if } v = t, \\
            \Span\left\{\ket{\psi_*(v)}\right\}, & \text{otherwise},
        \end{cases} \quad \text{where} \quad \ket{\psi_*(v)} = \sum_{e \in N_+(v)} \frac{\ket{(w_0^e)',\rightarrow}}{\norm{\ket{(w_0^e)'}}^2} + \sum_{e \in N_-(v)} \frac{\ket{(w_0^e)',\leftarrow}}{\norm{\ket{(w_0^e)'}}^2}.\]
        This defines a subspace graph, and we take the initial vector $\ket{\psi_0} = (\ket{s}-\ket{t})/\sqrt{2}$. We now observe that if $x \in \D$ is a positive input for $\mathcal{P}_e'$, then we can find vectors $\ket{(w_x^e)'} \in \H_e'(x)$ and $\ket{k_e'} \in \K_e'$ such that $\ket{(w_0^e)'} = \ket{(w_x^e)'} + \ket{k_e'}$. Thus, we find that
        \[\frac{\ket{(w_0^e)',\rightarrow} - \ket{(w_0^e)',\leftarrow}}{\sqrt{2}} = \mathcal{Z}_e(\ket{(w_0^e)'}) = \underbrace{\mathcal{Z}_e(\ket{(w_x^e)'})}_{\in \mathcal{A}_G} + \underbrace{\ket{k_e'}}_{\in \mathcal{B}_G},\]
        and thus we obtain that $(\mathcal{A}_G + \mathcal{B}_G)^{\perp} \cap \Xi_e = \{0\}$. Thus, intuitively, if $x \in \D$ is positive for $\mathcal{P}_e' = \lnot\mathcal{P}_e$, i.e., if $x \in \D$ is negative for $\mathcal{P}_e$, it turns off the edge $e \in E$ and bars any flow from flowing through it. Hence, analogously as in \cite[Section~3.3]{jeffery2024multidimensional}, the subspace graph checks whether there exists an $st$-path along edges $e \in E$ for which $x$ is a positive input to $\mathcal{P}_e$, and the witness analysis follows directly from the analysis in \cite[Section~3.3]{jeffery2024multidimensional}.
    \end{proof}

    Note that the definition of all the subspaces in the above proof do not make any reference to the circulation space. However, recall that the space $\mathcal{B}_G$ is defined on \cite[Page~8]{jeffery2024multidimensional} as
    \[\mathcal{B}_G = \left[\bigoplus_{e \in E} \Xi_e^{\mathcal{B}} \oplus \bigoplus_{u \in B} \Xi_u^{\mathcal{B}}\right] + \mathcal{V}_B + \bigoplus_{v \in V} \mathcal{V}_v,\]
    and that these constituent spaces are not necessarily orthogonal to one another. Thus, reflecting through $\mathcal{B}_G$ still requires some non-trivial amount of work. In the setting considered in \cite[Section~3.3]{jeffery2024multidimensional}, it seems to be similar in hardness to reflecting through the circulation space, so the time-complexity considerations of implementing the graph composition framework, in \cref{thm:implementation-graph-composition,thm:circulation-space-reflection-implementation}, can be used to understand the time complexity of implementing subspace graphs too.

    Finally, since the graph composition framework outputs a span program, we can use \cref{thm:feasible-sln-adv-bound} to generate a feasible solution to the dual adversary bound in a black-box manner. This sets the graph composition framework apart from the multidimensional quantum walk framework. We leave it for future work to figure out whether a similar conversion technique exists that turns instances of the multidimensional quantum walk framework into feasible solutions to the dual adversary bound.

    \subsection{Randomized algorithms}

    In this section, we prove that every bounded-error randomized query algorithm can be turned into an instance of the $st$-connectivity framework, with the same complexity up to a multiplicative constant. To that end, we first introduce an optimal $st$-connectivity construction for the threshold function. This recovers the results in \cite[Proposition~A.4, Claim~A.5]{reichardt2009span} and \cite[Proposition~3.32]{belovs2014applications}.

    \begin{theorem}[Threshold functions through the $st$-connectivity framework]
        \label{thm:threshold}
        Let $n \in \N$, and $k \in \{1, \dots, n\}$. Let $f_n^k$ be the threshold function on $n$ bits with threshold $k$, i.e., $f_n^k(x) = 1$ if and only if $|x| \geq k$. Let $x_j$ be a trivial span program computing the $j$th bit of the input. Then, for all $\S \subseteq [n]$, we recursively define
        \[\Th_\S^{k+1} = \bigvee_{j \in \S} (x_j \land k\Th_{\S \setminus \{j\}}^k), \qquad \text{and} \qquad \Th_\S^1 = \bigvee_{j \in \S} x_j.\]
        Then, $\Th_n^k := \Th_{[n]}^k$ computes $f_n^k$, with witness sizes
        \[w_+(x,\Th_n^k) = \begin{cases}
            \frac{1}{|x|-k+1}, & \text{if } |x| \geq k, \\
            \infty, & \text{otherwise},
        \end{cases} \qquad \text{and} \qquad w_-(x,\Th_n^k) = \begin{cases}
            \infty, & \text{if } |x| \geq k, \\
            \frac{k(n-k+1)}{k-|x|}, & \text{otherwise}.
        \end{cases}\]
        Consequently, $C(\Th_n^k) = \sqrt{k(n+k-1)}$, and this is optimal.
    \end{theorem}

    \begin{proof}
        The optimality was already proven in \cite[Proposition~3.32]{belovs2014applications}, so it remains to compute the witness sizes. We use \cref{thm:and-or-witness-sizes}, and perform induction on the size of $\S$. Suppose that $x \in \{0,1\}^n$, with $|x| \geq k+1$. Then,
        \begin{align*}
            w_+(x,\Th_\S^{k+1}) &= \left[\sum_{j \in \S} w_+\left(x, x_j \land k\Th_{\S \setminus \{j\}}^k\right)^{-1}\right]^{-1} = \left[\sum_{\substack{j \in \S \\ x_j = 1}} \left(1 + kw_+\left(x, \Th_{\S \setminus \{j\}}^k\right)\right)^{-1}\right]^{-1} \\
            &= \left[\sum_{\substack{j \in \S \\ x_j = 1}} \left(1 + \frac{k}{|x|-k}\right)^{-1}\right]^{-1} = \left[|x| \cdot \frac{|x|-k}{|x|}\right]^{-1} = \frac{1}{|x|-k}.
        \end{align*}
        On the other hand, suppose that $|x| < k+1$. Then, we have
        \begin{align*}
            w_-(x,\Th_\S^{k+1}) &= \sum_{j \in \S} w_-\left(x, x_j \land k\Th_{\S \setminus \{j\}}^k\right) = \sum_{j \in \S} \left[w_-(x,x_j)^{-1} + kw_-\left(x,\Th_{\S \setminus \{j\}}^k\right)^{-1}\right]^{-1} \\
            &= \frac1k\sum_{\substack{j \in \S \\ x_j = 1}} w_-\left(x,\Th_{\S \setminus \{j\}}^k\right) + \sum_{\substack{j \in \S \\ x_j = 0}} \left[1 + kw_-\left(x, \Th_{\S \setminus \{j\}}^k\right)^{-1}\right]^{-1} \\
            &= \frac{|x|}{k} \cdot \frac{k(n-k)}{k-(|x|-1)} + \frac{n-|x|}{1 + k \cdot \frac{k-|x|}{k(n-k)}} = \frac{|x|(n-k)}{k-|x|+1} + \frac{(n-|x|)(n-k)}{n-k+k-|x|} \\
            &= (n-k)\left[\frac{|x|}{k+1-|x|} + 1\right] = \frac{(n-k)(k+1)}{k+1-|x|}.\qedhere
        \end{align*}
    \end{proof}

    We provide a pictorial representation of the resulting graph composition construction for $n = 4$ and $k = 3$ in \cref{fig:threshold}.

    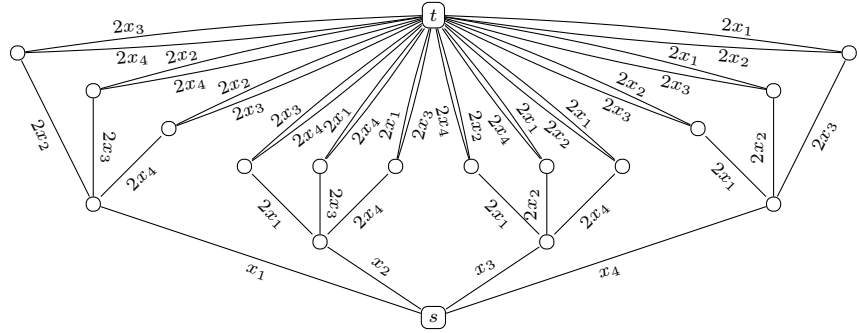
\begin{figure}[!ht]
        \centering\scriptsize
        \begin{tikzpicture}[vertex/.style = {draw, rounded corners = .3em}]
            \node[vertex] (s) at (0,0) {$s$};
            \node[vertex] (1) at (-4.5,1.5) {};
            \node[vertex] (2) at (-1.5,1) {};
            \node[vertex] (3) at (1.5,1) {};
            \node[vertex] (4) at (4.5,1.5) {};
            \node[vertex] (12) at (-5.5,3.5) {};
            \node[vertex] (13) at (-4.5,3) {};
            \node[vertex] (14) at (-3.5,2.5) {};
            \node[vertex] (21) at (-2.5,2) {};
            \node[vertex] (23) at (-1.5,2) {};
            \node[vertex] (24) at (-.5,2) {};
            \node[vertex] (31) at (.5,2) {};
            \node[vertex] (32) at (1.5,2) {};
            \node[vertex] (34) at (2.5,2) {};
            \node[vertex] (41) at (3.5,2.5) {};
            \node[vertex] (42) at (4.5,3) {};
            \node[vertex] (43) at (5.5,3.5) {};
            \node[vertex] (t) at (0,4) {$t$};

            \draw (s) to node[below, rotate = -{atan(1/3)}] {$x_1$} (1);
            \draw (s) to node[above, rotate = -{atan(2/3)}] {$x_2$} (2);
            \draw (s) to node[above, rotate = {atan(2/3)}] {$x_3$} (3);
            \draw (s) to node[below, rotate = {atan(1/3)}] {$x_4$} (4);

            \draw (1) to node[below, rotate = -{atan(2)}] {$2x_2$} (12);
            \draw (1) to node[above, rotate = -90] {$2x_3$} (13);
            \draw (1) to node[below, rotate = 45] {$2x_4$} (14);
            \draw (2) to node[below, rotate = -45] {$2x_1$} (21);
            \draw (2) to node[above, rotate = -90] {$2x_3$} (23);
            \draw (2) to node[below, rotate = 45] {$2x_4$} (24);
            \draw (3) to node[below, rotate = -45] {$2x_1$} (31);
            \draw (3) to node[above, rotate = 90] {$2x_2$} (32);
            \draw (3) to node[below, rotate = 45] {$2x_4$} (34);
            \draw (4) to node[below, rotate = -45] {$2x_1$} (41);
            \draw (4) to node[above, rotate = 90] {$2x_2$} (42);
            \draw (4) to node[below, rotate = {atan(2)}] {$2x_3$} (43);

            \draw (12) to[bend left=3] node[above=-.2em, rotate = {atan(1/11)}, near start] {$2x_3$} (t);
            \draw (12) to[bend right=3] node[below=-.2em, rotate = {atan(1/11)}, near start] {$2x_4$} (t);
            \draw (13) to[bend left=3] node[above=-.2em, rotate = {atan(2/9)}, near start] {$2x_2$} (t);
            \draw (13) to[bend right=3] node[below=-.2em, rotate = {atan(2/9)}, near start] {$2x_4$} (t);
            \draw (14) to[bend left=3] node[above=-.2em, rotate = {atan(3/7)}, near start] {$2x_2$} (t);
            \draw (14) to[bend right=3] node[below=-.2em, rotate = {atan(3/7)}, near start] {$2x_3$} (t);
            \draw (21) to[bend left=3] node[above=-.2em, rotate = {atan(4/5)}, near start] {$2x_3$} (t);
            \draw (21) to[bend right=3] node[below=-.2em, rotate = {atan(4/5)}, near start] {$2x_4$} (t);
            \draw (23) to[bend left=3] node[above=-.2em, rotate = {atan(4/3)}, near start] {$2x_1$} (t);
            \draw (23) to[bend right=3] node[below=-.2em, rotate = {atan(4/3)}, near start] {$2x_4$} (t);
            \draw (24) to[bend left=3] node[above=-.2em, rotate = {atan(4)}, near start] {$2x_1$} (t);
            \draw (24) to[bend right=3] node[below=-.2em, rotate = {atan(4)}, near start] {$2x_3$} (t);
            \draw (31) to[bend left=3] node[below=-.3em, rotate = -{atan(4)}, near start] {$2x_4$} (t);
            \draw (31) to[bend right=3] node[above=-.3em, rotate = -{atan(4)}, near start] {$2x_2$} (t);
            \draw (32) to[bend left=3] node[below=-.2em, rotate = -{atan(4/3)}, near start] {$2x_4$} (t);
            \draw (32) to[bend right=3] node[above=-.2em, rotate = -{atan(4/3)}, near start] {$2x_1$} (t);
            \draw (34) to[bend left=3] node[below=-.2em, rotate = -{atan(4/5)}, near start] {$2x_2$} (t);
            \draw (34) to[bend right=3] node[above=-.2em, rotate = -{atan(4/5)}, near start] {$2x_1$} (t);
            \draw (41) to[bend left=3] node[below=-.2em, rotate = -{atan(3/7)}, near start] {$2x_3$} (t);
            \draw (41) to[bend right=3] node[above=-.2em, rotate = -{atan(3/7)}, near start] {$2x_2$} (t);
            \draw (42) to[bend left=3] node[below=-.2em, rotate = -{atan(2/9)}, near start] {$2x_3$} (t);
            \draw (42) to[bend right=3] node[above=-.2em, rotate = -{atan(2/9)}, near start] {$2x_1$} (t);
            \draw (43) to[bend left=3] node[below=-.2em, rotate = -{atan(1/11)}, near start] {$2x_2$} (t);
            \draw (43) to[bend right=3] node[above=-.2em, rotate = -{atan(1/11)}, near start] {$2x_1$} (t);
        \end{tikzpicture}
        \caption{Graph composition for the threshold function on $4$ bits, with threshold $3$.}
        \label{fig:threshold}
    \end{figure}

    We similarly find an optimal construction for the exact-weight function, which matches the result from \cite[Propositions~A.6 and A.7]{reichardt2009span}.

    \begin{theorem}[Graph-composition for the exact-weight function]
        \label{thm:exact-weight}
        Let $n \in \N$, $k \in \{1, \dots, n-1\}$, and let $f_n^k$ be the exact-weight function on $n$ bits with weight $k$, i.e., $f_n^k(x) = 1$ if and only if $|x| = k$. We let
        \[\EW_n^k = k(n+k-1)\Th_n^k \land \lnot \Th_n^{k+1}.\]
        Then, $\EW_n^k$ computes $f_n^k$, with witness sizes
        \[w_+(x,\EW_n^k) = \begin{cases}
            n + 2k(n-k), & \text{if } |x| = k, \\
            \infty, & \text{otherwise},
        \end{cases} \qquad \text{and} \qquad w_-(x,\EW_n^k) = \begin{cases}
            \infty, & \text{if } |x| = k, \\
            \frac{1}{|k-|x||}, & \text{otherwise}.
        \end{cases}\]
        Consequently, $C(\EW_n^k) = \sqrt{n + 2k(n-k)}$, and this is optimal.
    \end{theorem}

    \begin{proof}
        The lower bound is shown in \cite[Proposition~A.7]{reichardt2009span}, so it remains to compute the witness sizes. If $|x| = k$, then
        \[w_+(x,\EW_n^k) = k(n+k-1)w_+(x,\Th_n^k) + w_-(x,\Th_n^{k+1}) = k(n-k+1) + (n-k)(k+1) = n + 2k(n-k).\]
        If $|x| < k$, we find
        \[w_-(x,\EW_n^k) = w_-(x,k(n+k-1)\Th_n^k) = \frac{1}{k - |x|},\]
        and if $|x| > k$, we have
        \[w_-(x,\EW_n^k) = w_+(x,\Th_n^{k+1}) = \frac{1}{|x|-k}.\qedhere\]
    \end{proof}

    It is possible to expand on these techniques to derive graph composition construction for some other symmetric functions that are exactly tight as well. We remark without proof that the threshold-not-all function, i.e., that the function on $n$ bits that accepts if and only if $k \leq |x| \leq n-1$, can be optimally graph-constructed by $\Th_n^k \land \lnot k(n-k+1)\Th_n^n$, with complexity $\sqrt{k(n-k+1) + n/(n-k)}$. We also note that the adversary value for the parity function on $n$ bits is exactly $n$. The proofs follow along very similar lines as for the exact-weight function.

    Next, we use the above construction for the threshold function that $\st$ can be much smaller than $\mathsf{R}_0$, i.e., zero-error randomized query complexity. In particular we show that $\st(f) = O(1)$ whereas $\mathsf{R}_0(f) \in \Omega(n)$, when $f$ is the gapped majority function.

    \begin{lemma}
        \label{lem:gapped-majority}
        Let $n \in \N$, $\D = \{x \in \{0,1\}^n : |x| < n/3 \lor |x| > 2n/3\}$, and $f_n : \D \to \{0,1\}$ be defined as $f_n(x) = 1$ if and only if $|x| > 2n/3$. That is, $f_n$ is the gapped majority function on $n$ inputs. Then, $\mathsf{st}(f_n) \in O(1)$, and $\mathsf{R}_0(f_n) \in \Omega(n)$.
    \end{lemma}

    \begin{proof}
        It is clear that $\mathsf{R}_0(f) \in \Omega(n)$, since we need to query at least $n/3$ bits to be convinced of either outcome. On the other hand, we use the construction from \cref{thm:threshold} with $k = n/2$. Then, we find that
        \[W_+ = \max_{\substack{x \in \{0,1\}^n \\ |x| > 2n/3}} \frac{1}{|x| - \frac{n}{2} + 1} < \frac{6}{n} \in O\left(\frac1n\right), \quad \text{and} \quad W_- = \max_{\substack{x \in \{0,1\}^n \\ |x| < n/3}} \frac{\frac{n}{2}(n-\frac{n}{2}+1)}{\frac{n}{2} - |x|} < 3\left(\frac{n}{2}+1\right) \in O(n),\]
        from which we deduce that $\st(f) \in O(1)$.
    \end{proof}

    The observation that the $\st$-complexity of the gapped-majority function is constant has profound consequences. Indeed, we can use it to show that $\st$ is upper bounded by $\mathsf{R}$, which is the objective of the following theorem.

    \begin{theorem}
        \label{thm:st-vs-R}
        Let $n \in \N$, $\D \subseteq \{0,1\}^n$ and $f : \D \to \{0,1\}$. Then, $\st(f) \in O(\mathsf{R}(f))$.
    \end{theorem}

    \begin{proof}
        Any randomized algorithm that computes $f$ can be thought of as a family of decision trees $(T_j)_{j=1}^k$ with depth at most $\mathsf{T} := \mathsf{R}(f)$, each of which is being evaluated according to the corresponding probability distribution $(p_j)_{j=1}^k$. First, we observe that we can assume without loss of generality that the probability distribution is uniform, i.e., all $p_j$'s are $1/k$, because we can freely duplicate decision trees and let $k$ grow arbitrarily, so that we can approximate this up to arbitrarily small error. Next, for every input $x \in \D$, we define a bit string $y \in \{0,1\}^k$, such that every $y_j$ corresponds to the outcome of $T_j$ evaluated on input $x$. Now, we observe that evaluating $f$ reduces to evaluating the gapped-majority function on $y$, and every bit evaluation of $y$ can be represented with an $st$-connectivity graph with the construction from \cref{fig:decision-tree-conversion}. Thus, composing both constructions together, we observe that the total $\st$-complexity satisfies $\st(f) \in O(\mathsf{T}) = O(\mathsf{R}(f))$.
    \end{proof}

    \subsection{Other complexity-measure relations}

    Finally, we investigate the relationship between the frameworks considered in the previous subsections, and well-known complexity measures for total boolean functions, like deterministic, randomized, and quantum query complexity.

    The first thing to note is that the graph composition framework always admits a query-optimal quantum algorithm, since we can embed an arbitrary span program on a single edge between $s$ and $t$, and we can always generate a span program with optimal complexity by generating one from an optimal solution to the dual adversary bound. As such, it doesn't make sense to define a ``graph composition'' complexity measure, since it would simply equal quantum query complexity up to constants.

    The same argument can be made for the divide and conquer framework. Indeed, one can always embed any boolean function $f : \{0,1\}^n \to \{0,1\}$ into the top-level auxiliary function $f_n^{\lambda,\aux}$, and then provide a span program with optimal complexity as $\mathcal{P}_n^{\lambda,\aux}$. Thus, if we were to define a ``divide and conquer complexity'', it would equal the quantum query complexity up to constants as well. Similarly, from the relations displayed in \cref{fig:framework-relations}, we observe that defining a ``multidimensional quantum walk complexity'' would also coincide with quantum query complexity.

    However, for the other frameworks, i.e., $st$-connectivity, (adaptive) learning graphs and (weighted) decision trees, it does make sense to define complexity measures, as we did in the previous sections. We prove several relations between these complexity measures here that were not present in the literature before.

    \begin{theorem}
        \label{thm:weighted-decision-tree-complexity}
        Let $f : \D \to \{0,1\}$ with $\D \subseteq \{0,1\}^n$. Then $\mathsf{D}(f) \leq \mathsf{WDT}(f)^2$.
    \end{theorem}

    \begin{proof}
        Let $T$ be a decision tree that minimizes $\WDT(f)$. We present a deterministic algorithm that makes at most $\WDT(f)^2$ queries. To that end, we observe from \cite[Theorem~5.4]{cornelissen2022improved} that if we have a tree $T$, where the two children of the root node have subtrees $T_L$ and $T_R$ rooted at it, then we have
        \[\WDT(T) = \frac{\WDT(T_L) + \WDT(T_R) + \sqrt{(\WDT(T_L) - \WDT(T_R))^2 + 4}}{2}.\]
        Observe that the above expression is increasing in both $\WDT(T_L)$ and $\WDT(T_R)$, since for $f(x,y) = (x+y+\sqrt{(x-y)^2+4})/2$, we have
        \[\frac{\partial f}{\partial x} = \frac12 + \frac{2(x-y)}{4\sqrt{(x-y)^2+4}} \geq \frac12\left[1 + \frac{1}{\sqrt{1 - \frac{4}{(x-y)^2}}}\right] > 0,\]
        and similarly for $y$. Now, let $P$ be the longest path in $T$ from the root to a leaf, with length $\mathrm{depth}(T)$. If for each vertex $v$ in $P$, we replace the other child with a leaf node to obtain a new tree $T'$, then we find that $\WDT(T) \geq \WDT(T')$. The resulting tree $T'$ has just one path with length $\mathrm{depth}(T)$, and leaves on all its other edges. We modify the proof of \cite[Corollary~1.7]{cornelissen2022improved}, to show that
        \[\WDT(f) = \WDT(T) \geq \WDT(T') \geq \sqrt{\mathsf{DTSize}(T')} = \sqrt{\mathrm{depth}(T)} \geq \sqrt{\mathsf{D}(f)},\]
        where the decision-tree size of $T'$ is the number of internal nodes it contains, which is $\mathrm{depth}(T)$ in this case. Finally, since $T$ computes $f$, $\mathrm{depth}(T)$ is an upper bound for $\mathsf{D}(f)$.
    \end{proof}

    \begin{corollary}
        \label{cor:R0-WDT02}
        Let $n \in \N$, $\D \subseteq \{0,1\}^n$, and $f : \D \to \{0,1\}$. Then, $\mathsf{R}_0(f) \leq \mathsf{\WDT}_0(f)^2$.
    \end{corollary}

    \begin{proof}
        Let $\mathcal{T}$ be a family of zero-error decision-trees that minimizes $\WDT_0(f)$. Draw $T \sim \mathcal{T}$, and evaluate $T$. Observe that evaluating it can be done in at most $\WDT_0(T)^2$ queries, by \cref{thm:weighted-decision-tree-complexity}. Thus, we obtain a zero-error randomized algorithm that runs in at most $\WDT_0(f)^2$ queries.
    \end{proof}

    \begin{theorem}
        \label{thm:GT-ub-WDT}
        Let $n \in \N$, $\D \subseteq \{0,1\}^n$, and $f : \D \to \{0,1\}$. Then, $\sqrt{\GT}(f) \leq \WDT(f)^{3/2}$ and $\sqrt{\GT}_0(f) \leq \WDT_0(f)^{3/2}$.
    \end{theorem}

    \begin{proof}
        Let $\mathcal{T}$ be a decision tree computing $f$, with guessing complexity $G$ and depth $T$. Analogous to \cref{thm:weighted-decision-tree-complexity}, we can prune the tree to either a complete binary tree with depth $G$, or a single-path tree of length $T$, and so $\WDT(\mathcal{T}) \geq G$, and $\WDT(\mathcal{T}) \in \Omega(\sqrt{T})$. Thus, we obtain that
        \[\WDT(\mathcal{T}) = \WDT(\mathcal{T})^{1/3} \cdot \WDT(\mathcal{T})^{2/3} \in \Omega(G^{1/3} \cdot \sqrt{T}^{2/3}) = \Omega(\sqrt{GT}^{2/3}),\]
        and since this holds for every decision tree $\mathcal{T}$, the relation also holds for the complexity measures $\WDT$ (resp.\ $\WDT_0$) and $\sqrt{\GT}$ (resp.\ $\sqrt{\GT_0}$).
    \end{proof}

    We also provide a separation that matches the above bound.

    \begin{theorem}
        \label{thm:GT-sep-WDT}
        Let $n \in \N$ and $f : \{0,1\}^{2n+1} \to \{0,1\}$ be defined as
        \[f(x) = \begin{cases}
            x_2 \oplus x_3 \oplus \cdots \oplus x_{n+1}, & \text{if } x_1 = 1, \\
            x_2 \lor x_3 \lor \cdots \lor x_{n^2+1}, & \text{if } x_1 = 0.
        \end{cases}.\]
        Then, $\WDT(f) \in O(n)$ and $\sqrt{\GT}_0(f) \in \Omega(n^{3/2})$.
    \end{theorem}

    \begin{proof}
        For the upper bound on $\WDT$, we take a decision tree that first queries $x_1$, and depending on the outcome either uses a full binary tree to evaluate the parity function on $n$ bits, or a single-path tree to evaluate the or function on $n^2$ bits. Both have a weighted-decision-tree complexity $O(n)$, and so the total weighted-decision-tree complexity of $f$ is also upper bound by $O(n)$.

        For the lower bound on $\sqrt{\GT}_0$, observe that we can obtain the or-function and the parity-function from restrictions of $f$. Thus, in any family of decision trees evaluating $f$ with zero error, there must be a decision tree $T$ that evaluates to $1$ on input $0(1)^{n^2}$. This decision tree generates a cover of the $(n^2+1)$-dimensional hypercube, where the subcube containing $0(1)^{n^2}$ must be monochromatic. This means that it can only contain $0(1)^{n^2}$, and thus its codimension must be $n^2$. That means the depth of the decision tree must be at least $n^2$, and so $\mathsf{T}_0(f) \geq n^2$, and hence it remains to prove that $\mathsf{G}_0(f) \in \Omega(n)$.

        On the other hand, any family of decision trees evaluating $f$ with zero-error must contain a decision tree that evaluates at least half of the inputs $1x0^n$ with $x \in \{0,1\}^n$ correctly, by the pigeonhole principle. If we now prune the subtrees of all the internal nodes that query outside of the parity function's inputs and that don't match with $1x0^{n(n-1)}$, the resulting pruned decision tree still evaluates $1x0^{n(n-1)}$ correctly for at least half of the $x \in \{0,1\}^n$. Thus, there must be at least $2^n/4$ $0$ or $1$-inputs that are correctly evaluated, and hence the monochromatic cover on $\{0,1\}^n$ generated by this decision tree must contain at least $2^n/4$ singletons. From this, we observe that the number of leaves at depth $n$ of the decision tree must be at least $2^n/4$, and so the decision tree size must be at least $2^n/2-1$. Thus, we find from \cite[Lemma~1]{ehrenfeucht1989learning} that the rank $G$ satisfies
        \[2^{n-1}-1 \leq 2\sum_{j=0}^{G}\binom{n}{j} - 1 \qquad \Leftrightarrow \qquad 2^{n-2} \leq \sum_{j=0}^G \binom{n}{j} \leq 2^{nH(G/n)},\]
        where $H(x) = -x\log(x) - (1-x)\log(1-x)$ is the binary entropy function, and the final inequality is \cite[Lemma~16.19]{flum2006parameterized}. Thus, using the inequality $H(x) \leq 1 - 2(x-1/2)^2$, we find that
        \[1 - \frac2n = \frac{n-2}{n} \leq H\left(\frac{G}{n}\right) \leq 1 - 2\left(\frac{G}{n} - \frac12\right)^2,\]
        and so
        \[\left|\frac{G}{n} - \frac12\right| \leq \frac{1}{\sqrt{n}},\]
        from which we find that $G \in n/2 - O(\sqrt{n}) \subseteq \Omega(n)$.
    \end{proof}

    We present an overview of the relations between the different complexity measures in \cref{fig:compl-meas}, where $\mathsf{FS}$ denotes the formula size, i.e., the minimum number of variables required to generate a read-once formula that computes $f$. The lower bound on formula size follows from a counting argument on the number of formulas of a given size. The lower bound on the learning graph for the threshold function follows from the lower bound on the learning graph complexity for the $k$-subset certificate structure, as discussed in \cite[Proposition~11]{belovs2014power}. The relations between the complexity measures $\mathsf{Q}$, $\mathsf{Q}_0$, $\mathsf{Q}_E$, $\mathsf{R}$, $\mathsf{R}_0$ and $\mathsf{D}$ follow from known results, see e.g.~\cite[Table~1]{aaronson2021degree} for an excellent overview, and the proper references.

    It is unknown whether $\st$ and $\mathsf{Q}$ can be separated. They are equal for all the ``usual suspects'', i.e., for the or, majority, parity, tribes, and indexing functions. Similarly, it is unknown how $\st$ and $\mathsf{R}$ can be maximally separated. Currently, the biggest separation we have is quadratic, for the OR-function, since $\st(\OR_n) \in O(\sqrt{n})$ and $R(\OR_n) \in \Omega(n)$. We leave figuring out tighter relations between these complexity measures for future research.

    \section{Applications}
    \label{sec:applications}

    In this section, we apply the $st$-connectivity framework to various string problems. We emphasize the simplicity of our constructions, and we obtain modest speed-ups in terms of time complexity over the state-of-the-art in the existing literature.

    \subsection{Pattern matching}

    In this section, we develop a quantum algorithm for (a version of) the pattern matching problem. In general, in the pattern matching problem, one is given access to two strings $x \in \Sigma^n$ and $y \in \Sigma^m$, where $m \leq n$, and $\Sigma$ is some finite alphabet of constant size. The pattern matching problem asks whether $y$ appears in $x$ as a substring. This problem was first studied in the quantum setting by Ramesh and Vinay~\cite{ramesh2003string}, who gave a quantum algorithm with query and time complexity $O(\sqrt{n}\log(n/m)\log(m) + \sqrt{m}\log^2(m))$. This was later improved to $O(\sqrt{n\log(m)} + \sqrt{m\log^3(m)\log\log(m)})$ by Wang and Ying~\cite[Section~4.3.2]{wang2024quantum}.

    Here, we consider a much easier setting, where we assume to be given a full description of the pattern $y \in \Sigma^m$, i.e., we only count the queries we make to $x$, but not to $y$. In this setting, Ramesh and Vinay's algorithm makes $O(\sqrt{n}\log(n/m)\log(m))$ queries, and Wang and Yang's algorithm makes $O(\sqrt{n\log(m)})$ queries. We match the latter's query complexity with the graph composition framework in the aperiodic case, and improve slightly in the case where the pattern is periodic.

    The core idea of the algorithm is to first find a match between the original string $x$ and a special subset of the characters in the pattern string $y$, known as the deterministic sample. We recall the core lemma that dates back to Vishkin~\cite{vishkin1991deterministic}, where we use the notation that for all $x \in \Sigma^n$ and $1 \leq k \leq \ell \leq n$, $x[k:\ell]$ denotes the substring starting at position $k$ and ending at $\ell$, inclusive.

    \begin{lemma}[Deterministic sample {\cite{vishkin1991deterministic}}]
        \label{lem:deterministic-sample}
        Let $\Sigma$ a finite alphabet of constant size, $m,n \in \N$ and $y \in \Sigma^m$ aperiodic, i.e., there does not exist a $p \leq m/2$ such that $y_{kp+\ell} = y_\ell$ for all $\ell \in [p]$ and $0 \leq k \leq \lfloor(n-\ell)/p\rfloor$. Then, there exists a subset $J \subseteq [m]$ and integer $0 \leq k \leq m/2$ such that $|J| \leq \log(m)$, and such that if $x[j:j+m-1]$ agrees with $y$ on $J$, with $j \in [n-m+1]$, then for any non-zero integer $\max\{0,-k\} \leq \ell \leq \min\{m/2-k,n-m+1\}$, $x[j+\ell:j+m-1+\ell]$ does not agree with $y$.
    \end{lemma}

    Intuitively, the above lemma can be interpreted as follows. For any interval of $m/2$ integers in $[n]$, we can have at most one index $j$ such that $x[j:j+m-1]$ that matches $y$ on the deterministic sample $J$. That means that we can first search for a $j \in [n-m+1]$ that matches $y$ on the deterministic sample, which requires only $O(\sqrt{n|J|}) \subseteq O(\sqrt{n\log(m)})$ queries. Now, if we find such an index $j$, we still have to check if $x[j:j+m-1]$ matches all of $y$, and the above lemma tells us that this shouldn't happen for too many $j$'s.

    We embed the above reasoning in a graph composition for this problem. We then prove the properties of the resulting quantum algorithm in the following theorem statement.

    \begin{theorem}[Quantum algorithm for aperiodic pattern matching]
        \label{thm:aperiodic-pattern-matching}
        Let $\Sigma$ be a finite alphabet of constant size, and $y \in \Sigma^m$ an aperiodic pattern. We can solve the aperiodic pattern matching problem with $O(\sqrt{n\log(m)})$ queries to the input string $x \in \Sigma^n$, $\widetilde{O}(\sqrt{n})$ elementary gates, and a QROM of size $O(m)$.
    \end{theorem}

    \begin{proof}
        For every $i \in [n]$ and $j \in [m]$, we generate a span program $[x_i = y_j]$, that computes whether $x_i$ equals $y_j$. Since we can do this exactly, i.e., without any error, using just one query, we can generate a span program with complexity $1$ for this problem as well. We assume without loss of generality that $W_+([x_i = y_j]) = W_-([x_i = y_j]) = 1$.

        Next, we construct a graph composition for the pattern matching problem. To that end, we first classically precompute a deterministic sample $\{j_1, \dots, j_k\} \subseteq [m]$ for $y$. Now, for every $i \in [n-m+1]$, we define graph compositions $\mathcal{P}_i^{\text{DS}}$ and $\mathcal{P}_i$, which check whether $x[i:i+m-1]$ matches $y$ on the deterministic sample, and in every position, respectively. Then, we make a bigger graph composition where we make $n-m+1$ parallel connections between $s$ and $t$, each composing of $\mathcal{P}_i^{\textrm{DS}}$ and $\mathcal{P}_i$ in series, for all $i \in [n-m+1]$. We refer to the resulting graph composition as $\mathcal{P}$. See also \cref{fig:pattern-matching}, where we equipped the above description with a suitable weighting scheme.

        \begin{figure}[!ht]
            \centering\scriptsize
            \begin{tikzpicture}[vertex/.style = {draw, rounded corners = .3em}, scale = 2]
                \begin{scope}
                    \node[vertex] (s) at (0,0) {$s$};
                    \node[left] at (s.west) {$\mathcal{P}_i^{\mathrm{DS}} = \;\;$};
                    \node[vertex] (1) at (1,0) {};
                    \node[vertex] (2) at (2,0) {};
                    \node[vertex] (3) at (3,0) {};
                    \node[vertex] (t) at (4,0) {$t$};

                    \draw (s) to node[above] {$[x_{i+j_1-1} = y_{j_1}]$} (1) to node[above] {$[x_{i+j_2-1} = y_{j_2}]$} (2);
                    \draw[dotted] (2) to (3);
                    \draw (3) to node[above] {$[x_{i+j_k-1} = y_{j_k}]$} (t);
                \end{scope}
                \begin{scope}[shift={(0,-.4)}]
                    \node[vertex] (s) at (0,0) {$s$};
                    \node[left] at (s.west) {$\mathcal{P}_i = \;\;$};
                    \node[vertex] (1) at (1,0) {};
                    \node[vertex] (2) at (2,0) {};
                    \node[vertex] (3) at (3,0) {};
                    \node[vertex] (t) at (4,0) {$t$};

                    \draw (s) to node[above] {$\frac1m[x_i = y_1]$} (1) to node[above] {$\frac1m[x_{i+1} = y_2]$} (2);
                    \draw[dotted] (2) to (3);
                    \draw (3) to node[above] {$\frac1m[x_{i+m-1} = y_m]$} (t);
                \end{scope}

                \begin{scope}[shift = {(5,-.2)}]
                    \node[vertex] (s) at (0,0) {$s$};
                    \node[left] at (s.west) {$\mathcal{P} = \;\;$};
                    \node[vertex] (1) at (.5,.4) {};
                    \node[vertex] (2) at (.5,.2) {};
                    \node[vertex] (3) at (.5,0) {};
                    \node[vertex] (4) at (.5,-.2) {};
                    \node[vertex] (5) at (.5,-.4) {};
                    \node[vertex] (t) at (1,0) {$t$};

                    \draw (s) to[bend left = 30] node[left] {$\mathcal{P}_1^{\mathrm{DS}}$} (1);
                    \draw (1) to[bend left = 30] node[right] {$\mathcal{P}_1$} (t);

                    \draw[dotted] (s) to[bend left = 20] (2);
                    \draw[dotted] (2) to[bend left = 20] (t);

                    \draw (s) to node[above] {$\mathcal{P}_i^{\mathrm{DS}}$} (3);
                    \draw (3) to node[above] {$\mathcal{P}_i$} (t);

                    \draw[dotted] (s) to[bend right = 20] (4);
                    \draw[dotted] (4) to[bend right = 20] (t);

                    \draw (s) to[bend right = 30] node[left] {$\mathcal{P}_{n-m+1}^{\mathrm{DS}}$} (5);
                    \draw (5) to[bend right = 30] node[right] {$\mathcal{P}_{n-m+1}$} (t);
                \end{scope}
            \end{tikzpicture}
            \caption{Graph composition construction for aperiodic pattern matching. $J = \{j_1, \dots, j_k\} \subseteq [m]$ is the deterministic sample. The dashed lines represent all the possible choices for $i \in [n-m+1]$.}
            \label{fig:pattern-matching}
        \end{figure}
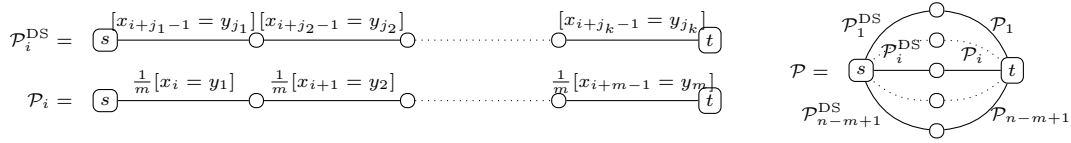

        For any positive instance $x$, suppose that $x[i:i+m-1]$ matches $y$, with $i \in [n-m+1]$. We observe that $x$ is a positive instance for $\mathcal{P}_i$, and
        \[w_+(x,\mathcal{P}) \leq w_+(x,\mathcal{P}_i) = |J| + m \cdot \frac1m \in O(\log(m)) \qquad \Rightarrow \qquad W_+(\mathcal{P}) \in O(\log(p)).\]

        On the other hand, let $x \in \Sigma^n$ be a negative instance. We observe from \cref{lem:deterministic-sample} that for every block of length $m/2$, there is at most one position $i$ for which $x$ matches $y$ on the deterministic sample. Thus, we can find a cut through the graph that cuts at most $n$ edges in the left half of the graph, and at most $O(n/m)$ edges in the right half of the graph. As such, we find that
        \[w_-(x,\mathcal{P}) \in O\left(n + \frac{n}{m} \cdot m\right) \subseteq O(n) \qquad \Rightarrow \qquad W_-(\mathcal{P}) \in O(n).\]

        We complete the analysis of the query complexity by observing that
        \[C(\mathcal{P}) = \sqrt{W_-(\mathcal{P}) \cdot W_+(\mathcal{P})} \in O\left(\sqrt{n\log(m)}\right).\]

        For the time complexity, observe that we have a symmetric series-parallel graph of constant depth, where the initial vectors are weighted combinations of a constant number of uniform superpositions. Thus, we can use \cref{thm:uniform-state-preparation,thm:circulation-space-reflection-implementation} to implement the reflection through the circulation space in time $\widetilde{O}(\log|E|) \subseteq \widetilde{O}(\log(n) + \log(m))$. For every edge, computing which bit in $x$ to query takes time $\widetilde{O}(\log(n) + \log(m))$, as long as we store the deterministic sample and $y$ in QROM, which requires size $O(m)$. Multiplying with the span program complexity results in the claimed bound.
    \end{proof}

    Finally, we generalize our results to the setting where the pattern is periodic.

    \begin{theorem}[Quantum algorithm for pattern matching]
        \label{thm:pattern-matching}
        Let $\Sigma$ be a finite alphabet of constant size, $n,m \in \N$ and $y \in \Sigma^m$ be a periodic pattern with period $p$, and $k = \lceil m/p\rceil \geq 2$ (partial) repetitions. Then, we can solve the pattern matching problem with $O(\sqrt{n\log(p)})$ queries to the input string $x \in \Sigma^n$.
    \end{theorem}

    \begin{proof}
        Let $\overline{y} := y[1:p]$ be one period of $y$. In particular, this means that $\overline{y}$ is aperiodic. We assume without loss of generality that $m|n$, because we can always enlarge the alphabet by a new character and pad the input string with these characters. Next, for all $j \in \{0, \dots, n/m-1\}$, let $x_j = x[jm+1:jm+2p] \in \Sigma^{2p}$. We observe that if $y$ appears in $x$, then there must exist a $j \in \{0, \dots, n/m-1\}$ for which $\overline{y}$ appears in $x_j$. Moreover, in each $x_j$, $\overline{y}$ can only appear at most twice, because otherwise it would contradict the aperiodicity of $\overline{y}$.

        Now, we construct the graph composition. We first look for a single period in one of the $x_j$'s. If we find it, we simply check in both directions for the first position where the pattern breaks, until we reach a total length of $m$. Based on this, we can figure out if the pattern can be found in a part of $x$ that contains $x_j$. The corresponding graph composition that achieves this, with a suitably chosen weighting scheme, is displayed in \cref{fig:pattern-matching-periodic}.

        \begin{figure}[!ht]
            \centering\scriptsize
            \begin{tikzpicture}[vertex/.style = {draw, rounded corners = .3em}, scale = 2]
                \clip (-1,.25) rectangle (7,-2.75);
                \begin{scope}
                    \node[vertex] (s) at (0,0) {$s$};
                    \node[left] at (s.west) {$\mathcal{P}_i^{\mathrm{DS}} = \;\;$};
                    \node[vertex] (1) at (1,0) {};
                    \node[vertex] (2) at (2,0) {};
                    \node[vertex] (3) at (3,0) {};
                    \node[vertex] (t) at (4,0) {$t$};

                    \draw (s) to node[above] {$[x_{i+j_1-1} = \overline{y}_{j_1}]$} (1) to node[above] {$[x_{i+j_2-1} = \overline{y}_{j_2}]$} (2);
                    \draw[dotted] (2) to (3);
                    \draw (3) to node[above] {$[x_{i+j_{k'}-1} = \overline{y}_{j_{k'}}]$} (t);
                \end{scope}
                \begin{scope}[shift={(0,-.4)}]
                    \node[vertex] (s) at (0,0) {$s$};
                    \node[left] at (s.west) {$\mathcal{P}_i^{\text{period}} = \;\;$};
                    \node[vertex] (1) at (1,0) {};
                    \node[vertex] (2) at (2,0) {};
                    \node[vertex] (3) at (3,0) {};
                    \node[vertex] (t) at (4,0) {$t$};

                    \draw (s) to node[above] {$\frac1p[x_i = \overline{y}_1]$} (1) to node[above] {$\frac1p[x_{i+1} = \overline{y}_2]$} (2);
                    \draw[dotted] (2) to (3);
                    \draw (3) to node[above] {$\frac1p[x_{i+p-1} = \overline{y}_p]$} (t);
                \end{scope}
                \begin{scope}[shift = {(0,-.8)}]
                    \node[vertex] (s) at (0,0) {$s$};
                    \node[left] at (s.west) {$\mathcal{P}^{\textrm{suff}}_{i,\ell} = \;\;$};
                    \node[vertex] (1) at (1,0) {};
                    \node[vertex] (2) at (2,0) {};
                    \node[vertex] (3) at (3,0) {};
                    \node[vertex] (t) at (4,0) {$t$};

                    \draw (s) to node[above] {$\frac1m[x_{i+m-\ell} = \overline{y}_{m-\ell+1}]$} (1) to node[below] {$\frac1m[x_{i+m-\ell+1} = \overline{y}_{m-\ell+2}]$} (2);
                    \draw[dotted] (2) to (3);
                    \draw (3) to node[above] {$\frac1m[x_{i+m-1} = \overline{y}_m]$} (t);
                \end{scope}
                \begin{scope}[shift = {(0,-1.5)}]
                    \node[vertex] (s) at (0,0) {$s$};
                    \node[left] at (s.west) {$\mathcal{P}_i = \;\;$};
                    \node[vertex] (1) at (.75,0) {};
                    \node[vertex] (2) at (1.5,0) {};
                    \node[vertex] (3) at (2.25,0) {};
                    \node[vertex] (4) at (3,0) {};
                    \node[vertex] (5) at (3.75,0) {};
                    \node[vertex] (t) at (4.5,0) {$t$};
                    \node[vertex] (21) at (1.5,-.5) {};
                    \node[vertex] (31) at (2.25,-.5) {};
                    \node[vertex] (41) at (3,-.5) {};
                    \node[vertex] (51) at (3.75,-.5) {};

                    \draw (s) to node[above] {$\mathcal{P}_i^{\textrm{DS}}$} (1) to node[above] {$\mathcal{P}_i^{\textrm{period}}$} (2) to node[above] {$\frac1k\mathcal{P}_{i-p}^{\text{period}}$} (3) to node[above] {$\frac1k\mathcal{P}_{i-2p}^{\text{period}}$} (4);
                    \draw[dotted] (4) to (5);
                    \draw (5) to node[above] {$\frac1k\mathcal{P}_{i-(k-1)p}^{\textrm{period}}$} (t);

                    \draw (2) to node[left] {$\lnot(p\mathcal{P}_{i-p}^{\textrm{period}})$} (21);
                    \draw (3) to node[left] {$\lnot(p\mathcal{P}_{i-2p}^{\textrm{period}})$} (31);
                    \draw[dotted] (4) to (41);
                    \draw (5) to node[left] {$\lnot(p\mathcal{P}_{i-(k-1)p}^{\textrm{period}})$} (51);

                    \draw (21) to[bend right = 80] node[near start, left] {$\mathcal{P}_{i+p,m-p}^{\text{suff}}$} (t);
                    \draw (31) to[bend right = 70] node[near start, left] {$\mathcal{P}_{i+p,m-2p}^{\text{suff}}$} (t);
                    \draw[dotted] (41) to[bend right = 50] (t);
                    \draw (51) to node[below right=-.4em] (x) {} (t);
                    \draw[dashed,<-] (x) to (4.5,-.5) node[right] {$\mathcal{P}_{i+p,m-(k-1)p}^{\text{suff}}$};
                \end{scope}
                \begin{scope}[shift = {(5,-.4)}]
                    \node[vertex] (s) at (0,0) {$s$};
                    \node[left] at (s.west) {$\mathcal{P} = \;\;$};
                    \node[vertex] (t) at (1,0) {$t$};

                    \draw (s) to[bend left = 80] node[above] {$\mathcal{P}_1$} (t);
                    \draw[dotted] (s) to[bend left = 40] (t);
                    \draw (s) to node[above] {$\mathcal{P}_i$} (t);
                    \draw[dotted] (s) to[bend right = 40] (t);
                    \draw (s) to[bend right = 80] node[below] {$\mathcal{P}_{n-m+1}$} (t);
                \end{scope}
            \end{tikzpicture}
            \caption{Graph composition for the periodic pattern matching problem. Here, $\{j_1, \dots, j_{k'}\} \subseteq [p]$ is a deterministic sample for $\overline{y}$, where we know by \cref{lem:deterministic-sample} that $k' \in O(\log(p))$. In the composition for $\mathcal{P}$, we let $i$ range over all intervals $\{jm+1,\dots,jm+2p\}$, with $j \in \{0, \dots, n/m-1\}$.}
            \label{fig:pattern-matching-periodic}
        \end{figure}
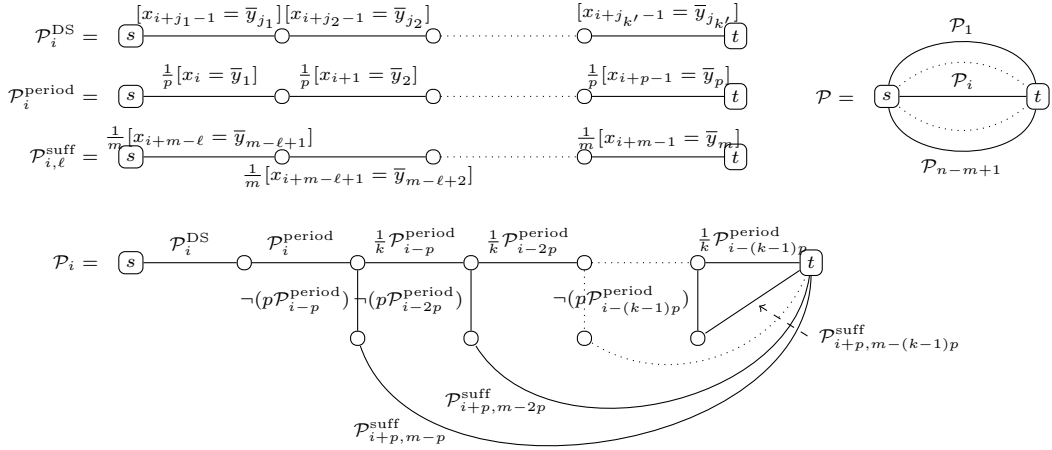

        Next, we analyze the witness sizes. To that end, suppose that $x \in \Sigma^n$ is a positive instance, and that $x[i:i+m-1]$ matches $y$ for $i \in [n-m+1]$. Then, $x$ is a positive instance for $\mathcal{P}_{i+jp}$ for some $j \in [k-1]$, and we find
        \[w_+(x,\mathcal{P}) \leq w_+(x,\mathcal{P}_{i+jp}) \leq k' + 1 + k \cdot \frac1k + \frac1p \cdot p + 1 \in O(\log(p)).\]
        On the other hand, suppose that $x \in \Sigma^n$ is a negative instance. Then, it suffices to find a cut through $\mathcal{P}$. To that end, recall that we can only match the period $\overline{y}$ on the deterministic sample for at most $2n/m$ choices of $i$. As such, we can find a cut through the graph with total weight
        \[w_-(x,\mathcal{P}) \leq n + \frac{2n}{m}(kp + m + kp) \in O(n).\]
        Thus, we conclude that
        \[C(\mathcal{P}) = \sqrt{W_+(\mathcal{P}) \cdot W_-(\mathcal{P})} \in O(\sqrt{n\log(p)}).\qedhere\]
    \end{proof}

    The time complexity in the periodic case can be analyzed in much the same way as in the aperiodic case. However, the construction is significantly more involved, and so the time complexity analysis is rather cumbersome. Therefore, we leave it for future work.

    \subsection{Various related string-search problems}

    Finally, we consider several other string problems, namely the $\OR \circ \pSEARCH$-problem, the $\Sigma^*20^*2\Sigma^*$-problem, the Dyck language recognition problem with depth $3$, and the $3$-increasing subsequence problem. It turns out that all these problems can be solved using very similar graph compositions, but each with a slightly more complicated construction than the last.

    \subsubsection{The $\OR \circ \pSEARCH$-problem}

    The $\pSEARCH$-problem first appeared in \cite[Definition~5]{belovs2018provably}, and its $\OR$-composition was subsequently considered in \cite[Section~4]{ambainis2023improved}. We start by defining it formally.

    \begin{definition}[The $\OR \circ \pSEARCH$-problem]
        Let $n \in \N$, and let $\D_n \subseteq \{0,1,*\}^n$, such that $x \in \D_n$ if and only if $x$ has exactly one position $j(x) \in [n]$ for which $x_{j(x)} \neq *$. The $\pSEARCH$-problem asks to output the value of $x_{j(x)}$, given query access to $x$.

        Next, let $m \in \N$, $T \in \{m,m+1,\dots,nm\}$ and let $\D_T \subseteq \D_n^m$, such that $(x^{(1)}, \dots, x^{(m)}) \in \D_T$ if and only if $j(x^{(1)}) + \cdots + j(x^{(m)}) = T$ and $|\overline{x}| \in \{0,1\}$, where
        \[\overline{x} = x^{(1)}_{j(x^{(1)})} \cdots x^{(m)}_{j(x^{(m)})} \in \{0,1\}^m.\]
        The $\OR \circ \pSEARCH$-problem asks to output $|\overline{x}|$, given query-access to $x \in \D_T$.
    \end{definition}

    It was proven in \cite[Theorem~2]{ambainis2023improved} that the query complexity of the $\OR \circ \pSEARCH$-problem is $\Omega(\sqrt{T\log(T)})$. Here, we prove that this is tight by constructing a graph composition that computes this problem.

    \begin{theorem}
        \label{thm:or-psearch}
        There is a quantum algorithm that solves the $\OR \circ \pSEARCH$-problem with $O(\sqrt{T\log(T)})$ queries, which is tight up to constants, $\widetilde{O}(\sqrt{T})$ elementary gates, and a QROM of size $\widetilde{O}(T)$.
    \end{theorem}

    \begin{proof}
        The lower bound follows from \cite[Theorem~2]{ambainis2023improved}, so it remains to prove the upper bound.

        For every $j \in [n]$ and $b \in \{0,1,*\}$, we let $[x_j = b]$ be the trivial span program that computes whether $x_j$ equals $b$. Since we can do this exactly, i.e., with zero error probability, using just $\Theta(1)$ queries, we can implement the resulting operations in $\Theta(1)$ queries and time as well.

        Now, we construct a graph composition that computes the $\OR \circ \pSEARCH$-problem. To that end, for each $i \in [m]$, we attach a sequence of edges to $s$ labeled by $(1/j)[x^{(i)}_j \neq *]$, with $j$ ranging from $1$ to $n$. At every node in this sequence with index $j$, we attach a connection to $t$ labeled by $[x^{(i)}_j = 1]$. The resulting graph is $G$, and the resulting graph composition is $\mathcal{P}$. See also \cref{fig:psearch}.

        \begin{figure}[!ht]
            \centering\scriptsize
            \begin{tikzpicture}[vertex/.style = {draw, rounded corners = .3em}, scale = 1.5]
                \begin{scope}
                    \node[vertex] (s) at (0,0) {$s$};
                    \node[left] at (s.west) {$\mathcal{P}_j' = \;\;$};
                    \node[vertex] (1) at (1,0) {};
                    \node[vertex] (2) at (2,0) {};
                    \node[vertex] (3) at (3,0) {};
                    \node[vertex] (4) at (4,0) {};
                    \node[vertex] (5) at (5,0) {};
                    \node[vertex] (t) at (5,-.5) {$t$};

                    \draw (s) to node[above] {$[x^{(j)}_1 = *]$} (1);
                    \draw (1) to node[above] {$\frac12[x^{(j)}_2 = *]$} (2);
                    \draw (2) to node[above] {$\frac13[x^{(j)}_3 = *]$} (3);
                    \draw[dotted] (3) to (4);
                    \draw (4) to node[above] {$\frac{1}{n-1}[x^{(j)}_{n-1} = *]$} (5);

                    \draw (s) to[bend right = 20] node[near start, below left, rotate = -20] {$[x^{(j)}_1 = 1]$} (t);
                    \draw (1) to[bend right = 15] node[near start, below left, rotate = -20] {$[x^{(j)}_2 = 1]$} (t);
                    \draw (2) to[bend right = 10] node[near start, below left, rotate = -20] {$[x^{(j)}_3 = 1]$} (t);
                    \draw[dotted] (3) to[bend right = 10] (t);
                    \draw[dotted] (4) to[bend right = 10] (t);
                    \draw (5) to node[right] {$[x^{(j)}_n = 1]$} (t);
                \end{scope}
                \begin{scope}[shift={(7,-.35)}]
                    \node[vertex] (s) at (0,0) {$s$};
                    \node[left] at (s.west) {$\mathcal{P} = \;\;$};
                    \node[vertex] (t) at (1,0) {$t$};

                    \draw (s) to[bend left = 80] node[above] {$\mathcal{P}_1'$} (t);
                    \draw[dotted] (s) to[bend left = 40] (t);
                    \draw (s) to node[above] {$\mathcal{P}_j'$} (t);
                    \draw[dotted] (s) to[bend right = 40] (t);
                    \draw (s) to[bend right = 80] node[above] {$\mathcal{P}_m'$} (t);
                \end{scope}
            \end{tikzpicture}
            \caption{The graph composition for the $\OR \circ \pSEARCH$-problem. On the right-hand side, $j$ runs over $[m]$.}
            \label{fig:psearch}
        \end{figure}
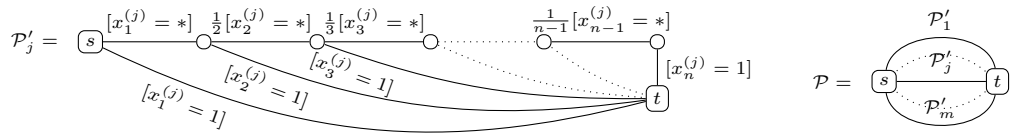

        For a positive instance $x$, now, we see that there exists a path from $s$ to $t$ along the positively-labeled edges. Moreover, if $x^{(i)}_{j(x^{(i)})} = 1$, then the resistance along the path is
        \[w_+(x,\mathcal{P}) = \sum_{j=1}^{j(x^{(i)})-1} \frac{1}{j} + 1 \in O(\log(j(x^{(i)}))) \subseteq O(\log(n)) \subseteq O(\log(T)).\]

        On the other hand, if $x$ is a negative instance, we find a cut through the graph. We can cut through the $i$th block at the $j(x^{(i)})$th position, and the effective resistance along the cut is $2j(x^{(i)})$. Thus, the total effective resistance becomes
        \[w_-(x,\mathcal{P}) \leq \sum_{i=1}^m 2j(x^{(i)}) = 2T.\]

        We conclude by observing that
        \[C(\mathcal{P}) = \sqrt{W_-(\mathcal{P}) \cdot W_+(\mathcal{P})} \in O\left(\sqrt{T\log(T)}\right).\]

        For the time complexity, observe that every tree between $s$ and $t$, shown in \cref{fig:psearch}, is isomorphic to the graph displayed in \cref{fig:ladder-decomposition}. Moreover, we observe that we can recursively perform a parallel decomposition on this graph, cf.\ \cref{lem:parallel-decomposition}, each time halving the number of edges that is present in the graph. Thus, we can perform a full tree-parallel decomposition for the graph used in $\mathcal{P}$, with $O(\log(n))$ levels of recursion.

        \begin{figure}[!ht]
            \centering
            \def\n{7}
            \begin{tikzpicture}[vertex/.style = {draw, rounded corners = .3em}]
                \node[vertex,fill=gray!40] (t) at (0,0) {$t$};
                \foreach \x [evaluate=\x] in {2-1,...-1,\n-1} {
                    \draw ({sin((\x-.5)*360/\n)},{cos((\x-.5)*360/\n)}) to ({sin((\x+.5)*360/\n)},{cos((\x+.5)*360/\n)}) to (t);
                }
                \foreach \x [evaluate=\x] in {2-1,...-1,\n-1} {
                    \node[vertex,fill=white] at ({sin((\x+.5)*360/\n)},{cos((\x+.5)*360/\n)}) {};
                }
                \node[vertex,fill=white] (s) at ({sin(180/\n)},{cos(180/\n)}) {$s$};
                \draw (s) to (t);
                \node[vertex,fill=gray!40] at (0,-1) {};
            \end{tikzpicture}
            \caption{This graph is isomorphic to the graph used in the construction of $\mathcal{P}_j'$ in \cref{fig:psearch}. By performing a parallel decomposition between the two gray nodes, we the graph falls apart in $3$ disjoint parts, all of which contain at most half the number of edges of the original graph.}
            \label{fig:ladder-decomposition}
        \end{figure}
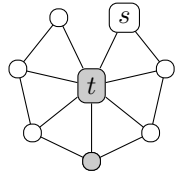

        Since, without loss of generality, we can assume that $n \leq T$ (otherwise we can simply ignore part of the input), we conclude that we can implement the reflection through the circulation space of the graph composition in time polylogarithmic in $T$. As such, the time complexity becomes equal to the query complexity up to polylogarithmic factors. Finally, the number of edges in the graph is $\widetilde{O}(nm)$, but since the construction for all the $\mathcal{P}_j'$s is the same, we don't have to store the routines for these constructions separately in QROM. Thus, the remaining size of QROM required is $\widetilde{O}(n) \subseteq \widetilde{O}(T)$.
    \end{proof}

    \subsubsection{The $\Sigma^*20^*2\Sigma^*$-problem}

    With a slight modification of the graph composition from the previous section, we can solve a different string problem that has been considered in the literature, namely, the $\Sigma^*20^*2\Sigma^*$-problem. We define it first.

    \begin{definition}[The $\Sigma^*20^*2\Sigma^*$-problem]
        Let $n \in \N$, $\Sigma = \{0,1,2\}$ and $x \in \Sigma^n$. The $\Sigma^*20^*2\Sigma^*$-problem asks whether $x$ is recognized by the regular language $\Sigma^*20^*2\Sigma^*$, i.e., whether there exist $j_1,j_2 \in [n]$ such that $j_1 < j_2$, $x_{j_1} = x_{j_2} = 2$, and $x_k = 0$ for all $k \in \{j_1+1, \dots, j_2-1\}$.
    \end{definition}

    This problem was first considered in \cite{aaronson2019quantum}, where it was also referred to as the dynamic-AND-OR language. They showed that the query complexity of any $*$-free regular language is $\widetilde{O}(\sqrt{n})$, which includes the dynamic-AND-OR language. It was also considered by Childs et al., who showed that the query complexity is $O(\sqrt{n\log(n)})$ using a divide-and-conquer approach~\cite[Theorem~4]{childs2022quantum}. We recover their result, and note that our approach does not use any divide-and-conquer strategy.

    \begin{theorem}
        \label{thm:202}
        There is a quantum algorithm that evaluates the $\Sigma^*20^*2\Sigma^*$-problem using $O(\sqrt{n\log(n)})$ queries, and it can be implemented in time $\widetilde{O}(\sqrt{n})$.
    \end{theorem}

    \begin{proof}
        For all $j \in [n]$ and $b \in \{0,1,2\}$, we construct the trivial span program $[x_j = b]$ that computes whether $x_j$ equals $b$. Since we can do this exactly, i.e., with error probability zero, using just $\Theta(1)$ queries and time, we can construct the operations of these span programs in $\Theta(1)$ queries and time as well.

        Next, we construct a graph composition that computes the $\Sigma^*20^*2\Sigma^*$-problem. To that end, for all $i \in [n-1]$, we attach an edge to $s$ labeled by $[x_i = 2]$. Attached to its leaf, we attach a sequence of edges labeled by $(1/j)[x_{i+j} = 0]$, for $j \in [n-i-1]$. Finally, we connect every $j$th node in this sequence with $t$, by an edge labeled by $[x_{i+j} = 2]$. We refer to the resulting graph as $G$, and the resulting graph composition as $\mathcal{P}$. See also \cref{fig:202}.

        \begin{figure}[!ht]
            \centering\scriptsize
            \begin{tikzpicture}[vertex/.style = {draw, rounded corners = .3em}, scale = 1.5]
                \begin{scope}
                    \node[vertex] (s) at (0,0) {$s$};
                    \node[left] at (s.west) {$\mathcal{P}_j' = \;\;$};
                    \node[vertex] (1) at (1,0) {};
                    \node[vertex] (2) at (2,0) {};
                    \node[vertex] (3) at (3,0) {};
                    \node[vertex] (4) at (4,0) {};
                    \node[vertex] (5) at (5,0) {};
                    \node[vertex] (6) at (6,0) {};
                    \node[vertex] (t) at (6,-.5) {$t$};

                    \draw (s) to node[above] {$[x_j = 2]$} (1);
                    \draw (1) to node[above] {$[x_{j+1} = 0]$} (2);
                    \draw (2) to node[above] {$\frac12[x_{j+2} = 0]$} (3);
                    \draw (3) to node[above] {$\frac13[x_{j+3} = 0]$} (4);
                    \draw[dotted] (4) to (5);
                    \draw (5) to node[above] {$\frac{1}{n-1-j}[x_{n-1} = 0]$} (6);

                    \draw (1) to[bend right = 20] node[near start, below left, rotate = -20] {$[x_{j+1} = 2]$} (t);
                    \draw (2) to[bend right = 15] node[near start, below left, rotate = -20] {$[x_{j+2} = 2]$} (t);
                    \draw (3) to[bend right = 10] node[near start, below left, rotate = -20] {$[x_{j+3} = 2]$} (t);
                    \draw[dotted] (4) to[bend right = 10] (t);
                    \draw[dotted] (5) to[bend right = 10] (t);
                    \draw (6) to node[right] {$[x_n = 2]$} (t);
                \end{scope}
                \begin{scope}[shift={(7.5,-.35)}]
                    \node[vertex] (s) at (0,0) {$s$};
                    \node[left] at (s.west) {$\mathcal{P} = \;\;$};
                    \node[vertex] (t) at (1,0) {$t$};

                    \draw (s) to[bend left = 80] node[above] {$\mathcal{P}_1'$} (t);
                    \draw[dotted] (s) to[bend left = 40] (t);
                    \draw (s) to node[above] {$\mathcal{P}_j'$} (t);
                    \draw[dotted] (s) to[bend right = 40] (t);
                    \draw (s) to[bend right = 80] node[above] {$\mathcal{P}_{n-1}'$} (t);
                \end{scope}
            \end{tikzpicture}
            \caption{The graph composition for the $\Sigma^*20^*2\Sigma^*$-problem. On the right-hand side, $j$ runs over $[n-1]$.}
            \label{fig:202}
        \end{figure}
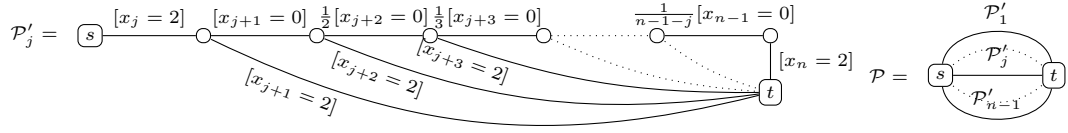

        For a positive instance $x$, observe that there indeed exists a path from $s$ to $t$ along positively-labeled edges. Moreover, if $j_1,j_2$ are such that $j_1 < j_2$, $x_{j_1} = x_{j_2} = 2$, and for all $k \in \{j_1+1, \dots, j_2-1\}$, $x_k = 0$, then the total resistance along this path is
        \[w_+(x,\mathcal{P}) \leq 2 + \sum_{j=j_1+1}^{j_2-1} \frac{1}{j-j_1} \in O(\log(j_2-j_1)) \subseteq O(\log(n)).\]

        On the other hand, for a negative instance $x$, we find a cut through the graph $G$ through negatively-labeled edges. For all $i \in [n]$ for which $x_i \neq 2$, we cut the edge adjacent to $s$ labeled by $[x_i = 2]$. This leaves to find a cut through all the $i$th blocks, for which $x_i = 2$. To that end, let $j > i$ be the minimal index for which $x_j = 2$. Since $x$ is a negative instance, there must be a $k_i \in \{i+1, \dots, j-1\}$, such that $x_{k_i} = 1$. We find a cut through the $i$th block, that cuts through $k_i - i$ edges with weight $1$, and one edge with weight $k_i - i$. Thus, the total resistance along this cut satisfies
        \[w_-(x,\mathcal{P}) \leq n + \sum_{\substack{i = 1 \\ x_i = 2}}^n 2(k_i-i) \leq n + 2n \in O(n),\]
        where we used that $k_i < j$, the latter of which is the $i$ from the next term, and so everything cancels by telescoping.

        We conclude by observing that
        \[C(\mathcal{P}) = \sqrt{W_-(\mathcal{P}) \cdot W_+(\mathcal{P})} \in O\left(\sqrt{n\log(n)}\right).\]

        For the time complexity, we use the same divide and conquer approach as in \cref{thm:or-psearch} to argue that the number of levels of recursion is $O(\sqrt{n})$. As such, the total time complexity of implementing the reflection around the circulation space is polylogarithmic in the number of edges, which is $O(n^2)$, and so the total overhead is polylogarithmic in $n$.
    \end{proof}

    \subsubsection{Recognizing the Dyck language}

    We can apply very similar ideas to recognize the Dyck language with bounded depth. We define the problem first.

    \begin{definition}[Bounded-depth Dyck language recognition problem]
        Let $k \in \N$, and $\Sigma = \{\texttt{(},\texttt{)}\}$. For any finite string $x \in \Sigma^*$, we define $w(x) = |x^{-1}(\texttt{(})| - |x^{-1}(\texttt{)})|$, i.e., the difference of the number of opening and closing brackets in $x$. The Dyck language of depth $k$ contains all words $x \in \Sigma^*$, for which $w(x) = 0$, and any prefix $x'$ (i.e., any substring at the start of $x$) satisfies $0 \leq w(x') \leq k$. For any even $n \in \N$, the depth-$k$ Dyck language recognition problem now asks to determine for a length-$n$ input $x \in \Sigma^n$, whether it is in the depth-$k$ Dyck language, given query-access to $x$.
    \end{definition}

    This problem was studied Ambainis et al.~\cite{ambainis2020quantum}, who gave a quantum query algorithm for recognizing the Dyck language of depth $k$ with $O(\sqrt{n}\log^{k/2}(n))$ queries. Later, it was shown by Khadiev and Kravchenko that the same complexity can be achieved when we consider multiple types of opening and closing brackets~\cite{khadiev2021quantum}.

    For the bounded-depth Dyck language recognition problem with $k = 1$ and $k = 2$, we now trivially give a $\Theta(\sqrt{n})$-query algorithm that recognizes the Dyck language.

    \begin{theorem}
        There exists a graph composition algorithm for recognizing the Dyck language of depths $1$ and $2$, making $O(\sqrt{n})$ queries and spending $\widetilde{O}(\sqrt{n})$ time.
    \end{theorem}

    \begin{proof}
        The lower bound search from a simple reduction to search, so it remains to prove the upper bounds. To that end, we consider the graph compositions from \cref{fig:dyck12}.

        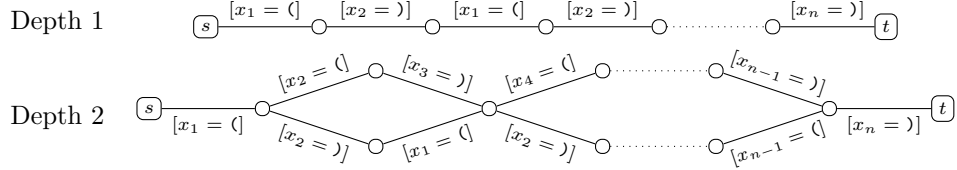
\begin{figure}[!ht]
            \centering\scriptsize
            \begin{tabular}{cc}
                {\normalsize Depth $1$} &
                \raisebox{-.5em}{\begin{tikzpicture}[vertex/.style = {draw, rounded corners = .3em}, xscale=1.5]
                    \node[vertex] (s) at (0,0) {$s$};
                    \node[vertex] (1) at (1,0) {};
                    \node[vertex] (2) at (2,0) {};
                    \node[vertex] (3) at (3,0) {};
                    \node[vertex] (4) at (4,0) {};
                    \node[vertex] (5) at (5,0) {};
                    \node[vertex] (t) at (6,0) {$t$};

                    \draw (s) to node[above] {$[x_1 = \texttt{(}]$} (1) to node[above] {$[x_2 = \texttt{)}]$} (2) to node[above] {$[x_1 = \texttt{(}]$} (3) to node[above] {$[x_2 = \texttt{)}]$} (4);
                    \draw[dotted] (4) to (5);
                    \draw (5) to node[above] {$[x_n = \texttt{)}]$} (t);
                \end{tikzpicture}} \\
                {\normalsize Depth $2$} &
                \raisebox{-2.5em}{\begin{tikzpicture}[vertex/.style = {draw, rounded corners = .3em}, xscale=1.5]
                    \node[vertex] (s) at (0,0) {$s$};
                    \node[vertex] (1) at (1,0) {};
                    \node[vertex] (20) at (2,-.5) {};
                    \node[vertex] (21) at (2,.5) {};
                    \node[vertex] (3) at (3,0) {};
                    \node[vertex] (40) at (4,-.5) {};
                    \node[vertex] (41) at (4,.5) {};
                    \node[vertex] (50) at (5,-.5) {};
                    \node[vertex] (51) at (5,.5) {};
                    \node[vertex] (6) at (6,0) {};
                    \node[vertex] (t) at (7,0) {$t$};

                    \draw (s) to node[below] {$[x_1 = \texttt{(}]$} (1) to node[below, rotate=-{atan(1/3)}] {$[x_2 = \texttt{)}]$} (20) to node[below, rotate={atan(1/3)}] {$[x_1 = \texttt{(}]$} (3) to node[below, rotate=-{atan(1/3)}] {$[x_2 = \texttt{)}]$} (40);
                    \draw (1) to node[above, rotate={atan(1/3)}] {$[x_2 = \texttt{(}]$} (21) to node[above, rotate=-{atan(1/3)}] {$[x_3 = \texttt{)}]$} (3) to node[above, rotate={atan(1/3)}] {$[x_4 = \texttt{(}]$} (41);
                    \draw[dotted] (40) to (50);
                    \draw[dotted] (41) to (51);
                    \draw (50) to node[below, rotate={atan(1/3)}] {$[x_{n-1} = \texttt{(}]$} (6) to node[below] {$[x_n = \texttt{)}]$} (t);
                    \draw (51) to node[above, rotate=-{atan(1/3)}] {$[x_{n-1} = \texttt{)}]$} (6);
                \end{tikzpicture}}
            \end{tabular}
            \caption{Graph composition for the depth-$1$ and depth-$2$ Dyck language recognition problem.}
            \label{fig:dyck12}
        \end{figure}

        For the witness sizes, observe that any path from $s$ to $t$ contains $n$ edges, so $W_+(\mathcal{P}) \in O(n)$, and similarly any cut contains $2$ edges, so $W_-(\mathcal{P}) \in O(1)$. Thus, $C(\mathcal{P}) \in O(\sqrt{n})$.

        For the time complexity, observe that the circulation space for the depth-$1$ construction is empty, so reflecting around it is trivial. For the depth-$2$ case, we can simply use \cref{lem:tree-decomposition} once to decompose into constant-sized unitaries, which take $\widetilde{O}(\polylog\dim(\H)) = \widetilde{O}(\polylog(n))$ time to implement.
    \end{proof}

    For $k = 3$, we make an observation about the structure of a string $x \in \Sigma^n$ that is not a Dyck word.

    \begin{lemma}
        \label{lem:dyck3-properties}
        Let $n \in \N$ even, and $x \in \Sigma^n$. Then $x$ is not a Dyck word of depth $3$ if and only if at least one of the following four statements is true:
        \begin{enumerate}[nosep]
            \item There is a $j \in [n]$ such that $j$ is odd, $x_j = \texttt{)}$, and for all $k \in [j-1]$, $x_k = \texttt{)}$ if and only if $k$ is even.
            \item There is a $j \in [n]$ such that $j$ is even, $x_j = \texttt{(}$, and for all $k \in \{j+1, \dots, n\}$, $x_k = \texttt{(}$ if and only if $k$ is odd.
            \item There are $1 \leq j < k \leq n$, such that $j$ is even, $k$ is odd, $x_j = x_{k+1} = \texttt{(}$, and for all $\ell \in \{j+1, j+2, \dots, k\}$, $x_{\ell} = \texttt{(}$ if and only if $\ell$ is odd.
            \item There are $1 \leq j < k \leq n$, such that $j$ is odd, $k$ is even, $x_j = x_{k+1} = \texttt{)}$, and for all $\ell \in \{j+1, j+2, \dots, k\}$, $x_{\ell} = \texttt{)}$ if and only if $\ell$ is even.
        \end{enumerate}
    \end{lemma}

    \begin{proof}
        We start with a few observations. First, suppose that $x$ is a valid Dyck word of bounded depth $3$, $j \in [n]$ is even and $x_j = \texttt{(}$. By a parity argument, we find that $w(x[:j-1])$ is odd, and we also have $w(x[:j]) = w(x[:j-1]) 1$ and $0 \leq w(x[:j-1]), w(x[:j]) \leq 3$. We conclude that $w(x[:j]) = 1$.

        Next, we observe that the problem is symmetric. Indeed, $x$ is a valid Dyck-word of bounded-depth $3$, if and only if $\overline{x}$ is too, where $\overline{x}$ is formed by reversing $x$ and changing every $\texttt{(}$ into $\texttt{)}$ and vice versa.

        Now, we check that if any of the conditions holds, then $x$ cannot be a valid Dyck word of bounded depth $3$. For the first condition, we check directly that it implies that there is an unmatched closing bracket at position $j$. For the third condition, suppose towards contradiction that $x$ is valid. Then, $w(x[:j-1]) = 1$, and so $w(x[:j]) = w(x[:k-1]) = 2$. This implies that $w(x[:k+1]) = w([x:k-1]) + 2 = 4$, which is a contradiction. By symmetry, claims $2$ and $4$ imply that $\overline{x}$ is not a valid Dyck-word of depth $3$, which is equivalant to $x$ not being a valid Dyck-word of depth $3$.

        It remains to check that for every invalid Dyck-word of depth $3$, at least one of the conditions is true. To that end, suppose that $x$ is not a valid Dyck-word of depth $3$.

        First, suppose that there exists a $k \in [n-1]$ for which $w(x[:k+1]) = 4$. Let $k \in [n]$ be the minimal such index, and let $j \in [k]$ be the largest value for which $w(x[:j-1]) = 1$. Then, the third condition holds.

        Next, suppose there is some $k \in [n-1]$ for which $w(x[:k+1]) = -1$. Let $k$ be the smallest such index. We distinguish two cases. First, suppose that there exists some $j \in \{2, \dots, k\}$ for which $w(x[:j-1]) = 2$. Let $j$ be the maximal such choice. Then, the fourth condition holds. On the other hand, if for all $j \in \{2, \dots, k\}$, $w(x[:j]) \leq 1$, then the first condition holds.

        It remains to check the case where $0 \leq w(x[:j]) \leq 3$ for all $j \in [n]$. Since $x$ is not a valid Dyck word, we must have that $w(x) = 2$. But then $w(\overline{x}) = -2$, and hence one of conditions 1 and 4 must hold for $\overline{x}$, which means that one of conditions 2 and 3 must hold for $x$.
    \end{proof}

    This characterization of invalid Dyck-words of depth $3$ allows us to construct an $O(\sqrt{n\log(n)})$-algorithm for the depth-$3$ Dyck language recognition problem.

    \begin{theorem}
        \label{thm:dyck-3}
        There is a quantum algorithm that recognizes the Dyck language of depth-$3$, making $O(\sqrt{n\log(n)})$ queries, and running in time $\widetilde{O}(\sqrt{n})$.
    \end{theorem}

    \begin{proof}
        For all $j \in [n]$ and $b \in \Sigma$, we let $[x_j = b]$ be the trivial span program that computes whether $x_j$ equals $b$. Since this can be done exactly, i.e., with zero error, in $\Theta(1)$ queries and time, we can implement the operations of this span program in $\Theta(1)$ queries and time as well.

        Next, we build graph compositions checking the conditions from \cref{lem:dyck3-properties} separately. For the first condition, we make a long sequence that checks for alternating \texttt{(} and \texttt{)}, with harmonically decreasing weights for every pair. Then, for every odd position, we additionally connect the vertex to $t$, checking for \texttt{(}, with weight $1$. We refer to the resulting graph-composed span program as $\mathcal{P}_1$. We construct $\mathcal{P}_2$ for condition $2$ in the same way, but with the string reversed. See also the top construction of \cref{fig:dyck3}.

        For the third condition, we attach an edge that checks whether $x_j = \texttt{(}$, for all even $j \in [n]$. Then, to every resulting leaf, we attach a sequence of edges that checks for alternating \texttt{(} and \texttt{)}, with harmonically decreasing weights for every pair. Finally, for every even position, we connect the vertex to $t$ with weight $1$, checking for \texttt{(}. The construction for the fourth condition is again similar. See also the bottom construction in \cref{fig:dyck3}.

        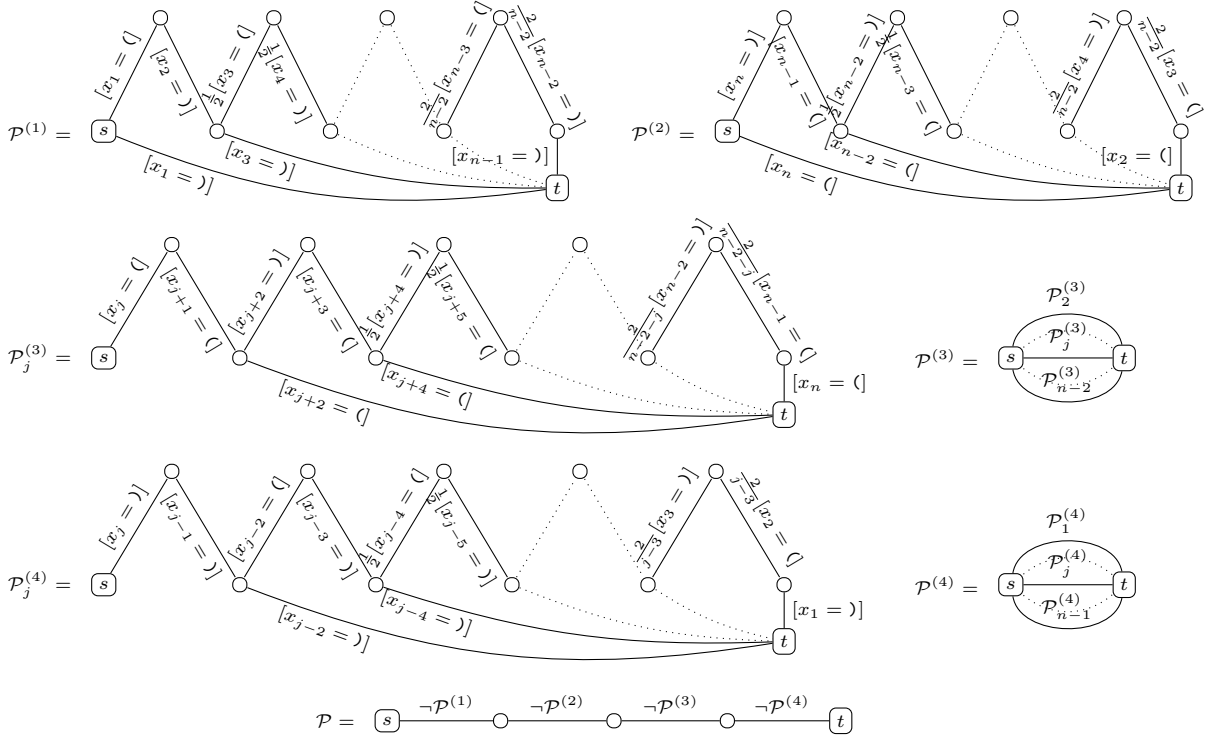
\begin{figure}[!ht]
            \centering\scriptsize
            \begin{tikzpicture}[vertex/.style = {draw, rounded corners = .3em}, scale = 1.5]
                \begin{scope}
                    \node[vertex] (s) at (0,0) {$s$};
                    \node[left] at (s.west) {$\mathcal{P}^{(1)} = \;\;$};
                    \node[vertex] (3) at (.5,1) {};
                    \node[vertex] (4) at (1,0) {};
                    \node[vertex] (5) at (1.5,1) {};
                    \node[vertex] (6) at (2,0) {};
                    \node[vertex] (7) at (2.5,1) {};
                    \node[vertex] (8) at (3,0) {};
                    \node[vertex] (9) at (3.5,1) {};
                    \node[vertex] (10) at (4,0) {};
                    \node[vertex] (t) at (4,-.5) {$t$};

                    \draw (s) to node[above, rotate = {atan(2)}] {$[x_1 = \texttt{(}]$} (3);
                    \draw (3) to node[below, rotate = {-atan(2)}] {$[x_2 = \texttt{)}]$} (4);
                    \draw (4) to node[above, rotate = {atan(2)}] {$\frac12[x_3 = \texttt{(}]$} (5);
                    \draw (5) to node[below, rotate = {-atan(2)}] {$\frac12[x_4 = \texttt{)}]$} (6);
                    \draw[dotted] (6) to (7) to (8);
                    \draw (8) to node[above, rotate = {atan(2)}] {$\frac{2}{n-2}[x_{n-3} = \texttt{(}]$} (9);
                    \draw (9) to node[above, rotate = {-atan(2)}] {$\frac{2}{n-2}[x_{n-2} = \texttt{)}]$} (10);

                    \draw (s) to[bend right = 15] node[near start, below left, rotate = -20] {$[x_1 = \texttt{)}]$} (t);
                    \draw (4) to[bend right = 10] node[near start, below left, rotate = -20] {$[x_3 = \texttt{)}]$} (t);
                    \draw[dotted] (6) to[bend right = 10] (t);
                    \draw[dotted] (8) to[bend right = 10] (t);
                    \draw (10) to node[left] {$[x_{n-1} = \texttt{)}]$} (t);
                \end{scope}
                \begin{scope}[shift={(5.5,0)}]
                    \node[vertex] (s) at (0,0) {$s$};
                    \node[left] at (s.west) {$\mathcal{P}^{(2)} = \;\;$};
                    \node[vertex] (3) at (.5,1) {};
                    \node[vertex] (4) at (1,0) {};
                    \node[vertex] (5) at (1.5,1) {};
                    \node[vertex] (6) at (2,0) {};
                    \node[vertex] (7) at (2.5,1) {};
                    \node[vertex] (8) at (3,0) {};
                    \node[vertex] (9) at (3.5,1) {};
                    \node[vertex] (10) at (4,0) {};
                    \node[vertex] (t) at (4,-.5) {$t$};

                    \draw (s) to node[above, rotate = {atan(2)}] {$[x_n = \texttt{)}]$} (3);
                    \draw (3) to node[below, rotate = {-atan(2)}] {$[x_{n-1} = \texttt{(}]$} (4);
                    \draw (4) to node[above, rotate = {atan(2)}] {$\frac12[x_{n-2} = \texttt{)}]$} (5);
                    \draw (5) to node[below, rotate = {-atan(2)}] {$\frac12[x_{n-3} = \texttt{(}]$} (6);
                    \draw[dotted] (6) to (7) to (8);
                    \draw (8) to node[above, rotate = {atan(2)}] {$\frac{2}{n-2}[x_4 = \texttt{)}]$} (9);
                    \draw (9) to node[above, rotate = {-atan(2)}] {$\frac{2}{n-2}[x_3 = \texttt{(}]$} (10);

                    \draw (s) to[bend right = 15] node[near start, below left, rotate = -20] {$[x_n = \texttt{(}]$} (t);
                    \draw (4) to[bend right = 10] node[near start, below left, rotate = -20] {$[x_{n-2} = \texttt{(}]$} (t);
                    \draw[dotted] (6) to[bend right = 10] (t);
                    \draw[dotted] (8) to[bend right = 10] (t);
                    \draw (10) to node[left] {$[x_2 = \texttt{(}]$} (t);
                \end{scope}
                \begin{scope}[shift={(0,-2)}]
                    \node[vertex] (s) at (0,0) {$s$};
                    \node[left] at (s.west) {$\mathcal{P}_j^{(3)} = \;\;$};
                    \node[vertex] (1) at (.6,1) {};
                    \node[vertex] (2) at (1.2,0) {};
                    \node[vertex] (3) at (1.8,1) {};
                    \node[vertex] (4) at (2.4,0) {};
                    \node[vertex] (5) at (3,1) {};
                    \node[vertex] (6) at (3.6,0) {};
                    \node[vertex] (7) at (4.2,1) {};
                    \node[vertex] (8) at (4.8,0) {};
                    \node[vertex] (9) at (5.4,1) {};
                    \node[vertex] (10) at (6,0) {};
                    \node[vertex] (t) at (6,-.5) {$t$};

                    \draw (s) to node[above, rotate = {atan(5/3)}] {$[x_j = \texttt{(}]$} (1);
                    \draw (1) to node[below, rotate = {-atan(5/3)}] {$[x_{j+1} = \texttt{(}]$} (2);
                    \draw (2) to node[above, rotate = {atan(5/3)}] {$[x_{j+2} = \texttt{)}]$} (3);
                    \draw (3) to node[below, rotate = {-atan(5/3)}] {$[x_{j+3} = \texttt{(}]$} (4);
                    \draw (4) to node[above, rotate = {atan(5/3)}] {$\frac12[x_{j+4} = \texttt{)}]$} (5);
                    \draw (5) to node[below, rotate = {-atan(5/3)}] {$\frac12[x_{j+5} = \texttt{(}]$} (6);
                    \draw[dotted] (6) to (7) to (8);
                    \draw (8) to node[above, rotate = {atan(5/3)}] {$\frac{2}{n-2-j}[x_{n-2} = \texttt{)}]$} (9);
                    \draw (9) to node[above, rotate = {-atan(5/3)}] {$\frac{2}{n-2-j}[x_{n-1} = \texttt{(}]$} (10);

                    \draw (2) to[bend right = 15] node[near start, below left, rotate = -20] {$[x_{j+2} = \texttt{(}]$} (t);
                    \draw (4) to[bend right = 10] node[near start, below left, rotate = -20] {$[x_{j+4} = \texttt{(}]$} (t);
                    \draw[dotted] (6) to[bend right = 10] (t);
                    \draw[dotted] (8) to[bend right = 10] (t);
                    \draw (10) to node[right] {$[x_n = \texttt{(}]$} (t);
                \end{scope}
                \begin{scope}[shift={(8,-2)}]
                    \node[vertex] (s) at (0,0) {$s$};
                    \node[left] at (s.west) {$\mathcal{P}^{(3)} = \;\;$};
                    \node[vertex] (t) at (1,0) {$t$};

                    \draw (s) to[bend left = 80] node[above] {$\mathcal{P}^{(3)}_2$} (t);
                    \draw[dotted] (s) to[bend left = 40] (t);
                    \draw (s) to node[above] {$\mathcal{P}^{(3)}_j$} (t);
                    \draw[dotted] (s) to[bend right = 40] (t);
                    \draw (s) to[bend right = 80] node[above] {$\mathcal{P}^{(3)}_{n-2}$} (t);
                \end{scope}
                \begin{scope}[shift={(0,-4)}]
                    \node[vertex] (s) at (0,0) {$s$};
                    \node[left] at (s.west) {$\mathcal{P}_j^{(4)} = \;\;$};
                    \node[vertex] (1) at (.6,1) {};
                    \node[vertex] (2) at (1.2,0) {};
                    \node[vertex] (3) at (1.8,1) {};
                    \node[vertex] (4) at (2.4,0) {};
                    \node[vertex] (5) at (3,1) {};
                    \node[vertex] (6) at (3.6,0) {};
                    \node[vertex] (7) at (4.2,1) {};
                    \node[vertex] (8) at (4.8,0) {};
                    \node[vertex] (9) at (5.4,1) {};
                    \node[vertex] (10) at (6,0) {};
                    \node[vertex] (t) at (6,-.5) {$t$};

                    \draw (s) to node[above, rotate = {atan(5/3)}] {$[x_j = \texttt{)}]$} (1);
                    \draw (1) to node[below, rotate = {-atan(5/3)}] {$[x_{j-1} = \texttt{)}]$} (2);
                    \draw (2) to node[above, rotate = {atan(5/3)}] {$[x_{j-2} = \texttt{(}]$} (3);
                    \draw (3) to node[below, rotate = {-atan(5/3)}] {$[x_{j-3} = \texttt{)}]$} (4);
                    \draw (4) to node[above, rotate = {atan(5/3)}] {$\frac12[x_{j-4} = \texttt{(}]$} (5);
                    \draw (5) to node[below, rotate = {-atan(5/3)}] {$\frac12[x_{j-5} = \texttt{)}]$} (6);
                    \draw[dotted] (6) to (7) to (8);
                    \draw (8) to node[above, rotate = {atan(5/3)}] {$\frac{2}{j-3}[x_3 = \texttt{)}]$} (9);
                    \draw (9) to node[above, rotate = {-atan(5/3)}] {$\frac{2}{j-3}[x_2 = \texttt{(}]$} (10);

                    \draw (2) to[bend right = 15] node[near start, below left, rotate = -20] {$[x_{j-2} = \texttt{)}]$} (t);
                    \draw (4) to[bend right = 10] node[near start, below left, rotate = -20] {$[x_{j-4} = \texttt{)}]$} (t);
                    \draw[dotted] (6) to[bend right = 10] (t);
                    \draw[dotted] (8) to[bend right = 10] (t);
                    \draw (10) to node[right] {$[x_1 = \texttt{)}]$} (t);
                \end{scope}
                \begin{scope}[shift={(8,-4)}]
                    \node[vertex] (s) at (0,0) {$s$};
                    \node[left] at (s.west) {$\mathcal{P}^{(4)} = \;\;$};
                    \node[vertex] (t) at (1,0) {$t$};

                    \draw (s) to[bend left = 80] node[above] {$\mathcal{P}^{(4)}_1$} (t);
                    \draw[dotted] (s) to[bend left = 40] (t);
                    \draw (s) to node[above] {$\mathcal{P}^{(4)}_j$} (t);
                    \draw[dotted] (s) to[bend right = 40] (t);
                    \draw (s) to[bend right = 80] node[above] {$\mathcal{P}^{(4)}_{n-1}$} (t);
                \end{scope}
                \begin{scope}[shift={(2.5,-5.2)}]
                    \node[vertex] (s) at (0,0) {$s$};
                    \node[left] at (s.west) {$\mathcal{P} = \;\;$};
                    \node[vertex] (1) at (1,0) {};
                    \node[vertex] (2) at (2,0) {};
                    \node[vertex] (3) at (3,0) {};
                    \node[vertex] (t) at (4,0) {$t$};

                    \draw (s) to node[above] {$\lnot\mathcal{P}^{(1)}$} (1) to node[above] {$\lnot\mathcal{P}^{(2)}$} (2) to node[above] {$\lnot\mathcal{P}^{(3)}$} (3) to node[above] {$\lnot\mathcal{P}^{(4)}$} (t);
                \end{scope}
            \end{tikzpicture}
            \caption{The graph composition construction for recognizing the Dyck language with depth $3$. For all $k \in [4]$, the construction for $\mathcal{P}^{(k)}$ checks for the $k$th condition in \cref{lem:3is-properties}. For $\mathcal{P}^{(3)}$, $j$ runs over all even numbers in $[n-2]$, and for $\mathcal{P}^{(4)}$, $j$ runs over all odd numbers in $[n-1]$.}
            \label{fig:dyck3}
        \end{figure}

        For the witness sizes, observe that $W_+(\mathcal{P}_1) \in O(\log(n))$, since it is a harmonic series. On the other hand, any cut is weight at most $2\ell$, where $\ell$ is the length of the cut, and as such $W_-(\mathcal{P}) \in O(n)$. Similarly, observe that $W_+(\mathcal{P}_3) \in O(\log(n))$, and $W_-(\mathcal{P}_4) \in O(n)$, since between any two even positions $j \in [n]$ for which we have a \texttt{(}, we cannot have a cut through the tree longer than the length of the interval, and so everything sums to $n$.

        Finally, we let $\mathcal{P} = \lnot(\mathcal{P}_1 \lor \mathcal{P}_2 \lor \mathcal{P}_3 \lor \mathcal{P}_4)$. Then $\mathcal{P}$ evaluates the Dyck language with depth $3$, and we obtain that
        \[W_+(\mathcal{P}) \leq \max_{j \in [4]} W_-(\mathcal{P}_j) \in O(n), \qquad \text{and} \qquad W_-(\mathcal{P}) \leq \sum_{j=1}^4 W_+(\mathcal{P}_j) \in O(\log(n)),\]
        and so $C(\mathcal{P}) \in O(\sqrt{n\log(n)})$.

        For the time complexity, we use the same divide and conquer strategy as in \cref{thm:or-psearch,thm:202} to decompose the graph using just $O(\log(n))$ recursion depth. Thus, the total overhead is polylogarithmic in the size of the Hilbert space, which is polynomial in $n$.
    \end{proof}

    Note that in \cref{fig:dyck3}, we can improve the weighting scheme on $\mathcal{P}_1$ to obtain $C(\mathcal{P}_1) \in O(\sqrt{n})$. We would not attain any asymptotic improvement for the Dyck-language recognition problem, though, so we leave it as an exercise for the reader.

    The natural follow-up question is whether it's possible to generalize this approach to general $k$. This requires generalizing \cref{lem:dyck3-properties} for general $k$, which seems quite challenging. We leave this for future work.

    \subsubsection{The increasing subsequence problem}

    Finally, we turn to the increasing subsequence problem. We start by formally defining it.

    \begin{definition}[$k$-increasing subsequence problem]
        Let $k,n \in \N$, $\Sigma$ be a totally ordered alphabet, and $x \in \Sigma^n$. A $k$-increasing subsequence for $x$ is a $k$-tuple $p \in [n]^k$ such that for all $i,j \in [k]$, $i < j$ implies $p_i < p_j$ and $x_{p_i} < x_{p_j}$. The extent of $p$ is $p_k-p_1$. The $k$-increasing subsequence problem asks whether a $k$-increasing subsequence exists in an instance $x \in \Sigma^n$, given query access to $x$.
    \end{definition}

    In \cite[Theorem~6]{childs2022quantum}, a quantum algorithm with $O(\sqrt{n}\log^{3(k-1)/2}(n))$ queries for solving the $k$-increasing subsequence problem was presented. Here, we improve on this complexity for the first few values of $k$.

    The problem is trivial for $k = 1$. For $k = 2$, we observe that suffices to check whether the input list is decreasing. We can do this by merely checking if there exists any $i \in [n]$ for which $x_i < x_{i+1}$. Using Grover's search algorithm, this costs $O(\sqrt{n})$ queries.

    For $k = 3$, we give a quantum algorithm that uses $O(\sqrt{n\log(n)})$ queries. To that end, we start by making an observation about the structure of $3$-increasing subsequences.

    \begin{lemma}[Properties of $3$-increasing subsequences]
        \label{lem:3is-properties}
        Let $\Sigma$ be a totally ordered alphabet, and $x \in \Sigma^n$. Suppose that $p = (i,j,k)$ is a $3$-increasing subsequence of $x$ of minimal extent. Then, $x_i < x_{i+1}$, $x_{k-1} < x_k$, and for all $\ell \in [i+1,k-2]$, $x_{\ell} \geq x_{\ell+1}$.
    \end{lemma}

    \begin{proof}
        We give a proof by contradiction. Suppose that $x_i \geq x_{i+1}$. Then, $j > i+1$, and so $(i+1,j,k)$ would be a $3$-increasing subsequence of $x$, contradicting the minimality of $p$. Similarly, suppose that $x_{k-1} \geq x_k$. Then, $j < k-1$, and $(i,j,k-1)$ would be a $3$-increasing subsequence of $x$, contradicting the minimality of $p$.

        Finally, suppose that there exists an $\ell \in [i+1,k-2]$, for which $x_{\ell} < x_{\ell+1}$. Now, we have
        \begin{align*}
            j \leq \ell \land x_{\ell+1} < x_k & \Rightarrow (\ell,\ell+1,k) \text{ is a } 3\text{-increasing subsequence of } x \\
            j \leq \ell \land x_{\ell+1} \geq x_k & \Rightarrow (i,j,\ell+1) \text{ is a } 3\text{-increasing subsequence of } x \\
            j > \ell \land x_{\ell} \leq x_i & \Rightarrow (\ell,j,k) \text{ is a } 3\text{-increasing subsequence of } x \\
            j > \ell \land x_{\ell} > x_i & \Rightarrow (i,\ell,\ell+1) \text{ is a } 3\text{-increasing subsequence of } x,
        \end{align*}
        each of which contradicts the minimality of the extent of $p$.
    \end{proof}

    We observe that in order to solve the $3$-increasing subsequence problem, it suffices to search for a $3$-increasing subsequence of minimal extent. We can then use \cref{lem:3is-properties} to design a graph composition that detects whether such minimal-extent $3$-increasing subsequences exist. This leads to the following theorem.

    \begin{theorem}
        \label{thm:3-is}
        There is a quantum algorithm that solves the $3$-increasing subsequence problem using $O(\sqrt{n\log(n)})$ queries, and running in time $\widetilde{O}(\sqrt{n})$.
    \end{theorem}

    \begin{proof}
        For every $j,k \in [n]$, we create the trivial span program $\mathcal{P} := [x_j < x_k]$ that computes whether $x_j < x_k$. Since this can be computed exactly, i.e., without error, with $\Theta(1)$ queries and time, the span program operations can be implemented using $\Theta(1)$ queries and time as well.

        Next, we generate the graph composition. For every $i \in [n]$, we attach the edge $[x_i < x_{i+1}]$ to the root node $s$. Below the resulting leaf, we create a sequence of edges $\frac{1}{\ell}[x_{i+\ell} \geq x_{i+\ell+1}]$, where we let $\ell$ run from $1$ until $n-i-1$. We attach outgoing edges labeled by the negation, i.e., $[x_{i+\ell} < x_{i+\ell+1}]$, creating new leaves. Next, we attach these to $t$ through a parallel composition of checks $[x_i < x_j]$ and $[x_j < x_{i+\ell+1}]$, where we let $j$ run from $i+1$ to $i+\ell$. We refer to the resulting graph-composed span program by $\mathcal{P}$. See also the pictorial representation in \cref{fig:3is}.

        \begin{figure}[!ht]
            \centering\scriptsize
            \begin{tikzpicture}[vertex/.style = {draw, rounded corners = .3em}, scale = 1.5]
                \begin{scope}
                    \node[vertex] (s) at (0,0) {$s$};
                    \node[left] at (s.west) {$\mathcal{P}_{j,k} = \;\;$};
                    \node[vertex] (1) at (1.5,.75) {};
                    \node[vertex] (2) at (1.5,.375) {};
                    \node[vertex] (3) at (1.5,0) {};
                    \node[vertex] (4) at (1.5,-.375) {};
                    \node[vertex] (5) at (1.5,-.75) {};
                    \node[vertex] (t) at (3,0) {$t$};

                    \draw (s) to node[above, rotate = {atan(1/2)}] {$[x_j < x_{j+1}]$} (1);
                    \draw (1) to node[above, rotate = {-atan(1/2)}] {$[x_{j+1} < x_k]$} (t);

                    \draw (s) to node[above] {$[x_j < x_\ell]$} (3);
                    \draw (3) to node[above] {$[x_\ell < x_k]$} (t);

                    \draw (s) to node[below, rotate = {-atan(1/2)}] {$[x_j < x_{k-1}]$} (5);
                    \draw (5) to node[below, rotate = {atan(1/2)}] {$[x_{k-1} < x_k]$} (t);

                    \draw[dotted] (s) to (2) to (t);
                    \draw[dotted] (s) to (4) to (t);
                \end{scope}
                \begin{scope}[shift = {(-1,-1.5)}]
                    \node[vertex] (s) at (0,0) {$s$};
                    \node[left] at (s.west) {$\mathcal{P}_j = \;\;$};
                    \node[vertex] (1) at (1.5,0) {};
                    \node[vertex] (2) at (3,0) {};
                    \node[vertex] (3) at (4.5,0) {};
                    \node[vertex] (4) at (6,0) {};
                    \node[vertex] (5) at (7.5,0) {};
                    \node[vertex] (6) at (9,0) {};
                    \node[vertex] (10) at (1.5,-.5) {};
                    \node[vertex] (20) at (3,-.5) {};
                    \node[vertex] (30) at (4.5,-.5) {};
                    \node[vertex] (40) at (6,-.5) {};
                    \node[vertex] (50) at (7.5,-.5) {};
                    \node[vertex] (60) at (9,-.5) {};
                    \node[vertex] (t) at (4.5,-1) {$t$};

                    \draw (s) to node[above] {$[x_j < x_{j+1}]$} (1);
                    \draw (1) to node[above] {$[x_{j+1} \geq x_{j+2}]$} (2);
                    \draw (2) to node[above] {$\frac12[x_{j+2} \geq x_{j+3}]$} (3);
                    \draw (3) to node[above] {$\frac13[x_{j+3} \geq x_{j+4}]$} (4);
                    \draw[dotted] (4) to (5);
                    \draw (5) to node[above] {$\frac{1}{n-2-j}[x_{n-2} \geq x_{n-1}]$} (6);

                    \draw (1) to node[left] {$[x_{j+1} < x_{j+2}]$} (10);
                    \draw (2) to node[left] {$[x_{j+2} < x_{j+3}]$} (20);
                    \draw (3) to node[left] {$[x_{j+3} < x_{j+4}]$} (30);
                    \draw[dotted] (4) to (40);
                    \draw[dotted] (5) to (50);
                    \draw (6) to node[left] {$[x_{n-1} < x_n]$} (60);

                    \draw (10) to[bend right = 20] node[near start, below left, rotate = -20] {$\mathcal{P}_{j,j+2}$} (t);
                    \draw (20) to[bend right = 15] node[near start, below, rotate = -20] {$\mathcal{P}_{j,j+3}$} (t);
                    \draw (30) to node[near start, left] {$\mathcal{P}_{j,j+4}$} (t);
                    \draw[dotted] (40) to[bend left = 10] (t);
                    \draw[dotted] (50) to[bend left = 15] (t);
                    \draw (60) to[bend left = 20] node[near start, below, rotate = 20] {$\mathcal{P}_{j,n}$} (t);
                \end{scope}
                \begin{scope}[shift={(5,0)}]
                    \node[vertex] (s) at (0,0) {$s$};
                    \node[left] at (s.west) {$\mathcal{P} = \;\;$};
                    \node[vertex] (t) at (2,0) {$t$};

                    \draw (s) to[bend left = 80] node[above] {$\mathcal{P}_1$} (t);
                    \draw[dotted] (s) to[bend left = 40] (t);
                    \draw (s) to node[above] {$\mathcal{P}_j$} (t);
                    \draw[dotted] (s) to[bend right = 40] (t);
                    \draw (s) to[bend right = 80] node[above] {$\mathcal{P}_{n-2}$} (t);
                \end{scope}
            \end{tikzpicture}
            \caption{The graph composition construction for the $3$-increasing subsequence problem. In the upper-left part, $\ell$ runs from $j+1$ to $k-1$, and in the upper-right part, $j$ runs from $1$ to $n-2$.}
            \label{fig:3is}
        \end{figure}

        Now, if there is a $3$-increasing subsequence in an input $x \in \Sigma^n$, there is also one with minimal extent, which we denote by $(i,j,k) \subseteq [n]$. From \cref{lem:3is-properties}, we observe that in that case, $s$ and $t$ are indeed connected, and by adding up the resistances along its path, we obtain that the effective resistance is at most $w_+(x,\mathcal{P}) \in O(\log(k-i)) \subseteq O(\log(n))$.

        On the other hand, if there is no $3$-increasing subsequence in $x$, then we find a cut in the graph. Indeed, for all $i$ for which $x_i \geq x_{i+1}$, we simply cut through the edges that are directly adjacent to $s$, which gives us a total resistance of at most $O(n)$. Next, if $x_i < x_{i+1}$, we find the smallest $j > i$ such that $x_j < x_{j+1}$. We observe that we can cut through the sequence of edges in the tree, with cost $j-i$, and we can also cut the $j-i$ edges of cost $1$ that come out of it. Finally, since there is no $3$-increasing subsequence, we can also cut through the last OR-tree, which also costs $j-i$. Since summing $j-i$ over all such pairs $(i,j)$ is at most $n$, we obtain that the negative witness size is $w_-(x,\mathcal{P}) \in O(n)$.

        We conclude the proof by observing that
        \[C(\mathcal{P}) = \sqrt{W_-(\mathcal{P}) \cdot W_+(\mathcal{P})} \in O\left(\sqrt{n\log(n)}\right).\]

        For the time complexity, we use the same decomposition technique as in \cref{thm:or-psearch,thm:202,thm:dyck-3}, to argue that we can decompose the graph in just $O(\log(n))$ levels of recursion. This then generates an overhead that is polylogarithmic in the total number of edges in the graph, which is polynomial in $n$.
    \end{proof}

    We suspect that this construction can be generalized to $k$-increasing subsequences, where $k \in \Theta(1)$. However, finding a succinct characterization of a minimal-extent $4$-increasing subsequence, like in \cref{lem:3is-properties}, seems to be rather cumbersome. We leave this generalization for future work.

    \section*{Acknowledgements}

    I would like to thank Galina Pass, Stacey Jeffery, Sebastian Zur, D\'aniel Szab\'o, Roman Edenhofer, Amin Shiraz Gilani, Subhasree Patro, Nikhil S.\ Mande, Francisca Vasconcelos, Stephen Piddock, Sander Gribling and Simon Apers for many fruitful discussions. I would also like to acknowledge Yixin Shen for directing my attention to the pattern matching problem as a potential application to this framework. Finally, I would like to thank my PhD committee for reading and commenting on earlier versions of the graph composition framework. I am supported by a Simons-CIQC postdoctoral fellowship through NSF QLCI Grant No.\ 2016245.

    \bibliographystyle{alpha}
    \bibliography{references}
\end{document}